\documentclass[12pt, journal,onecolumn]{IEEEtran}
\IEEEoverridecommandlockouts
\usepackage{cite}
\usepackage{amsmath,amssymb,amsfonts}
\usepackage[noend]{algorithmic}
\usepackage{graphicx}
\usepackage{textcomp}
\usepackage{xcolor}
\usepackage{amssymb}
\usepackage{indentfirst}
\usepackage{amsmath}
\usepackage{autobreak}
\usepackage{bm}
\usepackage{booktabs}
\usepackage[font=small]{caption}
\usepackage[subrefformat=parens]{subcaption}
\usepackage{setspace}

\usepackage{tikz}
\usepackage{pgfplots}
\usetikzlibrary{plotmarks}
\usetikzlibrary{matrix,positioning,shapes,arrows,decorations,calc,fit}
\usetikzlibrary{decorations.pathreplacing}
\usetikzlibrary{spy,backgrounds}
\usetikzlibrary{external}
\usetikzlibrary{patterns}
\usetikzlibrary{shapes,arrows}
\usetikzlibrary{shadows.blur}
\usetikzlibrary{shapes.symbols}
\usetikzlibrary{
	pgfplots.groupplots,
	matrix
}

\usepackage{epsfig}
\usepackage{setspace}
\usepackage{algorithm}
\usepackage{changepage}
\usepackage{amsthm} 
\usepackage{soul,color}
\usepackage{hyperref}
\usepackage{lipsum}
\usepackage{cuted}
\usepackage{multirow}
\DeclareMathSizes{10}{9}{7}{6}
\IEEEaftertitletext{\vspace{-1\baselineskip}}
\newtheorem{theorem}{Theorem}
\newtheorem{lemma}{Lemma}
\newtheorem{proposition}{Proposition}
\newtheorem{corollary}{Corollary}
\newtheorem{remark}{Remark}
\newtheorem{assumption}{Assumption}
\newtheorem{example}{Example}
\theoremstyle{definition}

\theoremstyle{remark}

\IEEEaftertitletext{\vspace{-2\baselineskip}}

\begin{document}

\setlength{\abovedisplayskip}{2.5pt}
\setlength{\belowdisplayskip}{2.5pt}

\title{Density Evolution Analysis of the Iterative Joint Ordered-Statistics Decoding for NOMA}

\author{Chentao Yue,~\IEEEmembership{Member,~IEEE,}
		Mahyar Shirvanimoghaddam,~\IEEEmembership{Senior Member,~IEEE,}\\
		Alva Kosasih,~\IEEEmembership{Student Member,~IEEE,} 
		Giyoon Park, Ok-Sun Park, \\
		Wibowo Hardjawana,~\IEEEmembership{Member,~IEEE,}
		Branka Vucetic,~\IEEEmembership{Life Fellow,~IEEE,}\\
		and Yonghui Li,~\IEEEmembership{Fellow,~IEEE}
		
		\thanks{C. Yue, M. Shirvanimoghaddam, A. Kosasih, W. Hardjawana, B. Vucetic, and Y. Li are with the School of Electrical and Information Engineering, the University of Sydney, NSW 2006, Australia (email:\{chentao.yue, mahyar.shm, alva.kosasih, wibowo.hardjawana, branka.vucetic, yonghui.li\}@sydney.edu.au).}
		\thanks{G. Park and O. Park are with the Electronics and Telecommunications Research Institute, Daejeon, South Korea (email: gypark@etri.re.kr; ospark@etri.re.kr).}
}

\maketitle

\begin{abstract}
 In this paper, we develop a density evolution (DE) framework for analyzing the iterative joint decoding (JD) for non-orthogonal multiple access (NOMA) systems, where the ordered-statistics decoding (OSD) is applied to decode short block codes. We first investigate the density-transform feature of the soft-output OSD (SOSD), by deriving the density of the extrinsic log-likelihood ratio (LLR) with known densities of the priori LLR. Then, we represent the OSD-based JD by bipartite graphs (BGs), and develop the DE framework by characterizing the density-transform features of nodes over the BG. Numerical examples show that the proposed DE framework accurately tracks the evolution of LLRs during the iterative decoding, especially at moderate-to-high SNRs. Based on the DE framework, we further analyze the BER performance of the OSD-based JD, and the convergence points of the two-user and equal-power systems.

\end{abstract}


\IEEEpeerreviewmaketitle

\vspace{-0.5em}
\section{Introduction}
\vspace{-0.5em}

   Ultra-reliable and low-latency communications (URLLC) has attracted particular attention in 5G and the upcoming 6G for mission-critical services \cite{Mahyar2019ShortCode,popovski2019wireless,dogra2020survey}. The key performance requirements of URLLC are the hundreds-of-microsecond time-to-transmit latency, block error rate (BLER) of $10^{-5}$, and the bit-level granularity of the codeword size and code rate. These ultra-low latency requirements necessitate a high processing speed in the physic layer\cite{shao2019survey}, and mandate the use of short block-length codes ($\leq150$ bits) \cite{Mahyar2019ShortCode}. The mission-critical services in URLLC also raise the challenge of achieving scalable and reliable connectivity to accommodate massive users and devices with limited channel spectrum resources \cite{bockelmann2016massive}. Non-orthogonal multiple access (NOMA) has recently gained increased attention as a promising technique for achieving superior spectral efficiency \cite{makki2020survey, dai2018survey, wong2017key}, which has been considered as a promising technique for 5G and beyond. It enables users to transmit signals non-orthogonally in the frequency, time, or code domains. 
   
   In the asymptotic large block-length scenario, NOMA can achieve certain corner points of the multiple access channel (MAC) capacity region using successive interference cancellation (SIC) \cite{wang2019near}, and can achieve arbitrary points on the boundary of the MAC capacity region with rate splitting and joint decoding (JD) techniques \cite{kramer2008topics, hou2006implementing}. However, in URLLC applications, SIC is inferior due to its sequential nature; the decoding order will critically affect the transmission latency and reliability of users accommodated. For instance, the last decoded user will experience the largest latency, while the early decoded users will be subject to severe multiple-access interference (MAI). Therefore, when providing URLLC services, NOMA should avoid using SIC and instead, use efficient JD schemes with short block-length codes.
    
    The design of low-complexity JD schemes with near-optimal performance is challenging due to MAI \cite{al2014uplink}. For uncoded NOMA systems, the maximum-likelihood (ML) detection of superposed signals is an NP-hard problem with the complexity of $\mathcal{O}(|\mathcal{S}|^{n_u})$ \cite{liu2019capacity}, where $|\mathcal{S}|$ is the cardinality of symbol set, $n_u$ is the number of users, and $\mathcal{O}(\cdot)$ is the big-O operand. In coded systems, the ML JD has the complexity as high as $\mathcal{O}(|\mathcal{S}|^{n_u}2^k)$, where $k$ is the information block length of the coding scheme applied. In the asymptotically large block-length scenario, many JD schemes have been devised for NOMA \cite{liu2019capacity,ping2004approaching,wang2019near,ebada2020iterative,zhang2017channel,gao2014low,montanari2005belief,ping2003interleave,balatsoukas2018design,sharifi2015ldpc}. They typically combines a multi-user detector (MUD) with  a-posterior probability  (APP) decoder to perform iterative decoding. The MUD can be configured to use either the parallel interference cancellation (PIC) \cite{marinkovic2001space} or the (linear) minimum mean square error (MMSE) \cite{madhow1994mmse} technique \cite{liu2019capacity,ping2004approaching,wang2019near,ebada2020iterative,ping2003interleave,gao2014low}. It has been proved that this structure with MUD and APP decoder can approach the boundary of the MAC capacity region in asymptotically large block-length scenarios \cite{boutros2002iterative,ping2003interleave,ping2004approaching}. In relation to the APP decoder, low density parity check (LDPC) codes were analyzed and optimized in \cite{wang2019near,sharifi2015ldpc,balatsoukas2018design}. Owning to their bipartite graph (BG) representation, LDPC codes were intelligently constructed to match with the MUD by using the extrinsic-information-transform (EXIT) chart analysis \cite{li2005exit, li2007analysis,wang2019near}. Moreover, moderate-to-long polar codes and successive cancellation list decoding  (SCL) were also investigated in \cite{ebada2020iterative,xiang2021iterative} to be concatenated with MUD. Apart from the large block-length scenario, \cite{sun2018short} studied the optimal rate and power allocation for short-packet NOMA. In \cite{schiessl2020noma}, the short-packet transmission latency was studied with imperfect channel state information (CSI). Nevertheless, designing an efficient JD for NOMA in the short block-length regime is rarely attempted.
    
    Recently, an efficient JD scheme based on ordered-statistics decoding (OSD) was proposed for the short-block NOMA system towards URLLC applications \cite{yue2021noma}. This OSD-based JD was designed by considering the decoding complexity and error-correction capability of short block codes. Specifically, only short codes that approach the Normal Approximation (NA) bound (Polyanskiy \emph{et. al.} \cite{PPV2010l}) were presumed. Generally, codes approaching NA have large minimum Hamming distances and high-density generator matrices, e.g., Bose-Chaudhuri-Hocquenghem (BCH) codes \cite{Mahyar2019ShortCode,liva2016codeSurvey}; however, their near-ML decoding is always challenging. As a universal decoder for block codes, OSD \cite{Fossorier1995OSD} has rekindled the interests and been applied in decoding NA-approaching short codes recently \cite{dhakal2016error,NewOSD-5GNR, yue2021linear, yue2021probability}. For a linear block code $\mathcal{C}(n,k)$ with minimum distance $d_{\mathrm{H}}$, it has been proved that an OSD with the order of $m = \lceil d_{\mathrm{H}}/4-1\rceil$ is asymptotically approaching the ML decoding performance\cite{Fossorier1995OSD}. In OSD, a higher decoding order indicates improved BLER performance at the expense of increased decoding complexity.
    
    The OSD-based JD \cite{yue2021noma} iteratively performs PIC and the soft-output OSD (SOSD). PIC conducts interference cancellation and outputs the estimated log-likelihood ratios (LLR) of symbols per user, which serves as the priori information fed to SOSD. Then SOSD outputs the extrinsic LLR, which will be the input of PIC for the next decoding iteration. Two novel techniques, decoding switch (DS) and decoding combiner (DC), are further used to accelerate the convergence speed \cite{yue2021noma}. DS controls the engagement of SOSD; when DS is turned off, the SOSD decoding is skipped. Moreover, DC adaptively combines the priori LLR and the extrinsic LLR according to a parameter representing the decoding quality. Even simple DS and DC strategies were shown to significantly reduce the number of iterations required to to reach convergence \cite{yue2021noma}. For instance, when SOSD is simply skipped for $n_u$ iterations before the engagement, the OSD-based JD requires a fewer number of decoding iterations than that of SIC \cite{yue2021noma}, indicating a significant reduction in overall JD complexity.

    Despite that the OSD-based JD is efficient for the NOMA system with short block codes, it lacks appropriate analysis, hindering further optimization. Specifically, if the distribution (i.e., density) of the LLRs propagated between PIC and SOSD can be characterized, then optimized strategies of DS and DC can be designed. In addition, if the relationship among the convergence point of JD, the coding schemes, and the decoding order of OSD can be determined, one can carefully select the optimal decoding order of OSD per iteration to further reduce the JD complexity. In previous work, many techniques have been developed for analyzing the iterative receiver for multiple-input-multiple-output (MIMO) and NOMA systems, including the variance-transform or signal-interference-to-noise-ratio-transform (SINR-transform) chart \cite{liu2019capacity,yuan2014energy,yuan2008evolution}, and the EXIT-chart analysis \cite{chi2018practical, liu2019capacity}; they are, however, insufficient for analyzing the OSD-based JD. Precisely, these chart-based analytical approaches generally assume Gaussian MAI, which is achieved either through applying multiple receiving antennas (MIMO-NOMA) or by accommodating a large number of users with similar or identical receiving power. As a result, they fall short of analyzing power-domain NOMA systems with limited numbers of users, for which MAI normally does not follow a Gaussian distribution. On the other hand, the EXIT-chart analysis can readily utilize the mutual-information-transform feature well developed for LDPC codes; nevertheless, for other codes, Monte Carlo method is usually applied.
    
    To provide tools for further optimizing the OSD-based JD for short-block NOMA, we develop an analytical framework based on density evolution (DE) in this paper, which allows us to track the evolution of the densities of LLRs propagated in the iterative JD at each iteration. The main contributions of this work are summarized below.
    
    \begin{itemize}
    \item We analyze the density-transform feature of SOSD; that is, we derive the distribution of the extrinsic LLRs output by SOSD with given priori LLRs, by introducing a variant algorithm, Dual-OSD, to simplify the analysis. Dual-OSD has two phases of decoding process, namely the phase-0 and phase-1 reprocessing. It is shown that by carefully selecting the parameters in the phase-0 and phase-1 reprocessing, Dual-OSD can approach the density-transform feature of SOSD or its variants. Numerical results and simulation results for short extended BCH (eBCH) codes validate the derived density transform feature.
    
    \item Based on the density-transform feature of SOSD, we develop the DE framework for the OSD-based JD. We propose a numerical approach for determining the density of LLRs output by PIC, when DS is turned off (SOSD is skipped). On the other hand, when DS is turned on, we develop the DE technique by examining the density-transform feature of the PIC and the SOSD decoder separately. Numerical results for two-user and three-user systems verify that the proposed DE can accurately describe the evolution of the priori and extrinsic LLRs during the JD process.
    
    \item We further provide some analytical examples based on the DE framework, including the bit error rate (BER) and the convergence point of the considered JD. First, it is demonstrated that the DE framework can be used to theoretically determine the BER performance of JD per iteration and user. Furthermore, the convergence condition and the converged LLR densities of the two-user system and the equal-power system are discussed. Numerical results show that the converged LLR densities can be directly determined by the proposed DE technique. Finally, we outline some future works based on the proposed analytical framework.

    \end{itemize}

        The rest of this paper is organized as follows. Section \ref{sec::Preliminaries} describes the preliminaries and reviews the OSD-based JD. Section \ref{Sec::DTF-OSD} studies the density-transform feature of SOSD. The DE framework is developed in Section \ref{Sec::DE}. Section \ref{Sec::Example} provides analytical examples regarding BER and convergence conditions. Finally, Section \ref{sec::Conclusion} concludes the paper.

    \subsubsection*{Notation} In this paper, we use $\mathrm{Pr}(\cdot)$ to denote the probability of an event. A bold letter, e.g., $\mathbf A$, represents a matrix, and a lowercase bold letter, e.g., $\mathbf{a}$, denotes a row vector. We also use $[a]_u^v$ to denote a row vector containing element $a_{\ell}$ for $u\le \ell\le v$, i.e., $[a]_u^v = [a_u,\ldots,a_v]$. We use a calligraphic uppercase letter to denote a set, e.g., $\mathcal{S}$, or the density of a random variable, also known as the probability density function ($\mathrm{pdf}$). For example, $\mathcal{V}(x)$ is the density of the random variable $V$. Furthermore, $\mathcal{V}(x;Y)$ denote the density of $V$ parameterized by a variable $Y$, and $\mathcal{V}(x|Y=y)$ is the density conditioning on $\{Y=y\}$. In particular, we simply denote the Gaussian density as $f(x) = \mathcal{N}(\mu,\sigma^2)(x)$ with the mean $\mu$ and variance $\sigma^2$. Other notations will be specifically stated.

\vspace{-0.5em}
\section{Preliminaries} \label{sec::Preliminaries}
\vspace{-0.5em}

\subsection{System Model}
\vspace{-0.5em}

    We consider a binary phase shift keying (BPSK) signal transmission in the uplink power-domain NOMA with $n_u$ simultaneous users, where the channel remains constant for the duration of one code block and changes independently between blocks. Given the block code ${\mathcal C}(n,k)$, all users share the identical codebook represented by the generator matrix $\mathbf{G}$. The information block of user $u$, $\mathbf{b}^{(u)} = [b^{(u)}]_1^k$, is encoded to its codeword, $\mathbf{c}^{(u)} = [c^{(u)}]_1^n$, with $\mathbf{c}^{(u)} = \mathbf{b}^{(u)}\mathbf{G}$, where $k$ and $n$ denote the information block and codeword lengths, respectively. The codeword $\mathbf{c}^{(u)}$ is interleaved by an interleaver $\Pi_u$, which is randomly selected for each user. All users simultaneously transmit the modulated symbol to the base station non-orthogonally. At the base station, the superposed signal $\mathbf{r}$ is received with a single antenna, i.e.,
        \begin{equation}\small \label{equ::Pri::Sysmod}
            \mathbf{r} = \mathbf{h}\mathbf{X} + \mathbf{w},
        \end{equation}
    where $\mathbf{h} = [h^{(1)},\ldots,h^{(n_u)}]$ is a $1\times n_u$ channel coefficient vector. For the simplicity of analysis, we consider the additive-white-Gaussian-noise (AWGN) channel in this paper. Specifically, coefficient $h^{(u)} \in \mathbb{R}$, $1\leq u \leq n_u$, and thus $(h^{(u)})^2$ represents the strength of receiving power of user $u$.  One can extend the results in this paper to fading channels by taking known distributions of complex $\mathbf{h}$ into account. $\mathbf{X} = [\mathbf{x}^{(1)};\mathbf{x}^{(2)};\ldots;\mathbf{x}^{(n_u)}]$ is a $n_u\times n$ matrix of modulated symbols, where $\mathbf{x}^{(u)}$ is the symbol vector of $\mathbf{c}^{(u)}$, i.e., $x_{i}^{(u)} = 1 - 2c_{i}^{(u)}$ for $1\leq i\leq n$. $\mathbf{w} = [w]_1^n$ is the vector of AWGN variables, where each entry $w_i \sim \mathcal{N}(0,\sigma^2)$ with $\sigma^2$ being the noise power. At the receiver, we assume that the channel coefficients are known a priori, and define the multi-user SNR as $\mathrm{SNR} = \sum_{u=1}^{n_u}\frac{1}{\sigma^2}(h^{(u)})^2$. 

\vspace{-0.5em}
\subsection{OSD Algorithm} \label{sec::Pri::OSD}
  \vspace{-0.5em}    
  
        We briefly introduce the OSD algorithm and its soft-output variant \cite{Fossorier1995OSD}. When decoding, a sequence of LLR $\bm{\ell} = [\ell]_1^n$ of a transmitted codeword $\mathbf{c} = [c]_1^n$ is input to OSD, defined as ${\ell}_{i} \triangleq \log \frac{\mathrm{Pr}(c_{i}=1|\bar{x}_i)}{\mathrm{Pr}(c_{i}=0|\bar{x}_i)}$ conditioning on an observation $\bar{x}_i$ of $c_i$. The observation $\bar{x}_i$ can be obtained from the channel or fed by the MUD. Then, OSD outputs the codeword estimate $\mathbf{c}_{\mathrm{op}}$ of $\mathbf{c}$ by decoding $\bm{\ell}$.

        At the beginning of the OSD, the bit-wise hard decision vector $\mathbf{y}= [y]_{1}^n$ is obtained according to the following rule: $y_{i}=1 $ for $\ell_{i}<0$, and $y_{i}=0 $ for $\ell_{i}\geq 0$. We define the magnitude of LLR $\ell_{i}$ as the reliability of $y_{i}$, denoted by $\alpha_{i}$, i.e., $\alpha_{i} = |\ell_{i}|$, where $|\cdot|$ is the absolute operation. 
        
        Then, a permutation $\pi_1$ is performed to sort the LLR $\bm{\ell}$ and corresponding columns of $\mathbf{G}$ in descending order of reliabilities $\bm{\alpha}$, which obtains $\pi_{1}(\bm\alpha)$ and $\pi_1(\mathbf{G})$. Next, OSD performs Gaussian elimination (GE) over $\pi_1(\mathbf{G})$ to derive the systematic matrix $\mathbf{\widetilde G} = [\mathbf{I}_k \  \mathbf{\widetilde{P}}]$, where $\mathbf{I}_k$ is a $k\times k$ identity matrix and $\mathbf{\widetilde{P}}$ is the parity sub-matrix. A permutation $\pi_{2}$ may be performed during GE to ensure that the first $k$ columns of $\mathbf{\widetilde G}$ are linearly independent. Finally, vectors $\mathbf{y}$, $\bm{\ell}$, and $\bm{\alpha}$ are respectively permuted to $\widetilde{\mathbf{y}} = \pi_{2}(\pi_{1}(\mathbf{y}))$, $\widetilde{\bm{\ell}} = \pi_{2}(\pi_{1}(\bm{\ell}))$, and $\bm{\widetilde \alpha} = \pi_{2}(\pi_{1}(\bm\alpha))$, corresponding to the columns of $\widetilde{\mathbf{G}}$. As shown by \cite[Eq. (59)]{Fossorier1995OSD}, $\pi_{2}$ is usually negligible and can be omitted, i.e., taking $\mathbf{a} = \pi_{2}(\mathbf{a})$ for an arbitrary length-$n$ vector $\mathbf{a}$.
    
        After the above permutations, the first $k$ positions of $\mathbf{\widetilde y}$, denoted by $\widetilde{\mathbf{y}}_{\mathrm{B}} =[\widetilde y]_1^k$, are referred to as the most reliable basis (MRB) \cite{Fossorier1995OSD}. A test error pattern (TEP) $\mathbf{e} = [e]_1^k$ is added to $\widetilde{\mathbf{y}}_{\mathrm{B}}$ to obtain a codeword estimate by re-encoding, i.e., $\mathbf{\widetilde c}_{\mathbf{e}} = \left(\widetilde{\mathbf{y}}_{\mathrm{B}}\oplus \mathbf{e}\right)\mathbf{\widetilde G}$
    	where $\mathbf{\widetilde c}_{\mathbf{e}}$ is the codeword estimate with respect to TEP $\mathbf{e}$. In OSD, a number of TEPs are re-encoded, which is referred to as reprocessing \cite{Fossorier1995OSD}. A general reprocessing strategy is starting from TEPs with zero Hamming weight, and increasing the weight until the maximum allowed weight is reached \cite{Fossorier1995OSD}. The maximum Hamming weight is known as the decoding order of OSD. It is proved that an OSD with order $m = \lceil d_{\mathrm{H}}/4-1\rceil$ is asymptotically approaching the ML decoding performance for a code with the minimum Hamming distance $d_{\mathrm{H}}$ \cite{Fossorier1995OSD}. Because of the overhead of multiplications in each re-encoding, the overall complexity of OSD is dominated by the number of TEPs.
    	
    	With BPSK modulation, OSD finds the best ordered codeword estimation $\widetilde{\mathbf{c}}_{\mathrm{op}}$ by minimizing the weighted Hamming distance (WHD) between each estimate $\widetilde{\mathbf{c}}_{\mathbf{e}}$ and $\mathbf{\widetilde{y}}$, which is defined as \cite{valembois2002comparison}
    	\begin{equation}\small \label{equ::WHDdefine}
        	d(\mathbf{\widetilde c}_{\mathbf{e}},\mathbf{\widetilde y}) \triangleq \sum_{\substack{1 \leq i \leq n \\ \widetilde{c}_{\mathbf{e},i}\neq \widetilde{y}_{i}}} \widetilde\alpha_{i}. 
    	\end{equation}
        Finally, the estimate $\mathbf{c}_{\mathrm{op}}$ of $\mathbf{c}$ is output by performing inverse permutations over $\mathbf{\widetilde c}_{\mathrm{op}}$, i.e.
    	$\mathbf{c}_{\mathrm{op}} = \pi_{1}^{-1}(\pi_{2}^{-1}(\mathbf{\widetilde c}_{\mathrm{op}}))$. Throughout this paper, we imply the following relationships $\widetilde{\mathbf{a}} = \pi_{2}(\pi_{1}(\mathbf{a}))$ and $\mathbf{a} = \pi_{1}^{-1}(\pi_{2}^{-1}(\widetilde{\mathbf{a}}))$ for an arbitrary length-$n$ vector $\mathbf{a}$.
    
    	OSD was modified to output soft information in \cite{fossorier1998soft}, known as SOSD. Given the input LLR sequence $\bm{\ell} = [\ell]_1^n$, the extrinsic LLR of the $i$-th input LLR, $\ell_i$, is derived by SOSD as 
        \begin{equation}\small \label{equ::LowSISO::extLLR}
        \begin{split}
            \delta_{i} = \sum_{j=1}^{n}\ell_j\left(c_j(i\!:\!0) - c_j(i\!:\!1)\right) - \ell_i ,
        \end{split}
        \end{equation}
        where $\mathbf{c}(i\!:\! 0) = [c(i\!:\! 0)]_1^n$ and $\mathbf{c}(i\!:\! 1)=[c(i\!:\! 1)]_1^n$ are codeword estimates generated by OSD whose $i$-th bit is 0 and 1, respectively, with the lowest WHD to $\mathbf{y}$. It is shown that SOSD \cite{fossorier1998soft} with order $m = \lceil d_{\mathrm{H}}/4\rceil$ almost always delivers the same soft outputs as the Max-Log-MAP algorithm  \cite{hagenauer1996iterative}.

         \begin{figure*} 
    		\begin{center}
    			\includegraphics[scale=0.55] {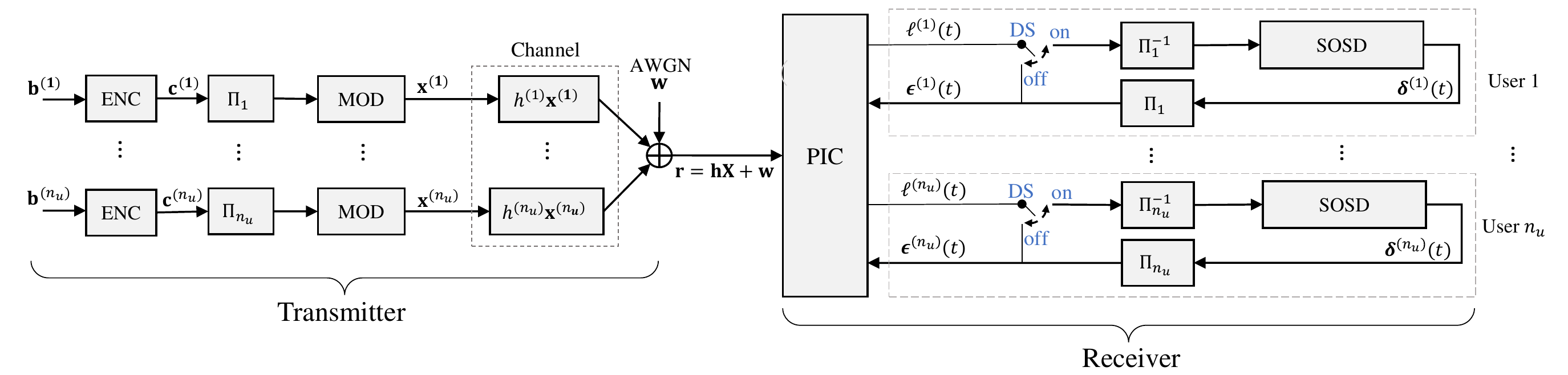}
    			\caption{The structure of the transmitter and receiver}
    			\label{Fig::Structure}
    		\end{center}
    	\end{figure*}
    	
    \vspace{-0.5em}
        \subsection{NOMA Joint Decoding with SOSD}
        \vspace{-0.5em}
        
        The signal $\mathbf{r}$ given by (\ref{equ::Pri::Sysmod}) is decoded by an iterative JD receiver based on SOSD \cite{yue2021noma}, as shown in Fig. \ref{Fig::Structure}. For clarity of notation, we do not differentiate variables before and after interleavers. We note that inteleavers are effective in reducing the correlation between the signals of different users \cite{ping2003interleave,wang2019near}. The JD receiver has two major phases in relation to the decoding switch (DS): 1) the DS-off phase and 2) the DS-on phase. At the beginning of JD, DS is turned off (DS-off phase). PIC finds the LLRs, $\bm{\ell}^{(u)}(t) = [\ell^{(u)}(t)]_1^n$, of each user at iteration $t$. Then, $\bm{\epsilon}^{(u)}(t) \leftarrow \bm{\ell}^{(u)}(t)$ is directly fedback to PIC, serving as the priori LLRs for the next iteration. After a few iterations, DS is turned on (DS-on phase). In the DS-on phase, $\bm{\ell}^{(u)}(t)$ is input to the SOSD decoder to output the extrinsic LLRs $\bm{\delta}^{(u)}(t)$. Finally, $\bm{\epsilon}^{(u)}(t) \leftarrow \bm{\delta}^{(u)}(t)$ is fedback to PIC for the next iteration. In each iteration, the SOSD decoding is performed in parallel for all users. The process of JD is terminated when the maximum iteration $t_{\max}$ is achieved, or when the decoding results converge.
        
        \subsubsection{Parallel interference cancellation}
         PIC \cite{liu2019capacity,kosasih2021bayesian} is applied to perform MUD in the proposed OSD-based JD. Taking the procedure for user $u$ as an example, the priori information $\bm{\epsilon}^{(u)}(t-1)$ is fed to PIC at the beginning of iteration $t$, $t>1$. We initialize $\bm{\epsilon}^{(u)}(0) = 0$ for the first iteration. For the $i$-th transmitted symbol of user $u$, $x_{i}^{(u)}$, PIC estimates its mean and variance, respectively, as \cite{kosasih2021bayesian}
         \begin{equation}\small     \label{equ::Receiver::Primean}
             \mu_{i}^{(u)}(t) = \tanh\left(\frac{\epsilon_{i}^{(u)}(t-1)}{2}\right) \ \ \text{and}  \ \  \upsilon_{i}^{(u)}(t) =  1 - \left(\mu_{i}^{(u)}(t)\right)^2.
         \end{equation}
         Next, PIC performs the interference cancellation and estimates the LLR of symbol ${x}_i^{(u)}$ as \cite{wang2019near,kosasih2021bayesian}
        \begin{equation}\small  \label{equ::Receiver::PICLLR}
             \ell_{i}^{(u)}(t)  = 2\frac{\frac{1}{h^{(u)}}\left(r_i - \sum_{j\neq u}h^{(j)}\mu_{i}^{(j)}(t)\right)}{\sum_{j\neq u}\left(\frac{h^{(j)}}{h^{(u)}}\right)^2 \upsilon_{i}^{(j)}(t) + \left(\frac{1}{h^{(u)}}\right)^2\sigma^2}.
        \end{equation}
        A decision statistics combiner (DSC) can be implemented with PIC to smooth the convergence behavior\cite{marinkovic2001space,kosasih2021bayesian}. DSC generally combines $\ell_{i}^{(u)}(t)$ of adjacent iterations with a parameter $\beta$ ($0\leq\beta\leq1$) according to $\tanh(\ell_{i}^{(u)}(t)) \!\leftarrow \! \beta \tanh(\ell_{i}^{(u)}(t) ) +(1-\beta) \tanh(\ell_{i}^{(u)}(t-1))$.
        
        \subsubsection{Decoding Switch} \label{Sec::DS}
       
        DS was proposed in our previous work \cite{yue2021noma} to determine the engagement of SOSD in JD iterations. Specifically, when PIC fails to cancel the MAI properly and produces low-quality LLRs, the DS is set to ``off''. When the quality of outputs of PIC improves after a few iterations of the DS-off phase, the DS is turned ``on'' and decoding begins. In \cite{yue2021noma}, a ``simple DS'' strategy was introduced: the receiver performs $n_u$ iterations without decoding, and then turns on DS for subsequent iterations. In other words, the DS-off phase has $n_u$ iterations. It has been shown that this simple strategy can significantly reduce the number of JD iterations required to achieve convergence.
       
       \subsubsection{SOSD Decoders}
       
       If the original SOSD is employed, the output of decoder, $\bm{\delta}^{(u)}(t)$, is obtained according to (\ref{equ::LowSISO::extLLR}). Alternatively, the original SOSD can be replaced by a low-complexity SOSD (LC-SOSD) devised in \cite{yue2021noma}. In a nutshell, LC-SOSD determines $\bm{\delta}^{(u)}(t)$ in a reduced space of codewords compared to the original SOSD, which applies a  stopping condition that terminates the decoding process early. It was validated that LC-SOSD has a very similar mutual information transform property to the original SOSD with a significantly reduced complexity.  We refer interested readers to  \cite{yue2021noma} for details. Moreover, \cite{yue2021noma} showed that the input and output of SOSD can be adaptively combined according to a predefined decoding quality $\gamma$ to accelerate the convergence, known as decoding combiner (DC). Specifically, at the output of the decoder, we have 
       $\tanh(\delta_i^{(u)}(t)) \!\leftarrow \! \gamma \tanh(\delta_{i}^{(u)}(t)) +(1-\gamma) \tanh(\ell_{i}^{(u)}(t))$.
       
       \subsubsection{Outputs of the Joint Decoding}
       At iteration $t$, the decoding result of user $u$, denoted by $\hat{\mathbf{c}}^{(u)}(t)$, is given by the hard decision over the posterior LLR, i.e., $\hat{c}_i^{(u)}(t) = 0$ for $\delta_{i}^{(u)}(t) + \ell_{i}^{(u)}(t) \geq 0$, and $\hat{c}_i^{(u)}(t) = 1$, for $\delta_{i}^{(u)}(t) + \ell_{i}^{(u)}(t) < 0$. The decoding iteration is stopped when $\hat{\mathbf{c}}^{(u)}(t) = \hat{\mathbf{c}}^{(u)}(t-1)$ holds for all users $u$, $1\leq u \leq n_u$. 
       
       Aiming at developing a general analytical framework for the OSD-based JD, we consider the JD without DSC and DC. Nonetheless, with known parameters $\beta$ and $\gamma$, one can easily characterize the effect of DSC and DC based on the framework introduced in this paper.
      
     \vspace{-0.5em}   
	\section{Density-Transform Feature of SOSD} \label{Sec::DTF-OSD}
	\vspace{-0.5em}
	
    In this section, we characterize the LLR density-transform feature of SOSD. That is, we theoretically determine the density (i.e., $\mathrm{pdf}$) of $\delta_{i}$ in (\ref{equ::LowSISO::extLLR}), when the density of input LLR $\ell_i$, the decoding order $m$, and the codebook $\mathcal{C}(n,k)$ are known in priori. The density-transform feature of SOSD will be used to develop the DE framework for the iterative JD in Section \ref{Sec::DE}.
    
    We assume that $\mathcal{C}(n,k)$ is a random code in the analysis of this section for simplicity. Specifically, $\mathbf{G}$ is a randomly constructed binary matrix and the weight spectrum of $\mathcal{C}(n,k)$ follows a binomial distribution $\mathcal{B}(n,\frac{1}{2})$. A number of widely used high-density codes, e.g., BCH codes and Primitive Rateless (PR) codes \cite{Mahyar2021primitive}, etc., have near binomial weight spectrum \cite{macwilliams1977codingtheory}. The random code assumption can significantly simplify the analysis. Henceforth, we regard LLRs (e.g., $\ell_i$ and $\delta_i$) as random variables without introducing extra notations. To start, we have the following assumption.
    \begin{assumption} \label{assum::Indep_Decoder_Input}
        The inputs of SOSD, denoted by $\bm{\ell} = [\ell]_1^n$, within the same block are identically and independently distributed (i.i.d.) variables.
    \end{assumption}
    Assumption \ref{assum::Indep_Decoder_Input} naturally holds for the single-user transmissions over a memoryless channel. In the case of the iterative JD, $\bm{\ell}$ is given by the output of PIC, which can also be regarded as i.i.d. variables because of the deep and random iterleavers applied between PIC and SOSD (also refer to Assumption \ref{Assum::PICbitsindependence}). Similar assumptions were widely used in the analyses of iterative MIMO receivers \cite{wang2019near,liu2019capacity}, turbo decoders \cite{richardson2001capacity,divsalar2001iterative} and the concatenated decoding \cite{ten2001convergence}. Under Assumption \ref{assum::Indep_Decoder_Input}, we introduce a variant of SOSD, namely Dual-OSD, that enables further analysis of the density-transform feature.
    
    \vspace{-0.5em}
    \subsection{Dual-OSD}
    \vspace{-0.5em}
    Based on the input sequence of LLR $\bm{\ell} = [\ell]_1^n$, Dual-OSD first preprocesses the generator matrix $\mathbf{G}$ as follows. Let $\mathbf{g}_i$ denote the $i$-th column of $\mathbf{G}$, $1\leq i\leq n$. For a given $i$, we first swap the column $\mathbf{g}_i$ and the first column of $\mathbf{G}$, and the rest of columns of $\mathbf{G}$ are permuted in the descending order of reliabilities $\bm{\alpha}' = [\alpha_1,\ldots,\alpha_{i-1},\alpha_{i+1},\ldots,\alpha_n]$. We denote the permuted generator matrix as $\mathbf{G}_i$. Then, similar to OSD, GE is performed over $\mathbf{G}_i$ to obtain the systematic generator matrix $\mathbf{G}'$. We denote the column index permutation between $\mathbf{G}'$ and $\mathbf{G}_i$ as $\pi'$. Next, we punctured the first column, $\mathbf{g}_1'$, of $\mathbf{G}'$ and obtain a matrix $\mathbf{G}^{*}$. Finally, we denote the first row of $\mathbf{G}^{*}$ as a length-$(n-1)$ vector $\mathbf{z}$, and we define a matrix $\bar{\mathbf{G}}$ obtained by puncturing $\mathbf{z}$ from $\mathbf{G}^{*}$.  We demonstrate the relationship between $\bar{\mathbf{G}}$ and $\mathbf{G}$ in Fig. \ref{Fig::Gtrans}. We note that although $\mathbf{G}'$, $\mathbf{G}^{*}$, and $\bar{\mathbf{G}}$ are constructed relevantly to $i$, we omit their subscription $i$.
    
     \begin{figure} 
		\begin{center}
			\includegraphics[scale=0.7] {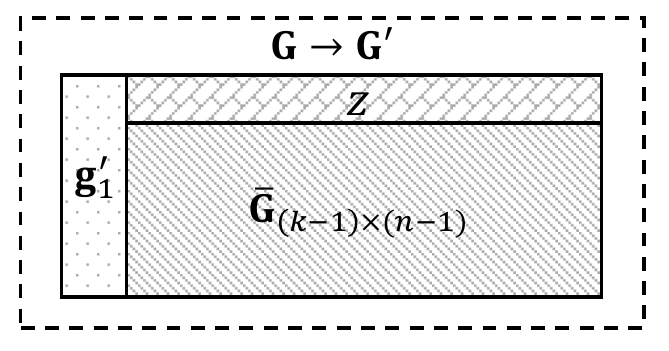}
			\caption{Obtain $\bar{\mathbf{G}}$ of Dual-OSD from $\mathbf{G}$.}
			\label{Fig::Gtrans}
		\end{center}
	\end{figure}
	
	Since $\mathbf{G}'$ is a systematic matrix, $\bar{\mathbf{G}}$ is also a systematic matrix and defines a punctured code $\mathcal{C}(n-1,k-1)$. Let $\bar{\bm{\ell}}$ be the sequence of LLRs corresponding to the columns of $\bar{\mathbf{G}}$, which is obtained by ordering $\bm{\ell}' = [\ell_1,\ldots,\ell_{i-1},\ell_{i+1},\ldots,\ell_n]$ in the descending order of $\bm{\alpha}'$. Also, let $\bar{\mathbf{y}}$ be the hard-decision vector obtained from $\bar{\bm{\ell}}$. Using the terminology of OSD, $[\bar{y}]_1^{k-1}$ is the MRB with respect to Dual-OSD, denoted by $\bar{\mathbf{y}}_{\mathrm{B}}$. 
	
	Next, we consider two phases of reprocessing over $\bar{\mathbf{G}}$, namely the phase-0 and phase-1 reprocessings. The phase-0 reprocessing re-encodes TEPs from a given TEP set $\bar{\mathcal{E}}^{(0)}$ with cardinality $|\bar{\mathcal{E}}^{(0)}| = N_0$, while the phase-1 reprocessing re-encodes TEPs from set $\bar{\mathcal{E}}^{(1)}$ with $|\bar{\mathcal{E}}^{(1)}| = N_1$. In both reprocessings, each length-$(k-1)$ TEP $\mathbf{e}$ is re-encoded to obtain a codeword as $\bar{\mathbf{c}}_{\mathbf{e}} = \left(\bar{\mathbf{y}}_{\mathrm{B}}\oplus \mathbf{e}\right)\bar{\mathbf{G}}$. Codewords generated by the phase-0 and phase-1 reprocessings are respectively included in sets $\bar{\mathcal{C}}_i^{(0)}$ and $\bar{\mathcal{C}}_i^{(1)}$. We assume that TEPs are sequentially processed in a manner similar to the original OSD (i.e., increasing the Hamming weight from 0) in both reprocessings; thus, the difference of these two reprocessing is that they re-encodes different numbers of TEPs. Particularly, we have $\bar{\mathcal{C}}_i^{(0)} = \bar{\mathcal{C}}_i^{(1)}$ when $N_{0} = N_{1}$. 
	
	Despite that $\bar{\mathcal{C}}_i^{(0)}$ and $\bar{\mathcal{C}}_i^{(1)}$ have duplicate elements (unnecessary for practical implementation), Dual-OSD can be used to conveniently find the density-transform feature of SOSD. Later, we will show that Dual-OSD with a proper selection of $N_1$ and $N_0$ is equivalent to SOSD (see Proposition \ref{Pro::Dual-SISO::MRB} and \ref{Pro::Dual-SISO::LRB}). 
	
	\vspace{-0.5em}
	\subsection{Outputs of Dual-OSD}
	\vspace{-0.5em}
	
	Let a random variable $V_{0}$ represent the minimum WHD from $\bar{\mathbf{y}}$ to codeword estimates in $\bar{\mathcal{C}}_i^{(0)}$, and a random variable $V_{1}$ represent the minimum WHD from $\bar{\mathbf{y}}\oplus \mathbf{z}$ to codeword estimates in $\bar{\mathcal{C}}_i^{(0)}$. Then, we have the following Lemma.
	
	\begin{lemma} \label{Lem::Dual-OSD}
        Let a length-$n$ sequence of i.i.d. LLRs be decoded by a Dual-OSD with $N_{0}$ and $N_{1}$. Then the $i$-th extrinsic LLR $\delta_i$ is given by 
        \begin{equation}\small\label{equ::DE::DEC::outLLR1}
            \delta_i = V_{0} - V_{1}.
        \end{equation}
    \end{lemma}
	\begin{IEEEproof}
	  We can rewrite the output extrinsic LLR $\delta_i$ of the $i$-th codeword bit, as
        \begin{equation}\small  \label{equ::APP::outLLR2}
        \begin{split}
            \delta_i & = \log\left(\frac{ \mathrm{Pr}(\bm{\ell}|\mathbf{c}(i:0) = \mathbf{c})}{ \mathrm{Pr}(\bm{\ell} | \mathbf{c}(i:1) = \mathbf{c})} \right)-\ell_i = \log\left(\frac{\prod_{c_j(i:0) = y_j}\exp(|\ell_j|)}{\prod_{c_j(i:1) = y_j}\exp(|\ell_j|)}\right)-\ell_i     \\
            &= \sum_{c_j(i:1) \neq y_j}|\ell_j|-\sum_{c_j(i:0) \neq y_j}|\ell_j| - \ell_i = \sum_{\substack{j\neq i \\ c_j(i:1) \neq y_j}}|\ell_j|-\sum_{\substack{j\neq i \\ c_j(i:0) \neq y_j}}|\ell_j|.
        \end{split}
        \end{equation} 
        
        Let $\mathcal{C}_i^{(0)}$ and $\mathcal{C}_i^{(1)}$ denote the sets of codewords $\mathcal{C}(n,k)$ whose $i$-th bit is 0 and 1, respectively. Matrix $\mathbf{G}' $ represents a codebook permuted from $\mathcal{C}(n,k)$ by $\pi'$. Thus, we can conclude that the $2^{k-1}$ codewords in $\bar{\mathcal{C}}_i^{(0)}$ are obtained by puncturing the $i$-th bit of the $2^{k-1}$ codewords in $\mathcal{C}_i^{(0)}$ and performing the permutation $\pi'$. In this way, $V_{0}$ is exactly a random variable representing $\sum_{\substack{j\neq i \\ c_j(i:0) \neq y_j}}|\ell_j|$. 
        
        Similarly, we can find that $2^{k-1}$ codewords in $\bar{\mathcal{C}}_i^{(1)}$ are are obtained by puncturing the $i$-th bit of the $2^{k-1}$ codewords in $\mathcal{C}_i^{(1)}$, performing the permutation $\pi'$, and XORing the vector $\mathbf{z}$. Specifically, let $\bar{\mathbf{c}}(i:1)$ denote the codeword from $\bar{\mathcal{C}}_i^{(1)}$ which have the minimum WHD to $\mathbf{y}'\oplus \mathbf{z}$. Then, $[1 \ \  \mathbf{z}\oplus \bar{\mathbf{c}}(i:1)]$ is exactly the codeword $\mathbf{c}(i:1)$ permuted by $\pi'$, i.e., $[1 \ \  \mathbf{z}\oplus \bar{\mathbf{c}}(i:1)] = \pi'(\mathbf{c}(i:1))$. Consequently, $V_{1}$ is exactly a random variable representing $\sum_{\substack{j\neq i \\ c_j(i:1) \neq y_j}}|\ell_j|$. Therefore, (\ref{equ::DE::DEC::outLLR1}) is equivalent to (\ref{equ::APP::outLLR2}).
	\end{IEEEproof}
	
	According to Lemma \ref{Lem::Dual-OSD}, when $N_{0} = N_{1} = 2^{k-1}$, Dual-OSD can output the ML extrinsic LLR obtained by exhausting the whole codebook of $\mathcal{C}(n,k)$. This is because $[0 \ \ \bar{\mathcal{C}}_i^{(0)}]$ and $[1 \ \ \mathbf{z}\oplus\bar{\mathcal{C}}_i^{(0)}]$ cover all the $2^k$ codewords in $\mathcal{C}(n,k)$. Furthermore, we have the following Proposition.
	\begin{proposition}  \label{Pro::Dual-SISO::identicaloutput}
	     For an arbitrary $i$ ($1\leq i \leq n$), the output of Dual-OSD, $\delta_i$, follows the same distribution.
	\end{proposition}
	\begin{proof}
        Since $\mathbf{G}$ is a linear block code, column $\mathbf{g}_1$ can be swapped with an arbitrary column of $\mathbf{G}$, corresponding to an arbitrary entry of the input LLR $[\ell]_1^n$, which does not change the distributions of $V_{0}$ and $V_{1}$.
	\end{proof}	
	From Proposition \ref{Pro::Dual-SISO::identicaloutput}, we can see that although Dual-OSD only outputs the $i$-th  extrinsic LLR for a given $i$, it is sufficient for us to characterize the overall density-transform feature of the decoder.
	
	In practice, SOSD employs a predetermined decoding order to limit the number of TEPs (generated codeword estimates). In other words, SOSD with a decoding order less than $k$ searches the decoding result within a reduced codebook of $\mathcal{C}(n,k)$. Accordingly, we can bound $N_{0}$ and $N_{1}$ of Dual-OSD to analyze the density-transform feature of SOSD, which is summarized in the following Propositions.
	\begin{proposition} \label{Pro::Dual-SISO::MRB}
	     Consider the $i$-th input LLR $\ell_i\in[\ell]_1^n$ to an order-$m$ SOSD ($m\geq1$). If $\ell_i$ is the LLR of an MRB position, its corresponding extrinsic LLR,  $\delta_i$, is equivalently obtained by a Dual-OSD with $N_{0} = \sum_{j=0}^{m}\binom{k-1}{j} $ and $N_{1} = \sum_{j=1}^{m}\binom{k-1}{j-1}$, i.e., with an order-$m$ phase-0 reprocessing and an order-$(m-1)$ phase-1 reprocessing.
	\end{proposition}
	\begin{proof}
        Assume that the hard-decision bit $y_i$ corresponding to $\ell_i$ is an MRB bit. Then, an order-$m$ SOSD generates $\sum_{j=0}^{m}\binom{k}{j}$ codeword estimates. Among them, there are $N_{1} = \sum_{j=1}^{m}\binom{k-1}{j-1}$ estimates in total whose $i$-th bit is opposite to $\widetilde{y}_i$, while there are $N_{0} = \sum_{j=0}^{m}\binom{k-1}{j}$ estimates in total whose $i$-th bit is the same as $y_i$. It is easy to obtain $\sum_{j=0}^{m}\binom{k}{j} = N_{0} + N_{1}$ with the recursive relationship $\binom{k}{j} = \binom{k-1}{j-1} + \binom{k-1}{j}$. Furthermore, it can be seen that $\bar{\mathbf{G}}$ and $\widetilde{\mathbf{G}}$ share the same MRB because $\ell_i$ is an MRB LLR. Therefore, the SOSD and Dual-OSD finds the same $\mathbf{c}(i:0)$ and $\mathbf{c}(i:1)$ in (\ref{equ::LowSISO::extLLR}).
	\end{proof}	
	\begin{proposition} \label{Pro::Dual-SISO::LRB}
	     If $\ell_i$ is not the LLR of an MRB position, its corresponding extrinsic LLR,  $\delta_i$, approximately follows the distribution of the output of a Dual-OSD with $N_{0} = N_{1} = \lceil\frac{1}{2}\sum_{j=0}^{m}\binom{k}{j} \rceil$.
	\end{proposition}
	\begin{proof}
        For arbitrary hard-decision bit $\widetilde{y}_i$, $k < i \leq n$ outside MRB, it will not be flipped by TEPs in the reprocessing of SOSD. Under the random code assumption, among the $\sum_{j=0}^{m}\binom{k}{j}$ estimates generated by the order-$m$ SOSD, on average $\frac{1}{2}\sum_{j=0}^{m}\binom{k}{j}$ of them have the $i$-th bit same as $\widetilde{y}_i$, or opposite to $\widetilde{y}_i$. Therefore, the extrinsic LLR $\delta_i$ can be approximated by a Dual-OSD with $N_{0} = N_{1} = \lceil\frac{1}{2}\sum_{j=0}^{m}\binom{k}{j} \rceil$, where rounding $\frac{1}{2}\sum_{j=0}^{m}\binom{k}{j}$ introduces the approximation.
	\end{proof}	
	
	\begin{remark}
	    In Proposition \ref{Pro::Dual-SISO::LRB}, $N_{0} = N_{1}$ holds only when $\mathcal{C}(n,k)$ has a binomial weight spectrum (under the random code assumption). When $\mathcal{C}(n,k)$ is an arbitrary code with a given generator matrix $\mathbf{G}$, instead we should have $N_{0} = p\sum_{j=0}^{m}\binom{k}{j}$ and $N_{1} = (1 - p) \sum_{j=0}^{m}\binom{k}{j}$. Here $p$ is the probability that an arbitrary parity bit of a codeword is zero, where the codeword is encoded by an information block with a given Hamming weight related to the TEP. In general, $p$ is determined by the structure of $\mathcal{C}(n,k)$, which was characterized in \cite[Lemma 4, Lemma 5]{yue2021revisit}. 
	\end{remark}
	
		Lemma \ref{Lem::Dual-OSD} and Proposition \ref{Pro::Dual-SISO::MRB} and \ref{Pro::Dual-SISO::LRB} indicate that we can analyze Dual-OSD to characterize the density transform of SOSD or any other variants of SOSD, by properly bounding the value of $N_{0}$ and $N_{1}$.
	
	\vspace{-0.5em}	
	\subsection{Density-Transform Feature}
	\vspace{-0.5em}
	
		In this subsection, we analyze the density-transform feature of Dual-OSD. For simplicity, we consider an order-$(m_0,m_1)$ Dual-OSD in this section, which has an order-$m_0$ phase-0 reprocessing and an order-$m_1$ phase-1 reprocessing, i.e., $N_0 = \sum_{i=0}^{m_0}\binom{k}{i}$ and $N_1 = \sum_{i=0}^{m_1}\binom{k}{i}$. In Section \ref{Sec::DualOSD::results}, we will show via simulations and numerical examples that the order-$(m_0,m_1)$ Dual-OSD and the order-$m$ SOSD have very similar output density, when $m_0=m_1=m$. Nevertheless, a more accurate density-transform feature of SOSD can be obtained by selecting $N_0$ and $N_1$ exactly according to Proposition \ref{Pro::Dual-SISO::MRB} and \ref{Pro::Dual-SISO::LRB} (see Remark \ref{rem::GeneralSOSD}).
		
	    To find the $\mathrm{pdf}$ of the extrinsic LLR $\delta_i$, $\mathrm{pdf}$s of $V_{0}$ and $V_{1}$ are required. According to Lemma \ref{Lem::Dual-OSD}, $V_{0}$ represents the minimum WHD from $\bar{\mathbf{y}}$ to codewords in $\bar{\mathcal{C}}_i^{(0)}$, which is equivalent to the minimum WHD within an order-$m_0$ original OSD decoding with the code $\mathcal{C}(n-1,k-1)$. Therefore, $V_{0}$ can be characterized by utilizing the results of the distributions of the minimum WHD introduced in our previous work \cite{yue2021revisit}. Let $\mathcal{V}_{0}(\ell)$ denote the $\mathrm{pdf}$ of $V_{0}$ resulted from an order-$m_0$ phase-$0$ reprocessing. Then, according to \cite[Theorem 4]{yue2021revisit}, we can approximately derive $\mathcal{V}_{0}(\ell)$ as
    	\begin{equation}\small \label{equ::dualOSD::pdfV0}
    	    \begin{split}
    	        \mathcal{V}_{0}(\ell) & \approx  \!\sum_{j=0}^{m_0}p_{E_1^{k-1}}(j)\! \cdot \left(f_{e}(\ell)\!\!\int_{\ell}^{\infty}\!\!\!f_{m}\left(u, N_0-1 \right) du + f_{m}\left(\ell, N_0-1\right) \int_{\ell}^{\infty}\!\!f_{e}(u)du \! \right)  +  \left(1 - \sum_{j=0}^{m_0}p_{E_1^{k-1}}(j)\right) f_{m}\left(\ell, N_0 \right) ,
    	    \end{split}
    	\end{equation}
    	where
        \begin{equation}\small \label{equ::dualOSD::pdfV0::com1}
    		\begin{split}
                f_{m}\left(\ell, b\right) &= \int_{-\infty}^{\infty} \left(\sqrt{1-\rho} \ \bar{\sigma}_2 \right)^{-1} \cdot f_{\phi} \left(\frac{(\ell - \bar{\mu}_2)/\bar{\sigma}_2   + \sqrt\rho z}{\sqrt{1- \rho}},b\right) \phi (z) \ dz, 	    
    		\end{split}
        \end{equation}
    	\begin{equation}\small
        	f_{\phi}(\ell,b) = b \ \phi (\ell)\left(1 - \int_{-\infty}^{\ell} \phi (u)du\right)^{b-1}.
    	\end{equation}	
	    Furthermore, $f_{e}(\ell)$ is the $\mathrm{pdf}$ of the Gaussian distribution $\mathcal{N}(\bar{\mu}_1,\bar{\sigma}_1^2)$, which represents the distribution of the WHD between the transmitted codeword and $\bar{\mathbf{y}}$, where $\bar{\mu}_1$ and $\bar{\sigma}_1^2$ are respectively given by \cite[Eq. (123)]{yue2021revisit} and \cite[Eq. (124)]{yue2021revisit}. $\phi(\ell)$ is the $\mathrm{pdf}$ of standard Gaussian distribution $\mathcal{N}(0,1)$, and $p_{E_1^{k-1}}(j)$ is given in \cite[Lemma 1]{yue2021revisit}, which represents the probability that some TEP with Hamming weight $j$ can eliminate the errors in MRB. $f_{m}\left(\ell, b\right)$ represents the $\mathrm{pdf}$ of $V_{0}$ when the decoder does not find the transmitted codeword, in which $\bar{\mu}_2$ and $\bar{\sigma}_2^2$ are respectively given by \cite[Eq. (126)]{yue2021revisit} and \cite[Eq. (127)]{yue2021revisit}. Finally, $\rho$ is a correlation coefficient given by \cite[Eq. (100)]{yue2021revisit}. We note that (\ref{equ::dualOSD::pdfV0}) is an approximation rather than an explicit expression, because $f_{e}(\ell)$ and $f_{m}\left(\ell, b\right)$ regard the WHDs from codewords to $\bar{\mathbf{y}}$ as Gaussian variables. According to the definition (\ref{equ::WHDdefine}) of WHD and Central Limit Theorem, WHDs tend to be Gaussian variables, when the code length $n$ is not small. We further note that $p_{E_1^{k-1}}(j)$, $f_{e}(\ell)$, and $f_{m}\left(\ell, b\right)$ are fully determined by the $\mathrm{pdf}$ of the input LLR $\bm{\ell}$. In other words, (\ref{equ::dualOSD::pdfV0}) is a function of the $\mathrm{pdf}$ of the input LLRs.

	    \begin{remark} \label{rem::GeneralSOSD}
	        More generally, when $\sum_{i=0}^{m_0-1}\binom{k}{i}<N_0 <\sum_{i=0}^{m_0}\binom{k}{i}$, i.e., only a part of TEPs with Hamming weight $m_0$ are re-encoded, $ \mathcal{V}_{0}(\ell)$ can be obtained by replacing $\sum_{j=0}^{m_0}p_{E_1^{k-1}}(j)$ in (\ref{equ::dualOSD::pdfV0}) by $\mathrm{Pe}$. Here, $\mathrm{Pe}$ is the probability that the errors over MRB is eliminated by one of the re-encoded TEPs $\{\mathbf{e}_1,\ldots, \mathbf{e}_{N_0}\}$, which is derived as $\mathrm{Pe} = \sum_{j=1}^{N_1}\mathrm{Pe}(\mathbf{e}_j)$, where $\mathrm{Pe}(\mathbf{e}_j)$ was given by \cite[Eq. (131)]{yue2021revisit}.
	    \end{remark}

	    The random variable $V_{1}$ should be determined differently from $V_{0}$. Random variable $V_{1}$ can be regarded as the minimum WHD from estimates to the vector $\bar{\mathbf{y}}\oplus \mathbf{z}$ in an order-$m_1$ OSD decoding of $\mathcal{C}(n-1,k-1)$. Let $\bar{\mathbf{d}}_{\mathbf{e}} = \bar{\mathbf{c}}_{\mathbf{e}}\oplus \bar{\mathbf{y}}_{\mathrm{B}}\oplus\mathbf{z}$ denote the difference pattern between the codeword $\bar{\mathbf{c}}_{\mathbf{e}}$ and $\bar{\mathbf{y}}_{\mathrm{B}}\oplus\mathbf{z}$. Then, when $\mathbf{e}$ can eliminate the hard-decision errors on $\bar{\mathbf{y}}_{\mathrm{B}}$, we have
        \begin{equation}\small \label{equ::dualOSD::eB=e}
            \bar{\mathbf{d}}_{\mathbf{e}} =  \left((\bar{\mathbf{y}}_{\mathrm{B}}\oplus \mathbf{e})\bar{\mathbf{G}}\right) \oplus (\bar{\mathbf{y}}\oplus\mathbf{z}) = [\mathbf{e} \ \ \bar{\mathbf{e}}_{\mathrm{P}}\oplus \mathbf{z}_{\mathrm{P}}],
        \end{equation}
        where $\bar{\mathbf{e}}_{\mathrm{P}}$ is the hard-decision error over $[\bar{y}]_{k}^{n-1}$ and $\mathbf{z}_{\mathrm{P}} = [z]_{k}^{n-1}$. On the other hand, let us check  $\bar{\mathbf{d}}_{\mathbf{e}}$ when $\mathbf{e}$ cannot eliminate the hard-decision errors over $\bar{\mathbf{y}}_{\mathrm{B}}$, and we have
        \begin{equation}\small \label{equ::dualOSD::e!=eB}
        \begin{split}
            \bar{\mathbf{d}}_{\mathbf{e}} &= \left((\bar{\mathbf{y}}_{\mathrm{B}}\oplus \mathbf{e})\bar{\mathbf{G}}\right) \oplus \bar{\mathbf{y}} = \left((\bar{\mathbf{e}}_{\mathrm{B}} \oplus \mathbf{e})\bar{\mathbf{G}}\right) \oplus \bar{\mathbf{e}} \oplus \mathbf{z}  = [\mathbf{e} \ \ \bar{\mathbf{e}}_{\mathrm{P}}\oplus \mathbf{z}_{\mathrm{P}} \oplus \bar{\mathbf{c}}_{\mathrm{P}}' ],
        \end{split}
        \end{equation}
        where $\bar{\mathbf{c}}_{\mathrm{P}}'$ is the parity part of a codeword $\bar{\mathbf{c}}' = (\bar{\mathbf{e}}_{\mathrm{B}} \oplus \mathbf{e})\bar{\mathbf{G}}$ from $\mathcal{C}(n-1,k-1)$.  Moreover, $\bar{\mathbf{c}}_{\mathrm{P}}'$ is also the parity part of a codeword $[0 \ \ (\bar{\mathbf{e}}_{\mathrm{B}} \oplus \mathbf{e})\bar{\mathbf{G}}]$ from $\mathcal{C}(n,k)$. Furthermore, it can be seen from Fig. \ref{Fig::Gtrans} that $[1 \ \ \mathbf{z}]$ is a codeword from $\mathcal{C}(n,k)$, and thus $\mathbf{z}_{\mathrm{P}} \oplus \bar{\mathbf{c}}_{\mathrm{P}}'$ is the parity part of a codeword $[1 \ \  (\bar{\mathbf{e}}_{\mathrm{B}} \oplus \mathbf{e})\bar{\mathbf{G}}\oplus \mathbf{z}]$ from $\mathcal{C}(n,k)$. Comparing (\ref{equ::dualOSD::eB=e}) and (\ref{equ::dualOSD::e!=eB}), the difference pattern $\bar{\mathbf{d}}_{\mathbf{e}}$ follows the same distribution (in terms of the distribution of nonzero elements) regardless of whether the TEP $\mathbf{e}$ can eliminate errors over MRB. This indicates that $V_{1}$ is given by the minimum variable among $N_1$ identical but dependent random variables representing WHDs, which can be characterized by the dependent ordered statistics \cite{tong2012multivariate}. Therefore, by leveraging and modifying \cite[Theorem 4]{yue2021revisit}, we can obtain the $\mathrm{pdf}$ of $V_{1}$ as 
    	\begin{equation}\small \label{equ::dualOSD::pdfV1}
    	        \mathcal{V}_{1}(\ell) \approx  f_{m}\left(\ell, N_1\right) ,
    	\end{equation}
        where $ f_{m}\left(\ell, b\right)$ is given by (\ref{equ::dualOSD::pdfV0::com1}). We omit the detailed derivation of (\ref{equ::dualOSD::pdfV0}) and (\ref{equ::dualOSD::pdfV1}) due to space limits and refer interested readers to \cite{yue2021revisit}.
        
     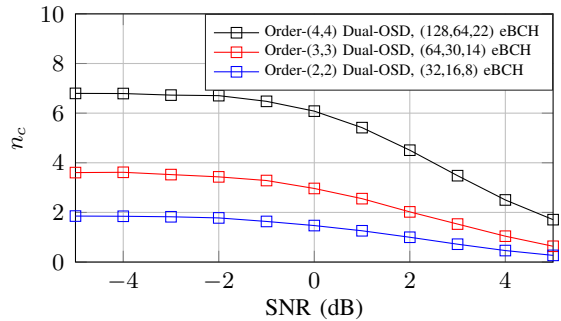
\begin{figure} 
     \centering
    \begin{tikzpicture}
    
    \begin{axis}[%
    width=2.5in,
    height=1.3in,
    at={(0.822in,0.529in)},
    scale only axis,
    xmin=-5,
    xmax=5,
    xlabel style={at={(0.5,1ex)},font=\color{white!15!black},font = \footnotesize},
    xlabel={SNR (dB)},
    ymin=0,
    ymax=10,
    ylabel style={at={(2.5ex,0.5)},font=\color{white!15!black},font = \footnotesize},
    ylabel={$n_c$},
    axis background/.style={fill=white},
    tick label style={font=\footnotesize},
    xmajorgrids,
    ymajorgrids,
    legend style={at={(1,1)}, anchor=north east, legend cell align=left, align=left, draw=white!15!black,font = \tiny,row sep=-2.5pt}
    ]
    \addplot [color=black, mark=square, mark options={solid, black}]
      table[row sep=crcr]{%
    -5	6.796\\
    -4	6.79\\
    -3	6.727\\
    -2	6.705\\
    -1	6.476\\
    0	6.081\\
    1	5.416\\
    2	4.507\\
    3	3.48\\
    4	2.502\\
    5	1.708\\
    };
    \addlegendentry{Order-(4,4) Dual-OSD, (128,64,22) eBCH}

    \addplot [color=red, mark=square, mark options={solid, red}]
      table[row sep=crcr]{%
    -5	3.6048\\
    -4	3.6166\\
    -3	3.5238\\
    -2	3.4324\\
    -1	3.2848\\
    0	2.9654\\
    1	2.5544\\
    2	2.0212\\
    3	1.5306\\
    4	1.0422\\
    5	0.636\\
    };
    \addlegendentry{Order-(3,3) Dual-OSD, (64,30,14) eBCH}
    
     \addplot [color=blue, mark=square, mark options={solid, blue}]
      table[row sep=crcr]{%
    -5	1.8532\\
    -4	1.8448\\
    -3	1.8242\\
    -2	1.7782\\
    -1	1.6362\\
    0	1.4704\\
    1	1.26\\
    2	0.9994\\
    3	0.7192\\
    4	0.4632\\
    5	0.2692\\
    };
    \addlegendentry{Order-(2,2) Dual-OSD, (32,16,8) eBCH}

    \end{axis}
    \end{tikzpicture}%
	\vspace{-0.5em}
    \caption{The value of $n_c = |(\bar{\mathbf{c}}(i:0)\oplus\bar{\mathbf{y}})\land (\bar{\mathbf{c}}(i:1)\oplus\bar{\mathbf{y}}\oplus \mathbf{z})|$ in decoding various codes }
	\vspace{-0.5em}
	\label{Fig::nc}
        
	\end{figure}

         Next, according to Lemma \ref{Lem::Dual-OSD}, the $\mathrm{pdf}$ of $\delta_i$ can be characterized by $\mathrm{pdf}$s of $V_{0}$ and $V_{1}$. We shall, however, exercise caution with the correlation between $V_{0}$ and $V_{1}$. Precisely, we assume that the codeword $\bar{\mathbf{c}}(i:0)$ has the minimum WHD $V_0$ to $\bar{\mathbf{y}}$, and the codeword $\bar{\mathbf{c}}(i:1)$ has the minimum WHD $V_1$ to $\bar{\mathbf{y}}\oplus\mathbf{z}$, respectively, found by the Dual-OSD. Then, $\bar{\mathbf{c}}(i:0)\oplus\bar{\mathbf{y}}$ and $\bar{\mathbf{c}}(i:1)\oplus\bar{\mathbf{y}}\oplus \mathbf{z}$ may share the common non-zero bits, which introduce the correlation between $V_{0}$ and $V_{1}$. Nevertheless, from the simulation, we can find that the impact of this correlation is marginal. Let $n_c$ denote the number of common nonzero-bits of $\bar{\mathbf{c}}(i:0)\oplus\bar{\mathbf{y}}$ and $\bar{\mathbf{c}}(i:1)\oplus\bar{\mathbf{y}}\oplus \mathbf{z}$, i.e., $n_c = |(\bar{\mathbf{c}}(i:0)\oplus\bar{\mathbf{y}})\land (\bar{\mathbf{c}}(i:1)\oplus\bar{\mathbf{y}}\oplus \mathbf{z})|$, where $\land$ is the bit-wise AND operator. The average values of $n_c$ in decoding various eBCH codes are depicted in Fig. \ref{Fig::nc}. It shows that values of $n_c$ are nearly negligible compared to code lengths at moderate-to-high SNRs, indicating that $V_{0}$ and $V_{1}$ only have a very small correlation. Therefore, we have the following assumption for the simplicity of analysis.
             
        \begin{assumption} \label{Assum::dualOSD::independence}
            $V_{0}$ and $V_{1}$ are independent random variables.
        \end{assumption}
        
        Under Assumption \ref{Assum::dualOSD::independence}, the $\mathrm{pdf}$ of $\delta_i$ is summarized in the following Theorem.
        \begin{theorem} \label{the::Extrinsicdensity}
            Let $\mathcal{D}_i (\ell)$ denote the $\mathrm{pdf}$ of the extrinsic LLR $\delta_i$ output by an order-$(m_0,m_1)$ Dual-OSD. Then $\mathcal{D}_i (\delta)$ is given by
            \begin{equation}\small\label{equ::dualOSD::Extrinsicdensity}
                \mathcal{D}_i(\ell) = \left(\mathcal{V}_0\otimes \mathcal{V}_1^{(-)} \right)(\ell),
            \end{equation}
        	where $f\otimes g$ is the convolution of $f(\ell)$ and $g(\ell)$, i.e., $(f\otimes g)(\ell) = \int_{-\infty}^{\infty}f(\tau)g(\ell-\tau)d\tau$, and $\mathcal{V}_1^{(-)}$ is the symmetric function of $\mathcal{V}_1$, i.e., $\mathcal{V}_1^{(-)}(\ell) = \mathcal{V}_1(-\ell)$.  $\mathcal{V}_0(\ell)$ and $\mathcal{V}_1(\ell)$ are respectively given by (\ref{equ::dualOSD::pdfV0}) and (\ref{equ::dualOSD::pdfV1}), which are determined by the $\mathrm{pdf}$, $\mathcal{L}_j(\ell)$, of input LLRs $\ell_j$, $1\leq j\leq n$ and $j\neq i$. 
        \end{theorem}
        \begin{IEEEproof}
            Theorem \ref{the::Extrinsicdensity} is directly proved from Lemma \ref{Lem::Dual-OSD} under Assumption \ref{Assum::dualOSD::independence}.
        \end{IEEEproof}
        
        We refer $\mathcal{D}_i(\ell)$ and  $\mathcal{L}_i(\ell)$ to as the extrinsic density and priori density of the $i$-th position, respectively. Then, the $\mathrm{pdf}$ of the posterior LLR, i.e., the posterior density, is summarized in the following Corollary.
        
        \begin{corollary}
            Let $\mathcal{P}_i (\ell)$ denote the $\mathrm{pdf}$ of the $i$-th posterior LLR, i.e., $\delta_i + \ell_i$, output by an order-$(m_0,m_1)$ Dual-OSD.  Then $\mathcal{P}_i (\ell)$ is given by
            \begin{equation}\small \label{equ::dualOSD::Postdensity}
                \mathcal{P}_i(\ell) = \left(\mathcal{L}_i \otimes \mathcal{V}_0\otimes \mathcal{V}_1^{(-)} \right)(\ell).
            \end{equation}
        \end{corollary}
        \begin{IEEEproof}
            Eq. (\ref{equ::dualOSD::Postdensity}) is directly obtained by observing from (\ref{equ::APP::outLLR2}) that $\ell_i$ is independent with $V_0$ and $V_1$.
        \end{IEEEproof}
        To reduce the correlation between the input and the output of the decoder, the extrinsic LLR rather than the posterior LLR, is usually used in iterative decoders \cite{yue2021noma,divsalar2001iterative}. Since the input LLRs $[\ell]_1^n$ and the output LLRs $[\delta]_1^n$ are respectively identically distributed, we omit the subscripts $i$ and $j$ of $\mathcal{D}_i$ and $\mathcal{L}_j$, and simply represent
        (\ref{equ::dualOSD::Extrinsicdensity}) as
        \begin{equation}\small \label{equ::dualOSD::transform}
            \mathcal{D}(\ell) = \Delta(\mathcal{L}(\ell)),
        \end{equation}
        where $\Delta(\cdot)$ is referred to as the density-transform feature of SOSD.
        
     \vspace{-0.5em}   
	\subsection{Numerical Examples} \label{Sec::DualOSD::results}
	\vspace{-0.5em}
	
	    In this subsection, we demonstrate the extrinsic densities with various priori densities given. We assume the all-zero transmission, i.e., the transmitted codeword $\mathbf{c}$ is an all-zero codeword and it is transmitted with all positive BPSK symbols. This assumption will not lost generality, when each codeword in the codebook is equally likely to be transmitted over a symmetric channel. Under the all-zero transmission, we denote the $\mathrm{pdf}$ of $\ell_i$ as $\dot{\mathcal{L}}_i(\ell) = \mathcal{L}_i(\ell|c_i=0)$, which is referred to as the single-side priori density. Accordingly, we denote the single-side density of $\delta_i$ as $\dot{\mathcal{D}}_i(\ell) = \mathcal{D}_i(\ell|c_i=0)$. By omitting the subscript, we directly have the following relationships.
        \begin{equation}\small
            \mathcal{L}(\ell) = \frac{1}{2}\dot{\mathcal{L}}(\ell) + \frac{1}{2}\dot{\mathcal{L}}(-\ell), \ \ \text{and} \ \ 
            \mathcal{D}(\ell) = \frac{1}{2}\dot{\mathcal{D}}(\ell) + \frac{1}{2}\dot{\mathcal{D}}(-\ell).
        \end{equation}
        
        \begin{example} \label{Exam::single64outDensity}
    	    We consider the single-user all-zero transmission over an AWGN channel, where the SNR is defined as $\mathrm{SNR} = \frac{1}{\sigma^2}$. Thus, the input LLRs follow the distribution $\mathcal{N}(\frac{2}{\sigma^2}.\frac{4}{\sigma^2} )$. The $(64,30,14)$ eBCH code is decoded by an order-3 SOSD. We depict and compare $\dot{\mathcal{D}}(\ell)$ and $\dot{\mathcal{L}}(\ell)$ at various SNRs in Fig. \ref{Fig::singleDensity::AWGN}, where the numerical results for $\dot{\mathcal{D}}(\ell)$ is obtained from (\ref{equ::dualOSD::Extrinsicdensity}) with $m_1=m_0=3$.
    	    
    	    As shown, although the simulation results of $\dot{\mathcal{D}}(\ell)$ are obtained by performing the order-3 SOSD, it is still well approximated by (\ref{equ::dualOSD::Extrinsicdensity}) with $m_1=m_0=3$. Nevertheless, as discussed in Remark \ref{rem::GeneralSOSD}, more accurate results can be obtained by selecting $N_1$ and $N_0$ more carefully. We further note that there is a gap between  (\ref{equ::dualOSD::Extrinsicdensity}) and the simulation results at SNR = 0 dB because of the relatively high correlation between $V_0$ and $V_1$ at low SNRs.
    	\end{example}
    	
    	\begin{example}
	        We consider the input LLR with the density of $\dot{\mathcal{L}} \sim 0.5\mathcal{N}(1,0.631) + 0.5\mathcal{N}(3,1.262)$, where the $(64,30,14)$ eBCH code is decoded by an order-3 SOSD. We depict and compare $\dot{\mathcal{D}}(\ell)$, $\dot{\mathcal{L}}(\ell)$, and  $\dot{\mathcal{P}}(\ell) = \mathcal{P}(\ell|c_i=0)$ in Fig. \ref{Fig::singleDensity::Irregular}. The density of the input LLR analogizes the input signal suffering from the interference from another user in the multi-user transmission. As demonstrated, while the input LLR does not follow a Gaussian distribution, the extrinsic and posterior LLRs approximately do.
	        
    	\end{example}

       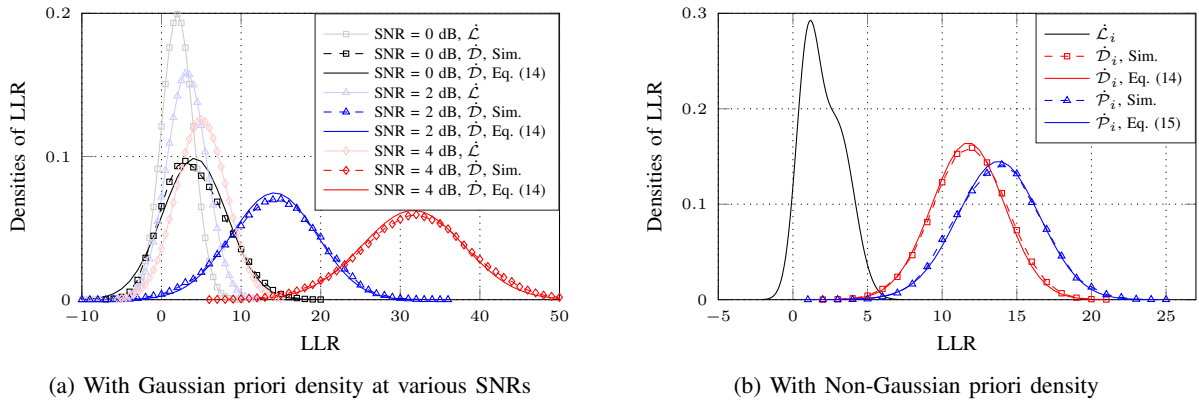
\begin{figure} 
             \centering
             \hspace{-0.81em}
             \begin{subfigure}[b]{0.45\columnwidth}
                 \centering
                 
                \begin{tikzpicture}
                
                \begin{axis}[%
                width=2.5in,
                height=1.5in,
                at={(0.785in,0.587in)},
                scale only axis,
                xmin=-10,
                xmax=50,
                xlabel style={at={(0.5,1ex)},font=\color{white!15!black},font=\scriptsize},
                xlabel={LLR},
                ymin= 0,
                ymax=0.2,
                yminorticks=true,
                ylabel style={at={(2ex,0.5)},font=\color{white!15!black},font=\scriptsize},
                ylabel={Densities of LLR},
                ytick={0,0.1,0.2},
                axis background/.style={fill=white},
                tick label style={font=\tiny},
                xmajorgrids,
                ymajorgrids,
                yminorgrids,
                minor grid style={dotted},
                major grid style={dotted,black},
                legend style={at={(1,1)}, anchor=north east, legend cell align=left, align=left, draw=white!15!black,font = \tiny,row sep=-4pt}
                ]

                \addplot [color=white!80!black,mark=square, mark size=1pt]
                  table[row sep=crcr]{%
                -10	3.03794142491164e-09\\
                -9.5	1.31962160178529e-08\\
                -9	5.38488002127164e-08\\
                -8.5	2.064235494315e-07\\
                -8	7.43359757367149e-07\\
                -7.5	2.51475364429622e-06\\
                -7	7.99187055345274e-06\\
                -6.5	2.38593182706025e-05\\
                -6	6.69151128824427e-05\\
                -5.5	0.000176297841183723\\
                -5	0.00043634134752288\\
                -4.5	0.00101452402864988\\
                -4	0.002215924205969\\
                -3.5	0.00454678125079553\\
                -3	0.00876415024678427\\
                -2.5	0.0158698259178337\\
                -2	0.026995483256594\\
                -1.5	0.0431386594132558\\
                -1	0.0647587978329459\\
                -0.5	0.091324542694511\\
                0	0.120985362259572\\
                0.5	0.150568716077402\\
                1	0.17603266338215\\
                1.5	0.193334058401425\\
                2	0.199471140200716\\
                2.5	0.193334058401425\\
                3	0.17603266338215\\
                3.5	0.150568716077402\\
                4	0.120985362259572\\
                4.5	0.091324542694511\\
                5	0.0647587978329459\\
                5.5	0.0431386594132558\\
                6	0.026995483256594\\
                6.5	0.0158698259178337\\
                7	0.00876415024678427\\
                7.5	0.00454678125079553\\
                8	0.002215924205969\\
                8.5	0.00101452402864988\\
                9	0.00043634134752288\\
                9.5	0.000176297841183723\\
                10	6.69151128824427e-05\\
                };
                \addlegendentry{SNR = 0 dB, $\dot{\mathcal{L}}$}
                
                \addplot [color=black, dashed, mark=square, mark size=1pt, mark options={solid, black}]
                  table[row sep=crcr]{%
                -7	0.00032\\
                -6	0.00072\\
                -5	0.00194\\
                -4	0.00554\\
                -3	0.01378\\
                -2	0.0292\\
                -1	0.0446\\
                -0	0.06528\\
                1	0.08234\\
                2	0.09332\\
                3	0.09684\\
                4	0.09248\\
                5	0.08658\\
                6	0.07874\\
                7	0.06842\\
                8	0.0568\\
                9	0.04596\\
                10	0.03532\\
                11	0.0265\\
                12	0.01714\\
                13	0.01166\\
                14	0.0067\\
                15	0.00392\\
                16	0.00176\\
                17	0.0011\\
                18	0.00052\\
                19	0.00022\\
                20	0.0001\\
                };
                \addlegendentry{SNR = 0 dB, $\dot{\mathcal{D}}$, Sim.}

                \addplot [color=black]
                  table[row sep=crcr]{%
                -10	0.000213085964218014\\
                -9	0.000491663374976333\\
                -8	0.00106709062689309\\
                -7	0.00217851726498687\\
                -6	0.00418359228047661\\
                -5	0.00755731523256871\\
                -4	0.0128415073910344\\
                -3	0.0205255545370846\\
                -2	0.0308605937345095\\
                -1	0.0436459720071746\\
                0	0.0580650200709502\\
                1	0.0726633655822189\\
                2	0.0855356092275014\\
                3	0.0947128467136635\\
                4	0.0986509419086999\\
                5	0.0966548566031889\\
                6	0.0890789670518548\\
                7	0.0772240249995257\\
                8	0.0629717959040717\\
                9	0.0482982541257679\\
                10	0.0348376647085801\\
                11	0.0236247830764073\\
                12	0.0150530349859755\\
                13	0.00900189809493556\\
                14	0.00504302089046673\\
                15	0.00263917918063901\\
                16	0.00128521066478274\\
                17	0.000579516666459982\\
                18	0.000240577608238608\\
                19	9.13820089585714e-05\\
                20	3.15629073355864e-05\\
                };
                \addlegendentry{SNR = 0 dB, $\dot{\mathcal{D}}$, Eq. (\ref{equ::dualOSD::Extrinsicdensity})}
                
                \addplot [color=white!80!blue, mark = triangle, mark size=1.5pt]
                  table[row sep=crcr]{%
                -6	0.000208807273152742\\
                -5.5	0.0004219607414014\\
                -5	0.000819732475116321\\
                -4.5	0.00153089676018451\\
                -4	0.00274848477997216\\
                -3.5	0.00474366965953632\\
                -3	0.00787062478252964\\
                -2.5	0.0125538706243785\\
                -2	0.0192495134953859\\
                -1.5	0.0283749781725662\\
                -1	0.0402091584423915\\
                -0.5	0.0547757206898793\\
                0	0.0717339732024011\\
                0.5	0.0903098958193332\\
                1	0.109299824731584\\
                1.5	0.127167829547071\\
                2	0.142235735094529\\
                2.5	0.152937453064944\\
                3	0.158085724678055\\
                3.5	0.157088764718281\\
                4	0.150062184380733\\
                4.5	0.137806934669073\\
                5	0.121659086250988\\
                5.5	0.103250389153152\\
                6	0.0842388662805185\\
                6.5	0.0660704095103798\\
                7	0.0498167226947328\\
                7.5	0.0361091260049255\\
                8	0.0251612648756602\\
                8.5	0.0168547234734884\\
                9	0.0108538661838602\\
                9.5	0.00671925252715553\\
                10	0.00399881334631419\\
                10.5	0.00228778375429941\\
                11	0.00125826605254171\\
                11.5	0.000665278757964547\\
                12	0.000338149308147588\\
                };
                \addlegendentry{SNR = 2 dB, $\dot{\mathcal{L}}$}
                
                \addplot [color=blue, dashed, mark=triangle, mark size=1.5pt, mark options={solid, blue}]
                  table[row sep=crcr]{%
                -10	2e-05\\
                -9	2e-05\\
                -8	4e-05\\
                -7	2e-05\\
                -6	0.00016\\
                -5	0.00018\\
                -4	0.00044\\
                -3	0.00092\\
                -2	0.00172\\
                -1	0.0028\\
                -0	0.00382\\
                1	0.00584\\
                2	0.0084\\
                3	0.01196\\
                4	0.01438\\
                5	0.01908\\
                6	0.02498\\
                7	0.03152\\
                8	0.03822\\
                9	0.04496\\
                10	0.05158\\
                11	0.05914\\
                12	0.0641\\
                13	0.06822\\
                14	0.06976\\
                15	0.07022\\
                16	0.06658\\
                17	0.06384\\
                18	0.05604\\
                19	0.04982\\
                20	0.04376\\
                21	0.03346\\
                22	0.02842\\
                23	0.02138\\
                24	0.01604\\
                25	0.01026\\
                26	0.0068\\
                27	0.0042\\
                28	0.00304\\
                29	0.00176\\
                30	0.00096\\
                31	0.00062\\
                32	0.00026\\
                33	0.0001\\
                34	0.0001\\
                35	2e-05\\
                36	4e-05\\
                };
                \addlegendentry{SNR = 2 dB, $\dot{\mathcal{D}}$, Sim.}
                
                \addplot [color=blue]
                  table[row sep=crcr]{%
                -5	0.000145912024530828\\
                -4	0.000274518808656095\\
                -3	0.00049934100105571\\
                -2	0.000878146142422783\\
                -1	0.00149307257480807\\
                0	0.00245436785748973\\
                1	0.00390070400191434\\
                2	0.00599364409539085\\
                3	0.00890396568139933\\
                4	0.0127885256404988\\
                5	0.0177583142940512\\
                6	0.0238411275641003\\
                7	0.030945298481676\\
                8	0.0388332256068441\\
                9	0.0471139423576843\\
                10	0.055261886870246\\
                11	0.0626642546677041\\
                12	0.0686927014473399\\
                13	0.0727884361422005\\
                14	0.0745450100466152\\
                15	0.073772091216007\\
                16	0.0705268380640128\\
                17	0.0651063424041502\\
                18	0.0580030204944612\\
                19	0.04983240847326\\
                20	0.0412476760781959\\
                21	0.0328564849747069\\
                22	0.0251538389046043\\
                23	0.0184802213143173\\
                24	0.013008712782122\\
                25	0.00875902356437019\\
                26	0.00563153387161731\\
                27	0.00345149566228235\\
                28	0.00201316178346312\\
                29	0.00111572192914737\\
                30	0.000586675333273375\\
                31	0.000292288619272541\\
                32	0.000137802599077543\\
                };
                \addlegendentry{SNR = 2 dB, $\dot{\mathcal{D}}$, Eq. (\ref{equ::dualOSD::Extrinsicdensity})}
                
                \addplot [color=white!80!red, mark=diamond, mark size=1.5pt]
                  table[row sep=crcr]{%
                -5.5	0.00050857966566436\\
                -5	0.000847999208031525\\
                -4.5	0.00137919577203176\\
                -4	0.00218801554622668\\
                -3.5	0.00338585928835842\\
                -3	0.00511071217807677\\
                -2.5	0.00752467863844851\\
                -2	0.0108065858109245\\
                -1.5	0.0151385078923372\\
                -1	0.0206857681441119\\
                -0.5	0.0275711099008942\\
                0	0.0358451879558681\\
                0.5	0.0454570753504152\\
                1	0.0562297476538564\\
                1.5	0.0678460853623142\\
                2	0.0798504687948292\\
                2.5	0.0916693520691334\\
                3	0.102651398156546\\
                3.5	0.112124260602526\\
                4	0.119461601067736\\
                4.5	0.124151246900108\\
                5	0.125854241380442\\
                5.5	0.124445341101703\\
                6	0.120028241793744\\
                6.5	0.112922962268421\\
                7	0.103627517886828\\
                7.5	0.0927602574578069\\
                8	0.0809921281353763\\
                8.5	0.0689791266700254\\
                9	0.0573042177382855\\
                9.5	0.046435433139774\\
                10	0.0367034111990904\\
                10.5	0.0282981062846515\\
                11	0.0212815047996436\\
                11.5	0.015611380502107\\
                12	0.0111705434093656\\
                12.5	0.00779652916635661\\
                13	0.0053078950961203\\
                13.5	0.0035248234563042\\
                14	0.00228321297903068\\
                14.5	0.00144261184908523\\
                15	0.000889091751770299\\
                };
                \addlegendentry{SNR = 4 dB$, \dot{\mathcal{L}}$}
                
                \addplot [color=red, dashed, mark=diamond, mark size=1.5pt, mark options={solid, red}]
                  table[row sep=crcr]{%
                6	5.9375e-05\\
                7	0.000196875\\
                8	0.00019375\\
                9	0.000359375\\
                10	0.000596875\\
                11	0.000896875\\
                12	0.00129375\\
                13	0.001621875\\
                14	0.002509375\\
                15	0.003375\\
                16	0.004646875\\
                17	0.006140625\\
                18	0.008246875\\
                19	0.010503125\\
                20	0.013371875\\
                21	0.016175\\
                22	0.020071875\\
                23	0.024496875\\
                24	0.029028125\\
                25	0.0343\\
                26	0.0390125\\
                27	0.04461875\\
                28	0.04953125\\
                29	0.052809375\\
                30	0.0561375\\
                31	0.058440625\\
                32	0.0593\\
                33	0.057878125\\
                34	0.05568125\\
                35	0.053225\\
                36	0.049459375\\
                37	0.04413125\\
                38	0.03819375\\
                39	0.0330625\\
                40	0.028190625\\
                41	0.023153125\\
                42	0.019490625\\
                43	0.015584375\\
                44	0.012703125\\
                45	0.009075\\
                46	0.00670625\\
                47	0.004771875\\
                48	0.003184375\\
                49	0.002321875\\
                50	0.001625\\
                51	0.00109375\\
                52	0.00081875\\
                53	0.000559375\\
                54	0.000325\\
                55	0.0002875\\
                56	0.00019375\\
                57	0.000109375\\
                58	0.000115625\\
                };
                \addlegendentry{SNR = 4 dB, $\dot{\mathcal{D}}$, Sim.}
                
                \addplot [color=red]
                  table[row sep=crcr]{%
                12	0.000651405600074829\\
                13	0.00102060449668292\\
                14	0.00156258978039919\\
                15	0.00233780918990035\\
                16	0.00341777550446947\\
                17	0.00488248252848268\\
                18	0.006815395756931\\
                19	0.00929572179319935\\
                20	0.01238804700404\\
                21	0.0161299695500267\\
                22	0.0205189635341182\\
                23	0.0255002955757965\\
                24	0.0309582149929477\\
                25	0.0367127071500184\\
                26	0.042523720266853\\
                27	0.0481039132949153\\
                28	0.0531397014443367\\
                29	0.0573188898431893\\
                30	0.060361770766641\\
                31	0.0620515374586388\\
                32	0.0622595162959578\\
                33	0.0609611950233447\\
                34	0.0582403055152995\\
                35	0.0542800909826562\\
                36	0.049342983514327\\
                37	0.0437418044336423\\
                38	0.0378068833656641\\
                39	0.0318539185183252\\
                40	0.0261569194574469\\
                41	0.0209293425051908\\
                42	0.0163148639787404\\
                43	0.0123875218385818\\
                44	0.00915954260140578\\
                45	0.00659429317855997\\
                46	0.00462153637122539\\
                47	0.00315245447524436\\
                48	0.00209256343904107\\
                49	0.0013514544469574\\
                50	0.000849073902432379\\
                51	0.000518850265405309\\
                52	0.000308336640486054\\
                53	0.000178168851795568\\
                };
                \addlegendentry{SNR = 4 dB, $\dot{\mathcal{D}}$, Eq. (\ref{equ::dualOSD::Extrinsicdensity})}
    
                \end{axis}
                \end{tikzpicture}%

                \vspace{-0.11em}
                \caption{With Gaussian priori density at various SNRs}     
                \vspace{-0.11em}
                \label{Fig::singleDensity::AWGN}

             \end{subfigure}
             \hspace{-0.3em}
             \begin{subfigure}[b]{0.45\columnwidth}
                \centering
                \begin{tikzpicture}
                
                \begin{axis}[%
                width=2.5in,
                height=1.5in,
                at={(0.642in,0.505in)},
                scale only axis,
                xmin=-5,
                xmax=27,
                xlabel style={at={(0.5,1ex)},font=\color{white!15!black},font=\scriptsize},
                xlabel={LLR},
                ymin=0,
                ymax=0.3,
                ylabel style={at={(2ex,0.5)},font=\color{white!15!black},font=\scriptsize},
                ylabel={Densities of LLR},
                axis background/.style={fill=white},
                tick label style={font=\tiny},
                xmajorgrids,
                ymajorgrids,
                minor grid style={dotted},
                major grid style={dotted,black},
                legend style={at={(1,1)}, anchor=north east, legend cell align=left, align=left, draw=white!15!black,font = \tiny,row sep=-3pt}
                ]
                
                \addplot [color=black]
                      table[row sep=crcr]{%
                    -2.1	0.000128270077358207\\
                    -2	0.000207392227312536\\
                    -1.9	0.000330234363869276\\
                    -1.8	0.000517818438366915\\
                    -1.7	0.000799517806145691\\
                    -1.6	0.00121548597113379\\
                    -1.5	0.00181938085142143\\
                    -1.4	0.00268122911087124\\
                    -1.3	0.00389017507258859\\
                    -1.2	0.00555674667619786\\
                    -1.1	0.00781416008773617\\
                    -1	0.0108180943850229\\
                    -0.9	0.0147443228716693\\
                    -0.799999999999999	0.0197836151346892\\
                    -0.699999999999999	0.0261334490650069\\
                    -0.6	0.0339863119051328\\
                    -0.5	0.0435147267286358\\
                    -0.399999999999999	0.0548535980513514\\
                    -0.299999999999999	0.0680809861340316\\
                    -0.199999999999999	0.0831989290806886\\
                    -0.0999999999999996	0.100116352161786\\
                    0	0.118636344839984\\
                    0.0999999999999996	0.138450065892289\\
                    0.199999999999999	0.159139199718317\\
                    0.299999999999999	0.180188218487644\\
                    0.399999999999999	0.201006744831653\\
                    0.5	0.220961154256205\\
                    0.6	0.239413349549825\\
                    0.699999999999999	0.255763554743564\\
                    0.799999999999999	0.269493189106667\\
                    0.9	0.280203538346329\\
                    1	0.287646128428823\\
                    1.1	0.291741436400465\\
                    1.2	0.292583765894815\\
                    1.3	0.290431619735201\\
                    1.4	0.285684512375615\\
                    1.5	0.27884865632924\\
                    1.6	0.270495123849167\\
                    1.7	0.261214775920472\\
                    1.8	0.251574390128058\\
                    1.9	0.242078018251996\\
                    2	0.233136753782733\\
                    2.1	0.225048938026303\\
                    2.2	0.217991559951639\\
                    2.3	0.212022386663362\\
                    2.4	0.207091346341488\\
                    2.5	0.203058973194977\\
                    2.6	0.199719357178179\\
                    2.7	0.196825008382869\\
                    2.8	0.194111292427664\\
                    2.9	0.191318537142813\\
                    3	0.188210461221899\\
                    3.1	0.184588146991315\\
                    3.2	0.180299303627652\\
                    3.3	0.175242998224776\\
                    3.4	0.169370347567788\\
                    3.5	0.162681860946288\\
                    3.6	0.155222216184559\\
                    3.7	0.147073258214942\\
                    3.8	0.138345955918469\\
                    3.9	0.129171961332599\\
                    4	0.119695304598112\\
                    4.1	0.110064642230297\\
                    4.2	0.100426364733777\\
                    4.3	0.0909187675825876\\
                    4.4	0.0816673996924398\\
                    4.5	0.0727816265296445\\
                    4.6	0.0643523808686134\\
                    4.7	0.0564510225423893\\
                    4.8	0.0491291888842721\\
                    4.9	0.0424194894799158\\
                    5	0.036336881803595\\
                    5.1	0.030880557573144\\
                    5.2	0.0260361722188625\\
                    5.3	0.0217782604038236\\
                    5.4	0.0180726974579986\\
                    5.5	0.0148790881053297\\
                    5.6	0.0121529881095282\\
                    5.7	0.00984788964156747\\
                    5.8	0.00791692566556179\\
                    5.9	0.0063142711178736\\
                    6	0.00499623812779518\\
                    6.1	0.00392207837375761\\
                    6.2	0.00305451761689932\\
                    6.3	0.00236005554659274\\
                    6.4	0.00180906860469317\\
                    6.5	0.0013757549025056\\
                    6.6	0.0010379592913649\\
                    6.7	0.000776913721161442\\
                    6.8	0.000576923834194704\\
                    6.9	0.00042502785102352\\
                    7	0.000310648682324217\\
                    7.1	0.000225255218373439\\
                    7.2	0.000162044175090675\\
                    7.3	0.000115649885250311\\
                    };
                    \addlegendentry{$\dot{\mathcal{L}}_i$}
                    
                    \addplot [color=red, dashed, mark=square, mark size = 1pt, mark options={solid, red}]
                      table[row sep=crcr]{%
                    2	9.0625e-05\\
                    3	0.00043125\\
                    4	0.001396875\\
                    5	0.00419375\\
                    6	0.011175\\
                    7	0.0239625\\
                    8	0.0466875\\
                    9	0.08125\\
                    10	0.1230125\\
                    11	0.152771875\\
                    12	0.159175\\
                    13	0.14374375\\
                    14	0.1115\\
                    15	0.072746875\\
                    16	0.039821875\\
                    17	0.0184\\
                    18	0.00671875\\
                    19	0.0021\\
                    20	0.000578125\\
                    21	0.00016875\\
                    };
                    \addlegendentry{$\dot{\mathcal{D}}_i$, Sim.}
                    
                    \addplot [color=red]
                      table[row sep=crcr]{%
                    2.5	0.000127736173870936\\
                    3	0.000270357858175624\\
                    3.5	0.000549111333500976\\
                    4	0.00107020907483734\\
                    4.5	0.00200148550418931\\
                    5	0.00359168913822102\\
                    5.5	0.00618432439382686\\
                    6	0.0102168130386575\\
                    6.5	0.0161937791483941\\
                    7	0.0246245204657565\\
                    7.5	0.035921008490616\\
                    8	0.0502645507460171\\
                    8.5	0.0674645952219999\\
                    9	0.086847234198895\\
                    9.5	0.107217080628592\\
                    10	0.126928361791086\\
                    10.5	0.14407737797887\\
                    11	0.156793510705331\\
                    11.5	0.163570684043811\\
                    12	0.163559311541633\\
                    12.5	0.156740939739743\\
                    13	0.143936135164616\\
                    13.5	0.126642213871782\\
                    14	0.106745106839928\\
                    14.5	0.0861820229761919\\
                    15	0.0666380362330599\\
                    15.5	0.0493402754111912\\
                    16	0.034977737973766\\
                    16.5	0.0237372100377533\\
                    17	0.0154188817653575\\
                    17.5	0.00958516239917744\\
                    18	0.00570175474666021\\
                    18.5	0.00324504156407823\\
                    19	0.00176675453568135\\
                    19.5	0.000920066938510425\\
                    20	0.000458241648738767\\
                    };
                    \addlegendentry{$\dot{\mathcal{D}}_i$, Eq. (\ref{equ::dualOSD::Extrinsicdensity})}
                    
                    \addplot [color=blue, dashed, mark=triangle, mark size = 1.5pt , mark options={solid, blue}]
                      table[row sep=crcr]{%
                    1	0\\
                    2	0.0001\\
                    3	0.0001\\
                    4	0.0002\\
                    5	0.0008\\
                    6	0.0031\\
                    7	0.007\\
                    8	0.0164\\
                    9	0.0346\\
                    10	0.0632\\
                    11	0.0891\\
                    12	0.114\\
                    13	0.1324\\
                    14	0.1414\\
                    15	0.1319\\
                    16	0.1018\\
                    17	0.0737\\
                    18	0.0457\\
                    19	0.0227\\
                    20	0.0132\\
                    21	0.0055\\
                    22	0.0023\\
                    23	0.0005\\
                    24	0.0002\\
                    25	0.0001\\
                    };
                    \addlegendentry{$\dot{\mathcal{P}}_i$, Sim.}
                    
                    \addplot [color=blue]
                      table[row sep=crcr]{%
                    2.5	0.000135788295049398\\
                    3	0.000220506632962729\\
                    3.5	0.000356084816686815\\
                    4	0.000571429888764193\\
                    4.5	0.000910366097413766\\
                    5	0.00143794107515527\\
                    5.5	0.0022482107561018\\
                    6	0.00347288099931977\\
                    6.5	0.00528928067186222\\
                    7	0.00792491720281673\\
                    7.5	0.0116545545722717\\
                    8	0.0167848692554683\\
                    8.5	0.0236220539374818\\
                    9	0.0324200719707332\\
                    9.5	0.0433120602572068\\
                    10	0.0562341698386473\\
                    10.5	0.0708581938784497\\
                    11	0.0865538437312102\\
                    11.5	0.102400500065415\\
                    12	0.117259981660878\\
                    12.5	0.129907284206387\\
                    13	0.139199321537326\\
                    13.5	0.144248313262622\\
                    14	0.144561949508572\\
                    14.5	0.140119126788013\\
                    15	0.131365900397238\\
                    15.5	0.119135873947462\\
                    16	0.10451607482407\\
                    16.5	0.0886887213205834\\
                    17	0.072779782291649\\
                    17.5	0.0577386098272468\\
                    18	0.0442625810724034\\
                    18.5	0.0327699581986777\\
                    19	0.0234153452265956\\
                    19.5	0.0161362699378561\\
                    20	0.0107167895953567\\
                    20.5	0.00685436146268207\\
                    21	0.00421893240116367\\
                    21.5	0.00249736608580678\\
                    22	0.00142082327274366\\
                    22.5	0.000776490783763589\\
                    23	0.000407434709992611\\
                    23.5	0.000205172258961074\\
                    24	9.91188747551174e-05\\
                    24.5	4.59231901980872e-05\\
                    };
                    \addlegendentry{$\dot{\mathcal{P}}_i$, Eq. (\ref{equ::dualOSD::Postdensity})}
                            
                \end{axis}
                \end{tikzpicture}%
                \vspace{-0.11em}
                 \caption{With Non-Gaussian priori density}
                 \vspace{-0.11em}
                \label{Fig::singleDensity::Irregular}
             \end{subfigure}
             \vspace{-0.11em}
             \caption{The extrinsic density with various priori densities. Simulation results are shown by using dashed lines.}
             \vspace{-0.31em}
             \label{Fig::singleDensity}
        \end{figure}

\vspace{-0.5em}
	\section{Density Evolution Framework} \label{Sec::DE}
\vspace{-0.5em}
	Richardson et al. \cite{richardson2001capacity} first proposed the DE method as a tool for performing asymptotic analysis on LDPC codes. To be more precise, when codeword lengths approach infinity, the BG of LDPC codes tends to be cycle-free. After that, during the message-passing decoding, the messages propagated between check nodes and variable nodes are independent between nodes and iterations, and their densities can be theoretically determined. Using DE, the decoder and the code design of LDPC codes have been thoroughly studied \cite{richardson2001capacity,richardson2001design,wang2005density}. In this section, we develop the DE framework for the iterative OSD-based JD.
    
    \vspace{-0.5em}
   \subsection{Bipartite Graph Representation}
    \vspace{-0.5em}
    We first assume that for one specific user $u$, the outputs of PIC , i.e., $\bm{\ell}^{(u)}(t) = [\ell^{(u)}(t)]_1^n $, are i.i.d. variables in a single transmitted block, which is summarized in the following assumption.
    \begin{assumption} \label{Assum::PICbitsindependence}
        In a single transmitted block, elements of the output of PIC for user $u$, i.e., $[\ell^{(u)}(t)]_1^n$, are i.i.d. variables.
    \end{assumption}
    
    In the DS-off phase, Assumption \ref{Assum::PICbitsindependence} naturally holds because $\mathbf{w} = [w]_1^n$ are i.i.d. AWGN variables, $\mathbf{h}$ is fixed in one transmitted block, and (\ref{equ::Receiver::Primean}) and (\ref{equ::Receiver::PICLLR}) are identical for each bit position $i$. In the DS-on phase, Assumption \ref{Assum::PICbitsindependence} also holds because of Proposition \ref{Pro::Dual-SISO::identicaloutput}, i.e., the output of SOSD follows the identical distribution, which is fedback to PIC. Furthermore, the independence between elements $[\ell^{(u)}(t)]_1^n$ is guaranteed by deep interleavers between SOSD and PIC. Under Assumption \ref{Assum::PICbitsindependence}, once again, Assumption \ref{assum::Indep_Decoder_Input} is also validated to be true.
    
    We denote the $i$-th received superposed symbol in one received block as $r_i$, $1\leq 1\leq n$. Under Assumption \ref{Assum::PICbitsindependence}, we next omit the subscript $i$ and simply denote $r_i$ as 
    \begin{equation}\small\label{equ::DE::PIC::receive2}
        r= \mathbf{h} \mathbf{x} + w,
    \end{equation}
    where $\mathbf{x} = [x^{(1)}, x^{(2)}, \ldots, x^{(N_u)}]^{\top}$ is the symbols transmitted by all users in the time slot with respect to $r$, and $w$ is the AWGN variable. Accordingly, we omit $i$ in (\ref{equ::Receiver::PICLLR}) and rewrite it as
    \begin{equation}\small  \label{equ::DE::PICLLR}
        \ell^{(u)}(t) = 2h^{(u)} \frac{r - \sum_{j\neq u}h^{(j)}\mu^{(j)}(t)}{\sum_{j\neq u}(h^{(j)})^2(1-(\mu^{(j)}(t))^2) + \sigma^2},
    \end{equation}
    where
    \begin{equation}\small \label{equ::DE::meanX}
        \mu^{(j)}(t) = \tanh\left(\frac{1}{2}\epsilon^{(j)}(t-1)\right).
    \end{equation}
 
    Based on (\ref{equ::DE::PICLLR}) and (\ref{equ::DE::meanX}), we represent the process of DS-off phase by a BG depicted in Fig. \ref{Fig::BG-off}. Over the BG, the DS-off phase is summarized by the following two steps.
    \begin{itemize}
        \item \textit{Step 1}: (At iteration $t$) Cancellation node $\mathrm{CN}^{(u)}$ computes $\ell^{(u)}(t)$ by (\ref{equ::DE::PICLLR}). $\epsilon^{(u)}(t) \leftarrow \ell^{(u)}(t)$ is passed to its neighbour estimation node $\mathrm{EN}^{(j)}$ with $j=u$.
        \item \textit{Step 2}: Estimation node $\mathrm{EN}^{(j)}$ computes $\mu^{(j)}(t)$ by (\ref{equ::DE::meanX}), which is passed to its neighbour CNs, i.e., $\mathrm{CN}^{(u)}$ with $u\neq j$.
    \end{itemize}
    
    On the other hand, the BG of the DS-on phase is shown in Fig. \ref{Fig::BG-on}. The SOSD decoder of user $u$, represented by the decoding node $\mathrm{DN}^{(u)}$, further processes the outputs of $\mathrm{CN}^{(u)}$ before passing it to $\mathrm{EN}^{(u)}$. As shown in Proposition \ref{Pro::Dual-SISO::identicaloutput}, we regard that SOSD outputs $n$ variables $[\delta^{(u)}(t)]_1^n$ following the identical distribution. Thus, we omit the subscript of $[\delta^{(u)}(t)]_1^n$ and let $\delta^{(u)}(t)$ denote the output of $\mathrm{DN}^{(u)}$. Then, the DS-on process over BG can be summarized by the following three steps.
	
    \begin{itemize}
        \item \textit{Step 1}: (At iteration $t$) $\mathrm{CN}^{(u)}$ computes $\ell^{(u)}(t)$ by (\ref{equ::DE::PICLLR}), which is passed to its decoding node $\mathrm{DN}^{(u)}$.
        \item \textit{Step 2}: $\mathrm{DN}^{(u)}$ obtains $\delta^{(u)}(t)$ from $\ell^{(u)}(t)$ by SOSD. $\epsilon^{(u)}(t) \leftarrow \delta^{(u)}(t)$ is passed to $\mathrm{EN}^{(j)}$ with $j=u$.
        \item \textit{Step 3}: $\mathrm{EN}^{(j)}$ computes $\mu^{(j)}(t)$ by (\ref{equ::DE::meanX}), which is passed to its neighbour CNs, i.e., $\mathrm{CN}^{(u)}$ with $u\neq j$.
    \end{itemize}
            
       \begin{figure} 
             \centering
             \hspace{-0.81em}
             \begin{subfigure}[b]{0.4\columnwidth}
                 \centering
                 \includegraphics[scale = 0.7]{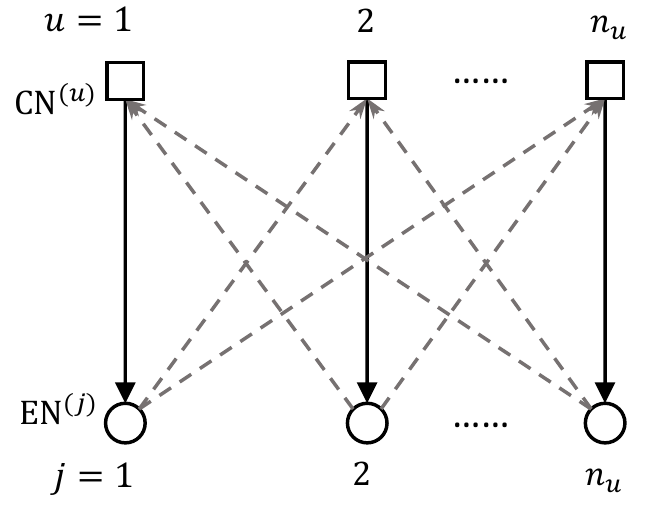}
                \vspace{-0.11em}
                \caption{BG of the DS-off phase}     
                \vspace{-0.11em}
                \label{Fig::BG-off}
             \end{subfigure}
             \hspace{-0.3em}
             \begin{subfigure}[b]{0.4\columnwidth}
                \centering
                \includegraphics[scale = 0.7]{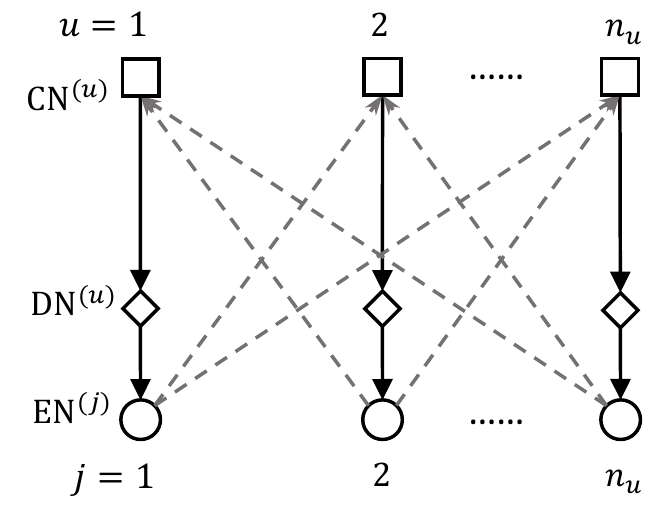}
                \vspace{-0.11em}
                 \caption{BG of the DS-on phase}
                 \vspace{-0.11em}
                \label{Fig::BG-on}
             \end{subfigure}
             \vspace{-0.11em}
             \caption{The BG representation of the iterative JD.}
             \vspace{-0.31em}
             \label{Fig::BG}
        \end{figure}
        
        We are particularly interested in the density of $\ell^{(u)}(t)$ and $\delta^{(u)}(t)$, i.e., the priori and extrinsic density.
    
    \vspace{-0.5em}
    \subsection{Density Evolution of the DS-off Phase}
    \vspace{-0.5em}
    When DS is off, the densities of $\ell^{(u)}(t)$ cannot be characterized by individually examining the density-transform features of cancellation nodes and estimation nodes, due to the correlation. Specifically, $\ell^{(u)}(t)$ and $\ell^{(j)}(t)$ can be correlated for $u\neq j$ with $t>1$, and $\ell^{(u)}(t_1)$ and $\ell^{(u)}(t_2)$ are also correlated for $t_1\neq t_2$. By this regard, we introduce a numerical approach determining the density of $\ell^{(u)}(t)$. We first rewrite (\ref{equ::DE::PICLLR}) as
    \begin{equation}\small  \label{equ::PICout::trans}
        \ell^{(u)}(t) =  2 \frac{\phi^{(u)}\left(\mathbf{x},\bm{\mu}(t)\right) + w}{\psi^{(u)}\left(\mathbf{x},\bm{\mu}(t)\right)},
    \end{equation}
     where $\bm{\mu}(t) = [\mu^{(1)}(t), \mu^{(2)}(t),\ldots, \mu^{(n_u)}(t)]$. Also, $\phi^{(u)}\left(\mathbf{x},\bm{\mu}(t)\right)$ and $\psi^{(u)}\left(\mathbf{x},\bm{\mu}(t)\right)$ are respectively given by
    \begin{equation}\small
        \phi^{(u)}\left(\mathbf{x},\bm{\mu}(t)\right) = h^{(u)} x^{(u)} - \sum_{j\neq u}\frac{h^{(j)}}{h^{(u)}}\left(x^{(j)}-\mu^{(j)}(t)\right)
    \end{equation}
    and
    \begin{equation}\small
        \psi^{(u)}\left(\mathbf{x},\bm{\mu}(t)\right) = \sum_{j\neq u}\frac{(h^{(j)})^2}{h^{(u)}} \left(1-\left(\mu^{(j)}(t)\right)^2\right) + \frac{\sigma^2}{h^{(u)}}.
    \end{equation}
    Then, we have the following recursive relationship regarding $\bm{\mu}(t)$
    \begin{equation}\small
        \bm{\mu}(t) = \tanh\left(\frac{\phi^{(u)}\left(\mathbf{x},\bm{\mu}(t-1)\right) + w}{\psi^{(u)}\left(\mathbf{x},\bm{\mu}(t-1)\right)}\right),
    \end{equation}
    which implies that $\bm{\mu}(t)$ is a deterministic function of $w$ when $\mathbf{x}$ is given. Thus, $\ell^{(u)}(t)$ is also a deterministic function of $w$. Therefore, we represent $ \ell^{(u)}(t)$ as a function $\Lambda_{t}^{(u)}(w;\mathbf{x})$ of $w$ parameterized by $\mathbf{x}$, i.e., 
    \begin{equation}\small
         \ell^{(u)}(t) = \Lambda_{t}^{(u)}(w;\mathbf{x})
    \end{equation}
    where $\Lambda_{t}^{(u)}(w;\mathbf{x})$ is recursively obtained by 
    \begin{equation}\small
        \Lambda_{t}^{(u)}(w;\mathbf{x}) = 2 \frac{\phi^{(u)}\left(\mathbf{x},\tanh\left(\frac{1}{2}\Lambda_{t-1}^{(u)}(w;\mathbf{x})\right)\right) + w}{\psi^{(u)}\left(\mathbf{x},\tanh\left(\frac{1}{2}\Lambda_{t-1}^{(u)}(w;\mathbf{x})\right)\right)},
    \end{equation}
    with $\Lambda_{1}^{(u)}(w;\mathbf{x}) = 2 \frac{\phi^{(u)}\left(\mathbf{x}, \mathbf{0}\right) + w}{\psi^{(u)}\left(\mathbf{x},\mathbf{0}\right)}$. 
    
    We denote the $\mathrm{pdf}$ of $w$ as $f_{W}(w)$, which follows $\mathcal{N}(0,\sigma^2)$. Then, the density of $\ell^{(u)}(t)$ parameterized by $\mathbf{x}$, denoted by $\mathcal{L}_t^{(u)}(\ell;\mathbf{x})$, is determined as
    \begin{equation}\small \label{equ::DSoff::PICDensity::Para}
        \mathcal{L}_{t}^{(u)}(\ell;\mathbf{x}) = \sum_{j} \left|\frac{\partial w_j}{\partial \Lambda_{t}^{(u)}(w_j;\mathbf{x})}\right|\cdot f_{W}(w_j),
    \end{equation}
    where $w_j$ is the $j$-th real solution of $\Lambda_{t}^{(u)}(w;\mathbf{x}) = \ell$ (considering that $\Lambda_{t}^{(u)}(w;\mathbf{x})$ may not be monotonic). Finally, the single-side density of $\ell^{(u)}(t)$ is obtained by deparameterizing $\mathbf{x}$, i.e.,
    \begin{equation}\small  \label{equ::DSoff::PICDensity::singleside}
        \mathcal{L}_{t}^{(u)}(\ell;x^{(u)}) = \frac{1}{|\mathcal{S}|^{(n_u-1)}} \sum_{\mathbf{x}\backslash x^{(u)}  \in \{\mathcal{S}\}^{n_u-1}}  \mathcal{L}_{t}^{(u)}(\ell;\mathbf{x}).
    \end{equation}
    where $\mathcal{S} = \{-1,1\}$ is the set of constellations of symbols transmitted, and $|\mathcal{S}|$ is its cardinality. We take $\dot{\mathcal{L}}_{t}^{(u)}(\ell) = \mathcal{L}_{t}^{(u)}(\ell|x^{(u)}=1) $ and the full density of $\ell^{(u)}(t)$ is directly given by $\mathcal{L}_{t}^{(u)}(\ell) = \frac{1}{2}\dot{\mathcal{L}}_{t}^{(u)}(\ell) + \frac{1}{2}\dot{\mathcal{L}}_{t}^{(u)}(-\ell)$.
    
    \begin{example}
        Consider a two-user NOMA system with $h^{(1)} = 1.225$, $h^{(2)} = 0.707$, and the multi-user SNR is $10$ dB, i.e., $\sigma^2 = 0.2$. We depict the single-side density $\dot{\mathcal{L}}_{t}^{(u)}(\ell)$ at different iterations of the DS-off phase in Fig. \ref{Fig::DE-off}. As shown, (\ref{equ::DSoff::PICDensity::singleside}) precisely describes the density of LLR $\ell^{(u)}(t)$ in the DS-off phase. When $t=1$, the density $\dot{\mathcal{L}}_1^{(u)}$, $u=1,2$, follows a mixture of Gaussian densities due to the inter-user interference. As the number of iteration increases, the quality of LLR $\ell^{(u)}(t)$ is improved, showing a lower BER (represented by the proportion of negative values of LLR). Thus, the DS-off phase can be performed before the DS-on phase to improve the quality of the decoder input. We note that the density $\dot{\mathcal{L}}_1^{u}$ is no more Gaussian for $t> 1$ because of the correlations regarding $\ell^{(u)}(t)$.

    \end{example}

           \begin{figure}
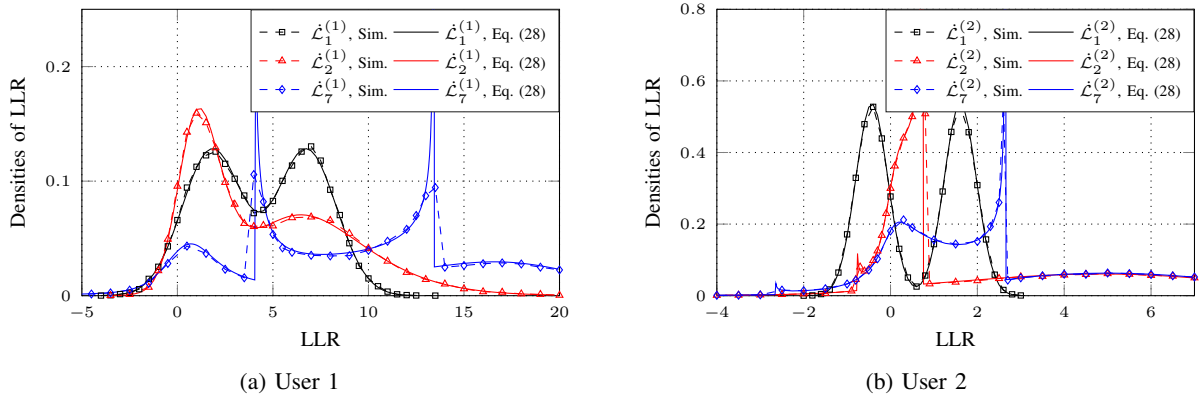
 
             \centering
             \hspace{-0.81em}
             \begin{subfigure}[b]{0.45\columnwidth}
                 \centering
                 
%

                \vspace{-0.11em}
                \caption{User 1}     
                \vspace{-0.11em}
                \label{Fig::DE-off::user1}

             \end{subfigure}
             \hspace{-0.3em}
             \begin{subfigure}[b]{0.45\columnwidth}
                \centering
                %
%
                \vspace{-0.11em}
                 \caption{User 2}
                 \vspace{-0.11em}
                \label{Fig::DE-off::user2}
             \end{subfigure}
             \vspace{-0.11em}
             \caption{Demonstration of the density $\dot{\mathcal{L}}_{t}^{(u)}$ at different iteration $t$.}
             \vspace{-0.11em}
             \label{Fig::DE-off}
        \end{figure}

      \vspace{-0.5em}  
    \subsection{Density Evolution of the DS-on Phase}
    \vspace{-0.5em}
     When DS is on and the BG in Fig. \ref{Fig::BG-on} is considered, the correlation of $\ell^{(u)}(t)$ can be removed by the decoding nodes and interleavers between PIC and SOSD, which is summarized in the following assumption.
    \begin{assumption} \label{Assum::Independence}
        In the DS-on phase, $\ell^{(u)}(t)$ is independent between different users $u$ and iterations $t$.
    \end{assumption}
    Similar assumptions were widely used in the analyses of iterative receivers \cite{liu2019capacity,wang2019near,yuan2014energy,yuan2014achievable} and the iterative turbo decoding \cite{ten2001convergence,divsalar2001iterative}. In the considered JD receiver, Assumption \ref{Assum::Independence} holds because the extrinsic density $\mathcal{D}(\ell)$ output by SOSD is independent to its corresponding priori density $\mathcal{L}(\ell)$. In other words, SOSD blocks the propagation of correlations between random variables. Under Assumption \ref{Assum::Independence}, the density-transform feature of cancellation nodes and estimation nodes can be individually characterized.
    
    \subsubsection{Cancellation Nodes} Recall (\ref{equ::PICout::trans}) and we obtain that 
    \begin{equation}\small  \label{equ::PIC::Density:Parameterized}
        \ell^{(u)}(t) \sim  \mathcal{N}\left(\frac{2\phi^{(u)}\left(\mathbf{x},\bm{\mu}(t)\right)}{\psi^{(u)}\left(\mathbf{x},\bm{\mu}(t)\right)} , \left(\frac{2\sigma}{\psi^{(u)}\left(\mathbf{x},\bm{\mu}(t)\right)}\right)^2 \right),
    \end{equation}
    which indicate that $\ell^{(u)}(t)$ follows a Gaussian distribution parameterized by $\mathbf{x}$ and $\bm{\mu}(t)$. We denote the Gaussian density of (\ref{equ::PIC::Density:Parameterized}) as $\mathcal{L}^{(u)}(\ell;\mathbf{x},\bm{\mu}(t))$, parameterized by $\mathbf{x}$ and $\bm{\mu}(t))$. Thus, $\bm{\mu}(t)$ uniquely determines $\mathcal{L}^{(u)}(\ell;\mathbf{x},\bm{\mu}(t))$ at different iteration $t$.
    
    Assume that the density of $\mu^{(u)}(t)$ is known as $\mathcal{J}_{t}^{(u)} (\mu;x^{(u)})$ parameterized by $x^{(u)}$. Thus, we can obtain $\mathcal{L}_t^{(u)}(\ell;\mathbf{x})$ by deparameterizing $\bm{\mu}(t)$ from $\mathcal{L}^{(u)}(\ell;\mathbf{x},\bm{\mu}(t))$ as follows.
    \begin{equation}\small 
        \mathcal{L}_{t}^{(u)}\left(\ell;\mathbf{x}\right) =  \underbrace{\int_{-1}^{1}\cdots \int_{-1}^{1}}_{n_u} \mathcal{L}^{(u)}\left(\ell;\mathbf{x},\bm{\mu}\right)\left(\prod_{j=1}^{n_u}  \mathcal{J}_{t}^{(j)} (\mu_j;x^{(j)})\right) \prod_{j=1}^{n_u} d \, \mu_j,
    \end{equation}
    where $\bm{\mu} = [\mu_1,\ldots, \mu_{n_u}]$. Finally, deparameterizing $\mathbf{x}$, we obtain
    \begin{equation}\small  \label{equ::DE::PIC::Densityl::deparamu}
        \mathcal{L}_{t}^{(u)}\left(\ell ; x^{(u)}\right) =  \frac{1}{|\mathcal{S}|^{(n_u-1)}} \sum_{\mathbf{x}\backslash x^{(u)}  \in \{\mathcal{S}\}^{n_u-1}} \underbrace{\int_{-1}^{1}\cdots \int_{-1}^{1}}_{n_u} \mathcal{L}^{(u)}\left(\ell;\mathbf{x},\bm{\mu}\right)\left(\prod_{j=1}^{n_u}  \mathcal{J}_{t}^{(j)} (\mu_j;x_j)\right) \prod_{j=1}^{n_u} d \, \mu_j .
    \end{equation}
    Let $\bm{\mathcal{J}}_{t}(\mu;\mathbf{x}) = \left[\mathcal{J}_{t}^{(1)} \left(\mu;x^{(1)}\right),\ldots, \mathcal{J}_{t}^{(n_u)} \left(\mu;x^{(n_u)}\right) \right]$, and we simply denote (\ref{equ::DE::PIC::Densityl::deparamu}) as
    \begin{equation}\small \label{equ::DE::CNtrans}
        \mathcal{L}_{t}^{(u)}\left(\ell; x^{(u)}\right) =  \Gamma_u\left(\bm{\mathcal{J}}_{t}(\mu;\mathbf{x})\right),
    \end{equation}
   where $\Gamma_u(\cdot)$ is referred to as the density-transform feature of the cancellation node $\mathrm{CN}_u$.

    \subsubsection{Decoding Nodes} The LLR $\ell^{(u)}(t)$ is passed to the decoding node $\mathrm{DN}_u$ to produce $\delta^{(u)}(t)$. Let $\mathcal{D}_{t}^{(u)}(\ell)$ denote the density of $\delta^{(u)}(t)$, and then $\mathcal{D}_{t}^{(u)}(\ell)$ can be determined by taking $\mathcal{L}(\ell) = \mathcal{L}_{t}^{(u)}(\ell)$ in (\ref{equ::dualOSD::transform}), i.e.,
    \begin{equation}\small   \label{equ::DE::DNtrans}
        \mathcal{D}_{t}^{(u)}\left(\ell;x^{(u)}\right) = \Delta\left(\mathcal{L}_{t}^{(u)}\left(\ell;x^{(u)} \right)\right),
    \end{equation}
     where $\Delta(\cdot)$ is the density-transform feature of the decoding nodes.

    \subsubsection{Estimation Nodes} 
    
     When estimation node $\mathrm{EN}_j$ receives the message output by $\mathrm{DN}_j$, e.g., $\delta^{(j)}(t)$, it computes $\mu^{(j)}(t+1) = \tanh\left(\frac{1}{2}\delta^{(j)}(t)\right)$. Thus, we directly obtain that 
    \begin{equation}\small \label{equ::DE::ENtrans}
    \begin{split}
        \mathcal{J}_{t+1}^{(j)}\left(\mu;x^{(u)}\right) = & 2\cosh\left(2\tanh^{-1}(\mu)\right)\cdot\mathcal{D}_{t}^{(j)}\left(2\tanh^{-1}(\mu);x^{(j)}\right) = \Theta\left(\mathcal{D}_{t}^{(j)}\left(\mu;x^{(j)}\right)\right),
    \end{split}
    \end{equation}
    where $\Theta(\cdot)$ is referred to as the density-transform feature of the estimation nodes.
    
    We note that $\Theta(\cdot)$ and $\Delta(\cdot)$ are respectively identical for each estimation node and each decoding node, while $\Gamma_u$ are specified for cancellation node $\mathrm{CN}_u$.
    
    \subsubsection{Density Evolution} For $\mathrm{CN}_u$, $1\leq u\leq n_u$, we simply denote
    \begin{equation}\small
        \bm{\mathcal{L}}_{t}\left(\ell;\mathbf{x}\right) =  \bm{\Gamma}\left(\bm{\mathcal{J}}_{t}(\mu;\mathbf{x})\right) = \left[\Gamma_1\left(\bm{\mathcal{J}}_{t}(\mu;\mathbf{x})\right),\Gamma_2\left(\bm{\mathcal{J}}_{t}(\mu;\mathbf{x})\right),\ldots, \Gamma_{n_u}\left(\bm{\mathcal{J}}_{t}(\mu;\mathbf{x})\right) \right],
    \end{equation}
    where $\bm{\mathcal{L}}_{t}\left(\ell;\mathbf{x}\right) = [\mathcal{L}_{t}^{(1)}(\ell;x^{(u)}),\ldots, \mathcal{L}_{t}^{(n_u)}(\ell;x^{(n_u)})]$. Combining (\ref{equ::DE::CNtrans}), (\ref{equ::DE::DNtrans}), and (\ref{equ::DE::ENtrans}), we have the following recursive relationship.
    \begin{equation}\small
       \bm{\mathcal{L}}_{t}\left(\ell;\mathbf{x}\right) = \bm{\Gamma}\left(\Theta\left(\Delta\left(\bm{\mathcal{L}}_{t-1}\left(\ell;\mathbf{x}\right)\right)\right)\right).
    \end{equation}
    Let $\left(f\circ g\right)(\cdot)$ denote the function composition $f(g(\cdot))$, and $f^{t}(\cdot) = (f\circ f^{t-1})(\cdot)$ for an integer $t > 1$. Thus, we have
    \begin{equation}\small \label{equ::DE::densityRecur}
        \bm{\mathcal{L}}_{t}\left(\ell;\mathbf{x}\right) = \left(\bm{\Gamma} \circ \Theta \circ \Delta\right)^t\left(\bm{\mathcal{L}}_{0}(\ell;\mathbf{x})\right) \ \ \text{and} \ \ \bm{\mathcal{L}}_{t}\left(\ell;\mathbf{x}\right) = \left(\bm{\Gamma} \circ \Theta \circ \Delta\right)^{t-t_1}\left(\bm{\mathcal{L}}_{t_1}(\ell;\mathbf{x})\right),
    \end{equation}
    for integers $t$ and $t_1$ with  $t > t_1 > 0$, where $\bm{\mathcal{L}}_{0}(\ell;\mathbf{x})$ is an initialized density. For the $u$-th entry, $\mathcal{L}_{t}^{(u)}\left(\ell;x^{(u)}\right)$, of $\bm{\mathcal{L}}_{t}\left(\ell;\mathbf{x}\right)$, we take the single-side density $\dot{\mathcal{L}}_{t}^{(u)}\left(\ell\right) = \mathcal{L}_{t}^{(u)}\left(\ell|x^{(u)}=1\right)$. 
    
    \vspace{-0.5em}
    \subsection{Density Evolution of the iterative JD}
    \vspace{-0.5em}
    As discussed in Section \ref{sec::Preliminaries}, the iterative OSD-based JD is composed of the DS-off and DS-on phases. In the DS-off phase, the density $\bm{\mathcal{L}}_{t}(\ell)$  can be numerically computed by (\ref{equ::DSoff::PICDensity::Para}), whilst in the DS-on phase, $\bm{\mathcal{L}}_{t}(\ell)$ is determined by (\ref{equ::DE::densityRecur}). Let us consider the iterative JD which has $t_{\mathrm{off}}$ iterations for the DS-off phase, and DS is turned on at iteration $t_{\mathrm{off}}+1$. In this case, each entry of $\bm{\mathcal{L}}_{t}\left(\ell;\mathbf{x}\right)$ is computed according to (\ref{equ::DSoff::PICDensity::Para}) for $t\leq t_{\mathrm{off}} + 1$. Then, for $t > t_{\mathrm{off}} + 1$, $\bm{\mathcal{L}}_{t}\left(\ell;\mathbf{x}\right)$ is obtained as
    \begin{equation}\small
        \bm{\mathcal{L}}_{t}\left(\ell;\mathbf{x}\right) = \left(\bm{\Gamma} \circ \Theta \circ \Delta\right)^{t-t_{\mathrm{off}}-1}\left(\bm{\mathcal{L}}_{t_{\mathrm{off}}+1}(\ell;\mathbf{x})\right).
    \end{equation}
    
     Next, we give numerical examples of DE of the iterative JD
    
    \begin{example} \label{exam::64DEDS-2}
        Consider a two-user NOMA transmission with $h^{(1)} = 1.225$, $h^{(2)} = 0.707$ and multi-user SNR = $8$ dB, i.e., $\sigma^2 = 0.317$. The $(64,30,14)$ eBCH code is employed and the order-3 SOSD is applied in the DS-on phase. We perform $t_{\mathrm{off}}=2$ iterations in the DS-off phase and then turn on DS at the third iteration. In Fig. \ref{Fig::DE-on-2user}, we depict densities $\dot{\mathcal{L}}_{t}^{(u)}(\ell)$ and $\dot{\mathcal{D}}_{t}^{(u)}(\ell)$ for user $u= 1$ and $u= 2$ at different iterations. We demonstrate only the DS-on phase ($t\geq 3$) for the sake of clarity.
        
        As shown, the user $u=1$, with a higher receiving power, converges at iteration $t=5$ in terms of the density of LLR. Upon the convergence, the extrinsic LLR, $\delta^{(1)}(5)$, has the mean value $\mathbb{E}[\delta^{(1)}(5)] = 76$. On the other hand, the user $u=2$ rapidly converges at iteration $t=4$ with $\mathbb{E}[\delta^{(2)}(4)] = 13$. Furthermore, despite that $\ell^{(u)}(t)$ follows a non-Gaussian distribution at the beginning of the DS-on phase (as shown by $\dot{\mathcal{L}}_3^{(1)}$ and $\dot{\mathcal{L}}_3^{(2)}$ in Fig. \ref{Fig::DE-on-2user}), we can observe that both $\ell^{(u)}(t)$ and $\delta^{(u)}(t)$ tend to be Gaussian after a few iterations in the DS-on phase. This observation is consistent with the EXIT-analysis of iterative and turbo decoders for concatenated codes \cite{ten2001convergence,divsalar2001iterative}, i.e., after a few decoding iterations, the extrinsic information passed between decoders tends to be Gaussian.
    \end{example}
    
           \begin{figure}
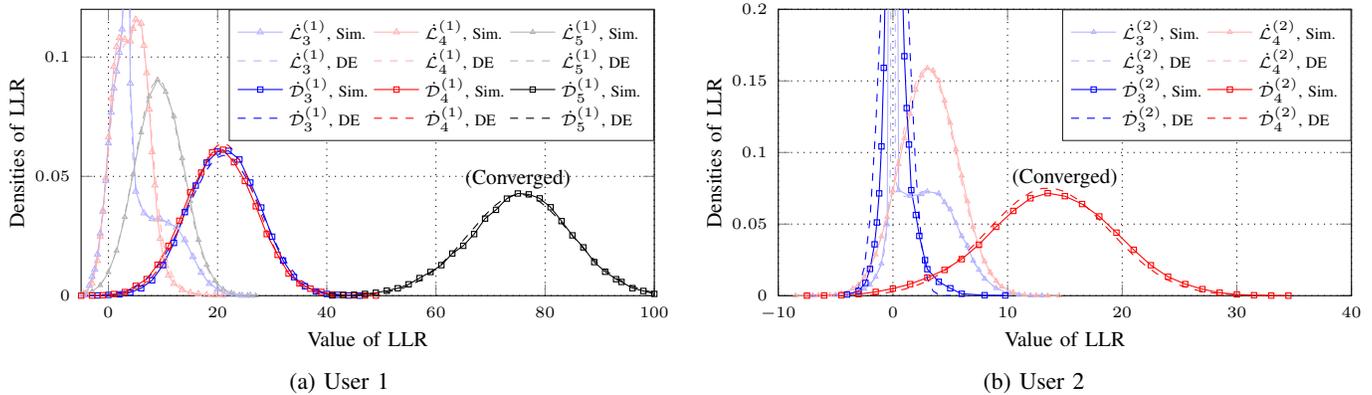
 
             \centering
             \hspace{-0.81em}
             \begin{subfigure}[b]{0.5\columnwidth}
                 \centering
                 
%

                \vspace{-0.11em}
                \caption{User 1}     
                \vspace{-0.11em}
                \label{Fig::DE-on-2user::user1}

             \end{subfigure}
             \hspace{-0.3em}
             \begin{subfigure}[b]{0.5\columnwidth}
                \centering
                %
%
                \vspace{-0.11em}
                 \caption{User 2}
                 \vspace{-0.11em}
                \label{Fig::DE-on-2user::user2}
             \end{subfigure}
             \vspace{-0.11em}
             \caption{The densities $\dot{\mathcal{L}}_{t}^{(u)}$ and $\dot{\mathcal{D}}_{t}^{(u)}$ of different users in the two-user system.}
             \vspace{-0.31em}
             \label{Fig::DE-on-2user}
        \end{figure}
        
        \begin{example} \label{exam::64DEDS-3}
            Consider a three-user NOMA transmission with $h^{(1)} = 1.4411$, $h^{(2)} = 0.8320$, $h^{(3)} = 0.4804$ and multi-user SNR = $12$ dB, i.e., $\sigma^2 = 0.1893$. Specifically, the $(64,30,14)$ eBCH code and the order-3 SOSD are applied, and the DS-off phase has $t_{\mathrm{off}}=2$ iterations. In Fig. \ref{Fig::DE-on-3user}, we demonstrate the density $\dot{\mathcal{D}}_{t}^{(u)}$ of different users in the DS-on phase $(t\geq 3)$. We note that densities $\dot{\mathcal{D}}_3^{(3)}$ and $\dot{\mathcal{D}}_4^{(3)}$ obtained by DE slightly deviates from the simulation, because (\ref{equ::dualOSD::Extrinsicdensity}) overlooks the correlation between $V_0$ and $V_1$. However, the correlation between $V_0$ and $V_1$ may not be negligible for user $u=3$, who is subjected to the high-power interference (recall the trend shown in Fig. \ref{Fig::nc}). We also notice that the deviation of $\dot{\mathcal{D}}_3^{(3)}$ and $\dot{\mathcal{D}}_4^{(3)}$ lead to the deviation of $\dot{\mathcal{D}}_4^{(2)}$ and $\dot{\mathcal{D}}_6^{(2)}$ of user $u=2$ in the subsequent iterations; nevertheless, DE of user $u=1$ is consistent with the simulation results, because user $u=1$ is is less susceptible to interference from user $u=3$.
        \end{example}
        
        \begin{figure} 
             \centering
             \hspace{-0.81em}
             \begin{subfigure}[b]{0.33\columnwidth}
                 \centering
                 
                \begin{tikzpicture}
                
                \begin{axis}[%
                width=1.7in,
                height=1.5in,
                at={(0.785in,0.587in)},
                scale only axis,
                xmin=-10,
                xmax=250,
                xlabel style={at={(0.5,1ex)},font=\color{white!15!black},font=\scriptsize},
                xlabel={Value of LLR},
                ymin= 0,
                ymax=0.07,
                yminorticks=true,
                ylabel style={at={(3ex,0.5)},font=\color{white!15!black},font=\scriptsize},
                ylabel={Densities of LLR},
                axis background/.style={fill=white},
                tick label style={font=\tiny},
                xmajorgrids,
                ymajorgrids,
                yminorgrids,
                minor grid style={dotted},
                major grid style={dotted,black},
                legend style={at={(1,1)}, anchor=north east, legend cell align=left, align=left, draw=white!15!black,font = \tiny,row sep=-1pt,legend columns=2}
                ]
                
               \addplot [color=black]
                  table[row sep=crcr]{%
                -10	10\\
                };
                \addlegendentry{$\dot{\mathcal{D}}_{t}^{(1)}$, Sim.}
                
               \addplot [color=red, dashed]
                  table[row sep=crcr]{%
                -10	10\\
                };
                \addlegendentry{$\dot{\mathcal{D}}_{t}^{(1)}$, DE}
                
               \addplot [only marks, mark=square, mark size = 1pt, mark options={solid, black}]
                  table[row sep=crcr]{%
                -10	10\\
                };
                \addlegendentry{$t=3$}
                
               \addplot [only marks, mark=triangle, mark size = 1.3pt, mark options={solid, black}]
                  table[row sep=crcr]{%
                -10	10\\
                };
                \addlegendentry{$t=4$}
                
               \addplot [only marks, mark=diamond, mark size = 1.3pt, mark options={solid, black}]
                  table[row sep=crcr]{%
                -10	10\\
                };
                \addlegendentry{$t=5$}
                
                   \addplot [only marks, mark=o, mark size = 1.3pt, mark options={solid, black}]
                  table[row sep=crcr]{%
                -10	10\\
                };
                \addlegendentry{$t=6$}

               \addplot [color=black, mark=square, mark size = 1pt, mark options={solid, black}]
                  table[row sep=crcr]{%
                -6	2.08494047494944e-05\\
                -3	4.16988094989888e-05\\
                1	0.000437837499739382\\
                4	0.00228300982006964\\
                8	0.00878802410191189\\
                9	0.0120613806475825\\
                10	0.0155745053478723\\
                11	0.020328169630757\\
                12	0.0243625294497842\\
                13	0.0294706336134103\\
                14	0.0357671538477576\\
                15	0.0423242916414736\\
                16	0.04763046515022\\
                17	0.0528532410399683\\
                18	0.0574713841919813\\
                19	0.0603069032379126\\
                20	0.0640285219856973\\
                22	0.0634343139503367\\
                23	0.061568292225257\\
                24	0.0576173300252278\\
                25	0.0494756374705502\\
                26	0.0425849092008423\\
                27	0.0356837562287597\\
                28	0.0304088568271376\\
                29	0.0245084752830307\\
                30	0.020328169630757\\
                31	0.0162312615974814\\
                33	0.010143235410629\\
                37	0.00182432291558076\\
                41	0.000427412797364635\\
                45	0.000104247023747472\\
                49	1.04247023747472e-05\\
                };
                
                \addplot [color=red, dashed , mark=square, mark size = 1pt, mark options={solid, red}]
                  table[row sep=crcr]{%
                -6	5.10288947263409e-05\\
                -3	0.000212642657438536\\
                0	0.000742899469764047\\
                3	0.00220839304549995\\
                6	0.00559737675465026\\
                9	0.0120950782011555\\
                10	0.0150918196446045\\
                11	0.0184996905604976\\
                12	0.0222778297752929\\
                13	0.0263549265267046\\
                14	0.0306284805886201\\
                15	0.0349669045657738\\
                16	0.0392147535974831\\
                17	0.0432010054055837\\
                18	0.0467499032498419\\
                19	0.0496934768067777\\
                21	0.0532087775726943\\
                23	0.0530194052047204\\
                25	0.0491497723611433\\
                26	0.0460527226973413\\
                27	0.0423728575364517\\
                28	0.0382821018235521\\
                29	0.0339590003888168\\
                30	0.0295759888159278\\
                31	0.0252884220068439\\
                32	0.0212262780501595\\
                33	0.0174890428029429\\
                36	0.00874523835701243\\
                39	0.0036904702429419\\
                42	0.00131150976787941\\
                45	0.000391648013792028\\
                48	9.80657811947587e-05\\
                };

                \addplot [color=black, mark=triangle, mark size = 1.3pt, mark options={solid, black}]
                  table[row sep=crcr]{%
                -6	1.04247023747472e-05\\
                -1	0.00018764464274545\\
                2	0.000406563392615141\\
                5	0.00191814523695348\\
                7	0.00422200446177262\\
                9	0.00698455059108062\\
                11	0.0121864770760795\\
                12	0.01484477618164\\
                13	0.0180138857035632\\
                14	0.0220482455225903\\
                15	0.0265100181389821\\
                16	0.0317015199216062\\
                17	0.0369034464066051\\
                18	0.0431165690219544\\
                19	0.0472134770552301\\
                20	0.0515397285407502\\
                21	0.0570126972874924\\
                23	0.0583887580009591\\
                25	0.0576694535371015\\
                26	0.0551779496695369\\
                27	0.0498613514584159\\
                28	0.0449408919375352\\
                29	0.0399370347976565\\
                30	0.0356316327168859\\
                31	0.029939745220274\\
                32	0.0254362737943832\\
                33	0.0208077059399954\\
                34	0.0161895627879824\\
                36	0.0103308800533745\\
                38	0.00608802618685237\\
                40	0.00284594374830599\\
                42	0.00154285595146259\\
                46	0.00026061755936868\\
                50	4.16988094989888e-05\\
                };
                
                \addplot [color=red, dashed , mark=triangle, mark size = 1.3pt, mark options={solid, red}]
                  table[row sep=crcr]{%
                -6	8.32125976962427e-06\\
                -2	6.93226026024315e-05\\
                0	0.000179356936470746\\
                3	0.00065104406246182\\
                5	0.00140375066090362\\
                6	0.0020056498044333\\
                7	0.00281386784758833\\
                8	0.00387646279007502\\
                9	0.00524385495461204\\
                10	0.00696543969464583\\
                11	0.00908507963367465\\
                12	0.0116356474555951\\
                13	0.0146329844646016\\
                14	0.0180698371667433\\
                15	0.0219104989496493\\
                16	0.0260869776744858\\
                17	0.0304974974181458\\
                18	0.0350080007561403\\
                19	0.0394570439758845\\
                20	0.0436640934490603\\
                21	0.0474407838842438\\
                22	0.0506042541887768\\
                23	0.0529913081674388\\
                25	0.0549605995093529\\
                27	0.0528854518367713\\
                28	0.0504212792536471\\
                29	0.0471567338144426\\
                30	0.0432545822697601\\
                31	0.0389020886799084\\
                32	0.0342963249725685\\
                33	0.0296295944378563\\
                34	0.0250764359708837\\
                35	0.0207833966540178\\
                36	0.0168623464325296\\
                37	0.013387626439457\\
                38	0.0103968443279352\\
                40	0.00585913435362532\\
                43	0.00207941332830278\\
                46	0.000592173568987741\\
                50	7.75045386793033e-05\\
                };
                
                \addplot [color=black, mark=diamond, mark size = 1.3pt, mark options={solid, black}]
                  table[row sep=crcr]{%
                45	1.04247023747472e-05\\
                49	8.33976189979776e-05\\
                53	0.00026061755936868\\
                57	0.000990346725600984\\
                63	0.00322123303379689\\
                65	0.00479536309238371\\
                67	0.00714092112670183\\
                70	0.0114984467193462\\
                72	0.0148135020745158\\
                74	0.0180660092154369\\
                76	0.0228092487959469\\
                78	0.0282196693284407\\
                80	0.030679899088881\\
                82	0.0340262285511749\\
                84	0.0368825970018556\\
                86	0.039259429143298\\
                88	0.0386652211079374\\
                90	0.0366219794424869\\
                92	0.0324103996830891\\
                94	0.0289598231970477\\
                96	0.0247273940329004\\
                98	0.0206930342138732\\
                100	0.0163563580259784\\
                102	0.0133331943373017\\
                104	0.010143235410629\\
                106	0.00666138481746346\\
                108	0.00478493839000897\\
                111	0.00249150386756458\\
                114	0.00116756666597169\\
                117	0.000385713987865646\\
                120	0.000218918749869691\\
                123	6.25482142484832e-05\\
                };
                
                \addplot [color=red, dashed , mark=diamond, mark size = 1.3pt, mark options={solid, red}]
                  table[row sep=crcr]{%
                45	1.03651631391333e-05\\
                48	3.49018267308296e-05\\
                52	0.000151696954982658\\
                56	0.00055580991647319\\
                60	0.00171667754570535\\
                62	0.00282974580171834\\
                64	0.00446947490164221\\
                66	0.00676417776548192\\
                68	0.00980893840517877\\
                70	0.0136293939790783\\
                72	0.0181458655963629\\
                74	0.0231485880699845\\
                76	0.0282954150026557\\
                78	0.033139919580452\\
                80	0.0371903045446507\\
                82	0.0399898714730375\\
                84	0.0412013420959386\\
                86	0.0406736265114097\\
                88	0.0384727926295648\\
                90	0.0348684721322728\\
                92	0.0302795648806716\\
                94	0.0251943390314573\\
                96	0.0200859354496446\\
                98	0.0153432034559687\\
                100	0.011229852529264\\
                104	0.00529162120109448\\
                108	0.00210150002280394\\
                112	0.000703378951704656\\
                115	0.000276651937640107\\
                119	6.86415138560301e-05\\
                123	1.43529651845644e-05\\
                };
                
                \addplot [color=black, mark=o, mark size = 1pt, mark options={solid, black}]
                  table[row sep=crcr]{%
                110	0.000114671726122219\\
                116	0.00018764464274545\\
                124	0.000761003273356546\\
                132	0.00263744970081104\\
                136	0.00394053749765444\\
                140	0.00620269791297458\\
                144	0.00806871963805433\\
                148	0.0120718053499573\\
                152	0.0147509538602673\\
                156	0.0187436148697955\\
                160	0.0226841523674499\\
                164	0.0242687071284115\\
                168	0.0235806767716782\\
                172	0.0229239205220691\\
                176	0.0207243083209974\\
                180	0.018514271417551\\
                184	0.0140003752892855\\
                188	0.0116235431478431\\
                192	0.0083918854116715\\
                196	0.00565018868711298\\
                200	0.00423242916414736\\
                204	0.00219961220107166\\
                208	0.00115714196359694\\
                216	0.000271042261743427\\
                224	3.12741071242416e-05\\
                };
                
                \addplot [color=red, dashed, mark=o, mark size = 1pt, mark options={solid, red}]
                  table[row sep=crcr]{%
                102	2.11418138061165e-05\\
                106	4.8391995080065e-05\\
                110	0.000105127210332816\\
                114	0.000216753980629532\\
                118	0.000424159650831471\\
                122	0.000787774562049226\\
                126	0.00138862468771482\\
                130	0.00232315475212773\\
                134	0.00368877127896025\\
                138	0.005558986818941\\
                142	0.00795096588192716\\
                146	0.0107933053716017\\
                150	0.0139059085315673\\
                154	0.0170041372343622\\
                158	0.0197342289931495\\
                162	0.0217368222766661\\
                166	0.0227238688863777\\
                170	0.0225464839219808\\
                174	0.0212317456402784\\
                178	0.0189759230348948\\
                182	0.0160964626765129\\
                186	0.0129589051966245\\
                190	0.00990185375208241\\
                194	0.0071808363879479\\
                198	0.00494246821180653\\
                202	0.0032286656036592\\
                206	0.00200176263463618\\
                210	0.0011779108408262\\
                214	0.000657843500727795\\
                218	0.000348692886003892\\
                222	0.000175417902920271\\
                226	8.37558348147048e-05\\
                230	3.79547869271992e-05\\
                };
                
                 \node[] at (axis cs: 170,0.03) {\scriptsize (Converged)};
          
                \end{axis}
                \end{tikzpicture}%

                \vspace{-0.11em}
                \caption{User 1}     
                \vspace{-0.11em}
                \label{Fig::DE-on-3user::user1}

             \end{subfigure}
             \begin{subfigure}[b]{0.33\columnwidth}
                 \centering
                 
                \begin{tikzpicture}
                
                \begin{axis}[%
                width=1.7in,
                height=1.5in,
                at={(0.785in,0.587in)},
                scale only axis,
                xmin=-10,
                xmax=80,
                xlabel style={/pgf/number format/fixed,
                /pgf/number format/precision = 0,
                at={(0.5,1ex)},font=\color{white!15!black},font=\scriptsize},
                xlabel={Value of LLR},
                ymin= 0,
                ymax=0.15,
                yminorticks=true,
                ylabel style={
                at={(2ex,0.5)},font=\color{white!15!black},font=\scriptsize},
                ylabel={Densities of LLR},
                yticklabel style={
                        /pgf/number format/fixed,
                        /pgf/number format/precision=5
                },
                axis background/.style={fill=white},
                tick label style={font=\tiny},
                xmajorgrids,
                ymajorgrids,
                yminorgrids,
                minor grid style={dotted},
                major grid style={dotted,black},
                legend style={at={(1,1)}, anchor=north east, legend cell align=left, align=left, draw=white!15!black,font = \tiny,row sep=-1pt,legend columns=2}
                ]
                
               \addplot [color=black]
                  table[row sep=crcr]{%
                -10	10\\
                };
                \addlegendentry{$\dot{\mathcal{D}}_{t}^{(2)}$, Sim.}
                
               \addplot [color=red, dashed]
                  table[row sep=crcr]{%
                -10	10\\
                };
                \addlegendentry{$\dot{\mathcal{D}}_{t}^{(2)}$, DE}
                
               \addplot [only marks, mark=square, mark size = 1pt, mark options={solid, black}]
                  table[row sep=crcr]{%
                -10	10\\
                };
                \addlegendentry{$t=3$}
                
               \addplot [only marks, mark=triangle, mark size = 1.3pt, mark options={solid, black}]
                  table[row sep=crcr]{%
                -10	10\\
                };
                \addlegendentry{$t=4$}
                
               \addplot [only marks, mark=diamond, mark size = 1.3pt, mark options={solid, black}]
                  table[row sep=crcr]{%
                -10	10\\
                };
                \addlegendentry{$t=5$}
                
                   \addplot [only marks, mark=o, mark size = 1.3pt, mark options={solid, black}]
                  table[row sep=crcr]{%
                -10	10\\
                };
                \addlegendentry{$t=6$}
                
                \addplot [color=black, mark=square, mark size = 1pt, mark options={solid, black}]
                  table[row sep=crcr]{%
                -5.4	5.19167670390829e-05\\
                -3.6	0.000674917971508078\\
                -2.8	0.00269967188603231\\
                -2.6	0.00612617851061179\\
                -2.4	0.00939693483407401\\
                -2.2	0.0133426091290443\\
                -2	0.0205590397474768\\
                -1.8	0.0285542218714956\\
                -1.6	0.0439215849150642\\
                -1.4	0.0610022012709225\\
                -1.2	0.0972920214312414\\
                -1	0.142615359056361\\
                -0.8	0.203306059725049\\
                -0.6	0.295873655355734\\
                -0.4	0.421096897454002\\
                -0.2	0.601144245545541\\
                -0	0.73991776384101\\
                0.2	0.596264069443868\\
                0.4	0.444667109689745\\
                0.6	0.319028533455165\\
                0.8	0.226201353989284\\
                1	0.167691157536238\\
                1.2	0.129324666694356\\
                1.4	0.0974477717323587\\
                1.6	0.0746563110022013\\
                2	0.0432985837105952\\
                2.4	0.0224280433608838\\
                2.8	0.0158346139469203\\
                3.2	0.00944885160111309\\
                3.6	0.00680109648211987\\
                4.4	0.00207667068156332\\
                5.6	0.000986418573742576\\
                6.8	0.000207667068156332\\
                8.2	5.19167670390829e-05\\
                };
                
                \addplot [color=red, dashed, mark=square, mark size = 1pt, mark options={solid, red}]
                  table[row sep=crcr]{%
                -3.5	7.21241672397167e-05\\
                -2.5	0.00528328628266846\\
                -2.2	0.0145106242178717\\
                -2	0.026501328165214\\
                -1.8	0.045719069184544\\
                -1.7	0.0587801003145045\\
                -1.6	0.0745032665925324\\
                -1.5	0.0930962859900209\\
                -1.4	0.114683633698685\\
                -1.3	0.139278019682869\\
                -1.2	0.16675381180739\\
                -1.1	0.19682533804645\\
                -1	0.229033120698339\\
                0	0.476139070750745\\
                1	0.23810813215682\\
                1.1	0.205420857297379\\
                1.2	0.174713673533893\\
                1.3	0.146494493999291\\
                1.4	0.121095431073838\\
                1.5	0.0986838902996611\\
                1.6	0.0792824048360714\\
                1.7	0.0627941833607495\\
                1.8	0.0490313744373665\\
                1.9	0.037743378032165\\
                2	0.0286430667556087\\
                2.5	0.00582229207727113\\
                3	0.000828838051139531\\
                4	5.76932129220381e-06\\
                };
                
                \addplot [color=red, dashed, mark=diamond, mark size = 1.3pt, mark options={solid, red}]
                  table[row sep=crcr]{%
                -4	8.4486549388803e-05\\
                -1	0.000723191741952838\\
                2	0.00415771331309267\\
                3.5	0.00858684341927515\\
                5.5	0.0193455365195147\\
                6.5	0.0271732365233783\\
                7.5	0.0365165710579423\\
                8.5	0.046948490273293\\
                9.5	0.0577467523380748\\
                10.5	0.0679505664191961\\
                11.5	0.0764881583293556\\
                12.5	0.0823556027831691\\
                13.5	0.0848071343771016\\
                14.5	0.0835078521149889\\
                15.5	0.0786063808582881\\
                16.5	0.0707064844326226\\
                17.5	0.0607454796190671\\
                19	0.044317028067774\\
                21	0.0246054342411934\\
                24	0.00697433374374016\\
                27	0.00121627313719153\\
                30	0.000125760123014279\\
                };
                                
                \addplot [color=black, mark=diamond, mark size = 1.3pt, mark options={solid, black}]
                  table[row sep=crcr]{%
                -7	0.000103833534078166\\
                -5	0.000259583835195415\\
                -3	0.00160941977821157\\
                -1	0.00363417369273581\\
                1	0.00773559828882336\\
                3	0.0147962786061386\\
                5	0.0225566308094862\\
                7	0.0298002242804336\\
                9	0.0410142459608755\\
                11	0.0493209286871288\\
                12	0.0532666029820991\\
                14	0.0600157826971799\\
                16	0.0608464509698052\\
                18	0.0559167670390829\\
                20	0.0450118370228849\\
                22	0.0310462266893716\\
                24	0.0211820409519458\\
                26	0.0139136935664742\\
                28	0.00721643061843253\\
                30	0.00332267309050131\\
                32	0.00171325331228974\\
                36	0.000155750301117249\\
                };
                
                \addplot [color=red, dashed, mark=triangle, mark size = 1.3pt, mark options={solid, red}]
                  table[row sep=crcr]{%
                -7	9.37166815574403e-05\\
                -4	0.000816932220351794\\
                -2	0.00275195491792541\\
                0	0.00771865513126848\\
                1	0.0120686590176054\\
                2	0.0180255234497085\\
                3	0.0257174655287049\\
                4	0.0350493260908202\\
                5	0.0456291048050002\\
                6	0.0567432167615231\\
                7	0.0674051978949165\\
                8	0.0764848521092574\\
                9	0.0828992105450875\\
                10	0.0858222138446283\\
                11	0.0848573467735712\\
                12	0.0801236027512593\\
                13	0.0722295082067498\\
                14	0.0621436973586219\\
                15	0.0510005810771295\\
                16	0.0398948367834345\\
                17	0.0297146913970251\\
                18	0.0210454062160483\\
                19	0.0141501380973142\\
                22	0.00308482605104165\\
                26	0.000170885599020551\\
                30	3.10833295939563e-06\\
                };
                
                \addplot [color=black, mark=triangle, mark size = 1.3pt, mark options={solid, black}]
                  table[row sep=crcr]{%
                -5	0.000207667068156332\\
                -3	0.000415334136312664\\
                -1	0.00129791917597707\\
                1	0.00321883955642314\\
                3	0.00799518212401877\\
                5	0.0138098600323961\\
                7	0.0244527972754081\\
                9	0.0403912447564065\\
                10	0.0467770071022137\\
                11	0.0537338538854508\\
                12	0.0601084437429912\\
                13	0.0681147983552768\\
                14	0.075123561905553\\
                16	0.0748639780703576\\
                18	0.0667649624122607\\
                19	0.0562777754703659\\
                20	0.049476678988246\\
                21	0.0388856585122731\\
                22	0.0329671470698177\\
                23	0.0268409685592059\\
                24	0.0208705403497113\\
                26	0.00934501806703493\\
                28	0.00332267309050131\\
                30	0.00129791917597707\\
                32	0.000155750301117249\\
                35	5.19167670390829e-05\\
                };

                \addplot [color=black, mark=o, mark size = 1.1pt, mark options={solid, black}]
                  table[row sep=crcr]{%
                20	0.000103833534078166\\
                22	0.00088258503966441\\
                25	0.00114216887485982\\
                27	0.00197283714748515\\
                29	0.00399759106200939\\
                31	0.00586659467541637\\
                33	0.00867010009552685\\
                35	0.0144328612368651\\
                37	0.0191572870374216\\
                39	0.026321800888815\\
                41	0.0359264027910454\\
                43	0.0424679154379698\\
                45	0.0467250903351746\\
                47	0.0498923412385264\\
                49	0.049253677783777\\
                51	0.0479191759770736\\
                53	0.0431428334094779\\
                55	0.0340573991776384\\
                57	0.0251277152469161\\
                59	0.0193130373385389\\
                61	0.0126157743904972\\
                63	0.00778751505586244\\
                65	0.00410142459608755\\
                67	0.00197283714748515\\
                69	0.000778751505586244\\
                71	0.000311500602234498\\
                73	0.000311500602234498\\
                75	0.000103833534078166\\
                };
                
                \addplot [color=red, dashed, mark=o, mark size = 1.1pt, mark options={solid, red}]
                  table[row sep=crcr]{%
                20	4.18291834319194e-05\\
                24	0.000302366920333698\\
                28	0.00159794947367527\\
                32	0.00616666992114503\\
                36	0.0173473932257903\\
                38	0.0258238292025454\\
                40	0.0354865919241084\\
                42	0.0449974862636308\\
                44	0.0526258483721815\\
                46	0.0567399485992202\\
                48	0.0563686663520097\\
                50	0.0515722830169652\\
                52	0.043429944668128\\
                54	0.0336447934717277\\
                56	0.0239642157198517\\
                58	0.0156851572644572\\
                61	0.00707983450394056\\
                64	0.00263431376952648\\
                68	0.000520579622340604\\
                72	7.2510408393817e-05\\
                76	7.09684634676868e-06\\
                80	4.86764110150935e-07\\
                };
                
                 \node[] at (axis cs: 50,0.065) {\scriptsize (Converged)};
          
                \end{axis}
                \end{tikzpicture}%

                \vspace{-0.11em}
                \caption{User 2}     
                \vspace{-0.11em}
                \label{Fig::DE-on-3user::user2}

             \end{subfigure}             
             \begin{subfigure}[b]{0.33\columnwidth}
                \centering
                \begin{tikzpicture}
                
                \begin{axis}[%
                width=1.7in,
                height=1.5in,
                at={(0.642in,0.505in)},
                scale only axis,
                xmin=-7,
                xmax=25,
                xlabel style={at={(0.5,1ex)},font=\color{white!15!black},font=\scriptsize},
                xlabel={Value of LLR},
                ymin=0,
                ymax=0.2,
                ylabel style={at={(2ex,0.5)},font=\color{white!15!black},font=\scriptsize},
                ylabel={Densities of LLR},
                yticklabel style={
                        /pgf/number format/fixed,
                        /pgf/number format/precision=5
                },
                axis background/.style={fill=white},
                tick label style={font=\tiny},
                xmajorgrids,
                ymajorgrids,
                minor grid style={dotted},
                major grid style={dotted,black},
                legend style={at={(1,1)}, anchor=north east, legend cell align=left, align=left, draw=white!15!black,font = \tiny,row sep=-2pt,legend columns=1}
                ]
                
                \addplot [color=black]
                  table[row sep=crcr]{%
                -10	10\\
                };
                \addlegendentry{$\dot{\mathcal{D}}_{t}^{(3)}$, Sim.}
                
               \addplot [color=red, dashed]
                  table[row sep=crcr]{%
                -10	10\\
                };
                \addlegendentry{$\dot{\mathcal{D}}_{t}^{(3)}$, DE}
                
               \addplot [only marks, mark=square, mark size = 1pt, mark options={solid, black}]
                  table[row sep=crcr]{%
                -10	10\\
                };
                \addlegendentry{$t=3$}
                
               \addplot [only marks, mark=triangle, mark size = 1.3pt, mark options={solid, black}]
                  table[row sep=crcr]{%
                -10	10\\
                };
                \addlegendentry{$t=4$}
                
               \addplot [only marks, mark=diamond, mark size = 1.3pt, mark options={solid, black}]
                  table[row sep=crcr]{%
                -10	10\\
                };
                \addlegendentry{$t=5$}

                \addplot [color=black, mark=square, mark size = 1pt, mark options={solid, black}]
                  table[row sep=crcr]{%
                -2.6	0.000104488840591825\\
                -1.6	0.000104488840591825\\
                -1.2	0.000104488840591825\\
                -0.8	0.00020897768118365\\
                -0.5	0.00700075231965226\\
                -0.4	0.0355262058012204\\
                -0.3	0.169794365961715\\
                -0.2	0.727660285881468\\
                -0.1	2.24389785170944\\
                0	3.58480314302433\\
                0.1	2.26782579620497\\
                0.2	0.744796455738527\\
                0.3	0.171675165092368\\
                0.4	0.0347947839170777\\
                0.5	0.00626933043550949\\
                0.9	0.00020897768118365\\
                1.3	0.000104488840591825\\
                };
                
                \addplot [color=black, mark=triangle, mark size = 1.3pt, mark options={solid, black}]
                  table[row sep=crcr]{%
                -4.8	0.000104488840591825\\
                -4	0.000835910724734598\\
                -3.5	0.00303017637716292\\
                -3	0.00647830811669314\\
                -2.7	0.0126431497116108\\
                -2.4	0.0251818105826298\\
                -2.2	0.0329139847864248\\
                -2	0.0507815765276268\\
                -1.8	0.0650965476887068\\
                -1.6	0.0921591574019895\\
                -1.4	0.127267407840843\\
                -1.2	0.171570676251776\\
                -1	0.220993897851709\\
                -0.8	0.268327342639806\\
                -0.6	0.33499122293739\\
                -0.3	0.413253364540667\\
                -0.1	0.430807489760094\\
                0.1	0.439584552369807\\
                0.3	0.423075315556299\\
                0.5	0.364457075984285\\
                0.7	0.315033854384352\\
                0.9	0.259445791189501\\
                1.1	0.200305107414528\\
                1.3	0.150881885814595\\
                1.5	0.115146702332191\\
                1.7	0.0842180055170108\\
                1.9	0.0641561481233804\\
                2.1	0.0451391791356683\\
                2.3	0.0327050071052412\\
                2.4	0.0263311878291398\\
                2.6	0.0179720805817939\\
                2.8	0.0121207055086517\\
                3	0.00898604029089693\\
                3.2	0.00564239739195854\\
                3.5	0.00240324333361197\\
                4	0.000835910724734598\\
                5	0.00020897768118365\\
                };
                
                \addplot [color=black, mark=diamond, mark size = 1.3pt, mark options={solid, black}]
                  table[row sep=crcr]{%
                -9	0.000104488840591825\\
                -7	0.00020897768118365\\
                -5	0.00156733260887737\\
                -4	0.00365710942071387\\
                -3	0.00700075231965226\\
                -2	0.0120162166680599\\
                -1	0.0218830560896096\\
                -0	0.0341531388447714\\
                1	0.0456321992811168\\
                2	0.054304020730586\\
                3	0.063618991891666\\
                5	0.0767992978349912\\
                7	0.0827992978349912\\
                9	0.0767992978349912\\
                11	0.0591406837749728\\
                13	0.0365710942071387\\
                15	0.0194349243500794\\
                17	0.0100309286968152\\
                19	0.00334364289893839\\
                21	0.000313466521775474\\
                23	0.000104488840591825\\
                25	0.00020897768118365\\
                };
                
                \addplot [color=red, dashed, mark=triangle, mark size = 1.3pt, mark options={solid, red}]
                  table[row sep=crcr]{%
                -3	3.15908407615046e-05\\
                -2.7	0.000204000129149956\\
                -2.4	0.00108288086382291\\
                -2.2	0.00295523819958965\\
                -2	0.00739219863890652\\
                -1.8	0.0169482198517559\\
                -1.7	0.0248377822729282\\
                -1.6	0.0356158959155145\\
                -1.5	0.04997091044657\\
                -1.4	0.0686014004824531\\
                -1.3	0.09214908611274\\
                -1.2	0.121113181892065\\
                -1.1	0.155752177370796\\
                -1	0.195983346996758\\
                -0.9	0.241294012171604\\
                -0.799999999999999	0.290680735459242\\
                -0.699999999999999	0.342632288388374\\
                -0.6	0.395168805620949\\
                -0.5	0.44594300167593\\
                -0.399999999999999	0.492400383191045\\
                -0.299999999999999	0.531985435674413\\
                -0.199999999999999	0.562371669310927\\
                -0.0999999999999996	0.581687131872086\\
                0	0.588705115721643\\
                0.100000000000001	0.582973071021615\\
                0.200000000000001	0.564860894431832\\
                0.300000000000001	0.535521431168181\\
                0.4	0.496769052942909\\
                0.5	0.45089408815618\\
                0.600000000000001	0.400439472267195\\
                0.700000000000001	0.34796979759806\\
                0.800000000000001	0.29586156603696\\
                0.9	0.246137555913433\\
                1	0.200359319028483\\
                1.1	0.159581866756411\\
                1.2	0.12436548450287\\
                1.3	0.0948327887766342\\
                1.4	0.0707553867557057\\
                1.5	0.0516538654821762\\
                1.6	0.036896780804161\\
                1.7	0.0257879284248881\\
                1.8	0.0176354589796535\\
                1.9	0.0118004728220189\\
                2.1	0.00494939010516804\\
                2.3	0.0019027130042432\\
                2.6	0.000385188373321078\\
                2.7	0.000216533182734604\\
                3	3.37545493751323e-05\\
                };

                \addplot [color=red, dashed, mark=square, mark size = 1pt, mark options={solid, red}]
                  table[row sep=crcr]{%
                -1.7	8.04376712339213e-05\\
                -1.5	0.000649674787742184\\
                -1.3	0.00404304272826347\\
                -1.2	0.00914650998136995\\
                -1.1	0.019386367558189\\
                -1	0.0384974018843762\\
                -0.9	0.0716242928898616\\
                -0.799999999999999	0.124848469403807\\
                -0.699999999999999	0.203891934210479\\
                -0.6	0.311968552279502\\
                -0.5	0.44721415138491\\
                -0.399999999999999	0.600639938042243\\
                -0.299999999999999	0.75579988760592\\
                -0.199999999999999	0.891032178310244\\
                -0.0999999999999996	0.984178564483898\\
                0	1.01847028802407\\
                0.0999999999999996	0.987453774278228\\
                0.199999999999999	0.896972509263672\\
                0.299999999999999	0.763370617819132\\
                0.399999999999999	0.608675323675808\\
                0.5	0.454705177352187\\
                0.6	0.318249735498227\\
                0.699999999999999	0.208689285402548\\
                0.799999999999999	0.128211269631339\\
                0.9	0.0737982729040966\\
                1	0.0397978988471424\\
                1.1	0.0201079611698903\\
                1.2	0.00951852994052174\\
                1.3	0.00422148908884282\\
                1.5	0.000682871673276344\\
                1.7	8.51115143839765e-05\\
                };
                
                \addplot [color=red, dashed,mark=diamond, mark size = 1.3pt, mark options={solid, red}]
                  table[row sep=crcr]{%
                -10	6.58872745997817e-05\\
                -9	0.000148473571491758\\
                -8	0.0003187081719789\\
                -7	0.000651679239153504\\
                -6	0.00126932122892646\\
                -5	0.00235508349656724\\
                -4	0.00416234609052617\\
                -3	0.00700756659018115\\
                -2	0.0112381147528688\\
                -1	0.0171678881661232\\
                0	0.0249825908954116\\
                1	0.0346302209503981\\
                2	0.0457267388281767\\
                3	0.0575151620478193\\
                4	0.0689114679285229\\
                5	0.0786497324000036\\
                6	0.0855062921898818\\
                7	0.0885501271786406\\
                8	0.087348768571572\\
                9	0.0820662380755031\\
                10	0.0734218394375812\\
                11	0.0625243231505317\\
                12	0.0506360664195492\\
                13	0.0389390322047504\\
                14	0.0283618414446021\\
                15	0.0194944026335538\\
                16	0.0125830835591069\\
                17	0.00758210233776356\\
                18	0.00423688046182809\\
                19	0.00218064586069766\\
                20	0.00102687576498258\\
                21	0.000439719705188415\\
                22	0.00017029007786672\\
                23	5.93628023864247e-05\\
                };
                
                 \node[] at (axis cs: 17.5,0.05) {\scriptsize (Converged)};
                \end{axis}
                \end{tikzpicture}%
                \vspace{-0.11em}
                 \caption{User 3}
                 \vspace{-0.11em}
                \label{Fig::DE-on-3user::user3}
             \end{subfigure}
             \hspace{-0.81em}
             \vspace{-0.11em}
             \caption{The density $\dot{\mathcal{D}}_{t}^{(u)}$ of different users in the three-user system.}
             \vspace{-0.31em}
             \label{Fig::DE-on-3user}
        \end{figure}
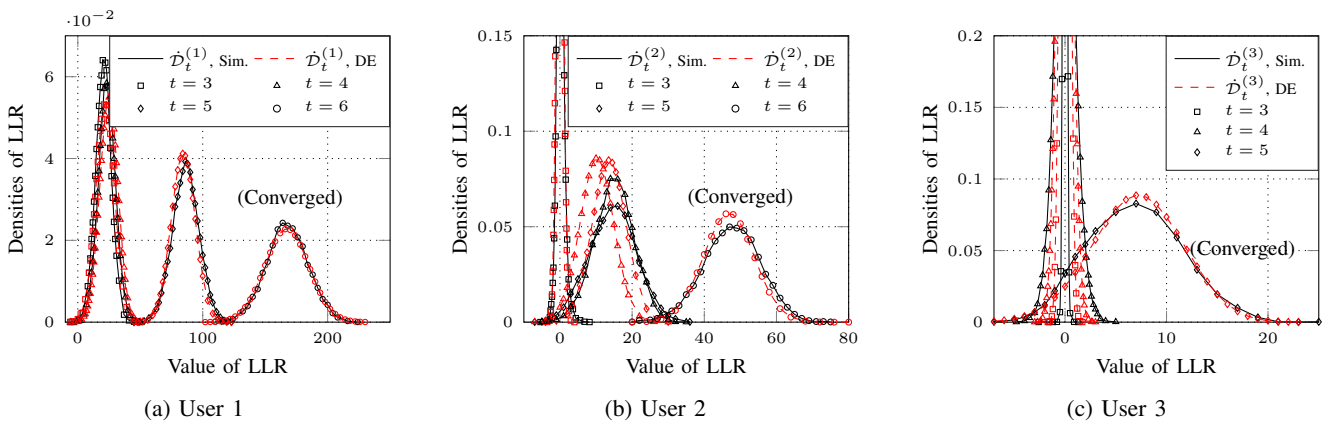
        
    \vspace{-0.5em}    
    \section{Analytical Examples based on the Density Evolution Framework}   \label{Sec::Example}
    \vspace{-0.5em}
    
    Using the DE framework introduced in Section \ref{Sec::DE}, one can analyze the BER, convergence behavior, etc., of the OSD-based JD, and accordingly conduct further optimizations. In this section, we will discuss some preliminary analytical examples and outline future works based on the DE framework due to the space limit.
    
    \vspace{-0.5em}
    \subsection{Bit Error Rate}
    \vspace{-0.5em}
    Let $\mathrm{P}_b^{(u)}(t)$ denote the BER of the $u$-th user at iteration $t$. In the DS-off phase, $\mathrm{P}_b^{(u)}(t)$ is given by
    \begin{equation}\small\label{equ::BER::DS-off}
        \mathrm{P}_b^{(u)}(t) = \int_{-\infty}^{0} \dot{\mathcal{L}}_t^{(u)}(\ell) d \ell,
    \end{equation}
    where $\dot{\mathcal{L}}_t^{(u)}(\ell)$ is given by (\ref{equ::DSoff::PICDensity::singleside}). On the other hand, in the DS-on phase, $\mathrm{P}_b^{(u)}(t)$ can be determined as 
    \begin{equation}\small \label{equ::BER::DS-on}
        \mathrm{P}_b^{(u)}(t) = \int_{-\infty}^{0} \left(\dot{\mathcal{L}}_t^{(u)}\otimes\dot{\mathcal{D}}_t^{(u)}\right)(\ell) d \ell = \int_{-\infty}^{0} \left(\dot{\mathcal{L}}_t^{(u)}\otimes \Delta\left(\dot{\mathcal{L}}_t^{(u)}\right)\right)(\ell) d \ell,
    \end{equation}
    where $\dot{\mathcal{L}}_t^{(u)}$ is obtained by the DE framework, i.e., (\ref{equ::DE::densityRecur}). We note that (\ref{equ::BER::DS-off}) and (\ref{equ::BER::DS-on}) are valid under the all-zero transmission assumption over symmetric channels without losing generality, which was widely considered in DE techniques \cite{richardson2001capacity}.
    
    We compare the BER obtained from simulation and DE for the two-user system described in Example \ref{exam::64DEDS-2}. In Fig. \ref{Fig::BER::SNR}, we illustrate BER at different SNRs, when the iteration number is set to $t=2$ and $t=5$. Specifically, when $t=2$ ($t\leq t_{\mathrm{off}}$), the BER from DE is obtained by (\ref{equ::BER::DS-off}), while when $t=5$, it is obtained by (\ref{equ::BER::DS-on}). As shown, the proposed DE framework can accurately estimate the BER performance of the iterative JD, particularly at moderate-to-high SNRs. In Fig. \ref{Fig::BER::iter}, we depict the BER at different iterations when SNR is $8$ dB. It is seen that the BER decreases as the number of iterations increases. It is worth noting that there is a significant gain in the BER performance after DS is turned on, i.e., after the JD switches from the DS-off phase to the DS-on phase.
    
           \begin{figure} 
             \centering
             \hspace{-0.81em}
             \begin{subfigure}[b]{0.48\columnwidth}
                 \centering
                 
                \begin{tikzpicture}
                
                \begin{axis}[%
                width=2.5in,
                height=1.5in,
                at={(0.785in,0.587in)},
                scale only axis,
                xmin=2,
                xmax=22,
                xlabel style={at={(0.5,1ex)},font=\color{white!15!black},font=\scriptsize},
                xlabel={SNR (dB)},
                ymode=log,
                ymin=1e-05,
                ymax=1,
                yminorticks=true,
                ylabel style={at={(1ex,0.5)},font=\color{white!15!black},font=\scriptsize},
                ylabel={Bit Error Rate},
                axis background/.style={fill=white},
                tick label style={font=\tiny},
                xmajorgrids,
                ymajorgrids,
                yminorgrids,
                minor grid style={dotted},
                major grid style={dotted,black},
                legend style={at={(1,1)}, anchor=north east, legend cell align=left, align=left, draw=white!15!black,font = \tiny,row sep=-1pt,legend columns=1}
                ]
                
                \addplot [color=black, mark=square, mark size = 1.5pt, mark options={solid, black}]
                  table[row sep=crcr]{%
                2	0.13421875\\
                3	0.085390625\\
                4	0.04746875\\
                5	0.00862847222222222\\
                6	0.000199652777777778\\
                7	0\\
                8	0\\
                };
                \addlegendentry{$\mathrm{P}_b^{(1)}(t)$, Sim.}
                
                \addplot [color=black, mark=triangle, mark size = 2pt, mark options={solid, black}]
                  table[row sep=crcr]{%
                2	0.344296875\\
                3	0.3147265625\\
                4	0.27434375\\
                5	0.183263888888889\\
                6	0.092265625\\
                7	0.0255922379032258\\
                8	0.003259375\\
                };
                \addlegendentry{$\mathrm{P}_b^{(2)}(t)$, Sim.}
                
                \addplot [color=black, mark=square, mark size = 1.5pt, mark options={solid, black}, forget plot]
                  table[row sep=crcr]{%
                2	0.183645833333333\\
                4	0.14849609375\\
                6	0.115828125\\
                8	0.0898307291666667\\
                10	0.0616232638888889\\
                12	0.0364797794117647\\
                14	0.016958984375\\
                16	0.00522313399280576\\
                18	0.000932291666666667\\
                20	6.45833333333333e-05\\
                };
                
                \addplot [color=black, mark=triangle, mark size = 2pt, mark options={solid, black}, forget plot]
                  table[row sep=crcr]{%
                2	0.357083333333333\\
                4	0.341484375\\
                6	0.2755625\\
                8	0.195872395833333\\
                10	0.119670138888889\\
                12	0.0598621323529412\\
                14	0.023041015625\\
                16	0.00603080035971223\\
                18	0.000974479166666667\\
                20	6.5625e-05\\
                };
                
                \addplot [color=blue, dashed]
                  table[row sep=crcr]{%
                2	0.453155607616922\\
                3	0.380948465900968\\
                4	0.299231474998998\\
                5	0.185986630043677\\
                6	0.107571600112698\\
                7	0.0260906521309386\\
                8	0.00197034292882732\\
                9	3.319092386450512e-05\\
                10	1.061309087844331e-07\\
                };
                \addlegendentry{$\mathrm{P}_b^{(1)}(t)$, DE}   
                
                \addplot [color=red, dashed]
                  table[row sep=crcr]{%
                2	0.223900020557001\\
                3	0.135296082363817\\
                4	0.0590635789314279\\
                5	0.00547011106515375\\
                6	1.14582546563279e-04\\
                7	9.96124944935932e-08\\
                };
                \addlegendentry{$\mathrm{P}_b^{(2)}(t)$, DE} 
                
                \addplot [color=blue, dashed,forget plot]
                  table[row sep=crcr]{%
                2	0.182368022874744\\
                4	0.148207477036597\\
                6	0.118065211855569\\
                8	0.0899502758700625\\
                10	0.0644108772000239\\
                12	0.0368704942145307\\
                14	0.016702653686712\\
                16	0.00519594652238148\\
                18	0.000938837332391173\\
                20	6.65658333784582e-05\\
                21  7.62561895574655e-06\\
                };
                
                \addplot [color=red, dashed,forget plot]
                  table[row sep=crcr]{%
                2	0.363684011011373\\
                4	0.316319178391191\\
                6	0.271710304512348\\
                8	0.21002845082578\\
                10	0.125365624965752\\
                12	0.064582779071211\\
                14	0.0239188981312445\\
                16	0.0065135594803752\\
                18	0.000781285575554525\\
                20	3.96263470966726e-05\\
                21  5.04561895574655e-06\\
                };

                \draw[dashed, color=black] 
                (axis cs:7.5,0.0001) ellipse[ x radius = 25, y radius = 0.75];    
                \node[] at (axis cs: 3.75,0.00015) {\scriptsize $t=5$};
                
                \draw[dashed, color=black] 
                (axis cs:20,0.0001) ellipse[ x radius = 25, y radius = 0.75];    
                \node[] at (axis cs: 16.25,0.00015) {\scriptsize $t=2$};
          
                \end{axis}
                \end{tikzpicture}%

                \vspace{-0.11em}
                \caption{BER at different SNRs when iteration $t=2$ and $t=5$}     
                \vspace{-0.11em}
                \label{Fig::BER::SNR}

             \end{subfigure}
             \hspace{-0.3em}
             \begin{subfigure}[b]{0.48\columnwidth}
                \centering
                \begin{tikzpicture}
                
                \begin{axis}[%
                width=2.5in,
                height=1.5in,
                at={(0.642in,0.505in)},
                scale only axis,
                xmin=1,
                xmax=7,
                xlabel style={at={(0.5,1ex)},font=\color{white!15!black},font=\scriptsize},
                xlabel={Iteration $t$},
                ymode=log,
                ymin=1e-16,
                ymax=1,
                ylabel style={at={(1ex,0.5)},font=\color{white!15!black},font=\scriptsize},
                ylabel={Bit error rate},
                axis background/.style={fill=white},
                tick label style={font=\tiny},
                xmajorgrids,
                ymajorgrids,
                minor grid style={dotted},
                major grid style={dotted,black},
                legend style={at={(0,0)}, anchor=south west, legend cell align=left, align=left, draw=white!15!black,font = \tiny,row sep=-2pt,legend columns=1}
                ]

                \addplot [color=black, mark=square, mark size = 1.5pt, mark options={solid, black}]
                  table[row sep=crcr]{%
                1	0.0899\\
                2	0.0892\\
                3	9.19e-05\\
                4	0.00010815\\
                };
                \addlegendentry{$\mathrm{P}_b^{(1)}(t)$, Sim.}
                
                \addplot [color=black, mark=triangle, mark size = 2pt, mark options={solid, black}]
                  table[row sep=crcr]{%
                1	0.4115\\
                2	0.1938\\
                3	0.2471\\
                4	0.003511\\
                5	0.003125\\
                6	0.0027\\
                };
                \addlegendentry{$\mathrm{P}_b^{(2)}(t)$, Sim.}

                \addplot [blue,dashed]
                  table[row sep=crcr]{%
                  1	0.0899\\
                2	0.0891\\
                3	0.000108265888802189\\
                4	7.41623882112298e-05\\
                5	3.08363491155045e-15\\
                6	2.59405646123362e-15\\
                7	2.13778147520072e-15\\
                };
                \addlegendentry{$\mathrm{P}_b^{(1)}(t)$, DE}
                
                \addplot [red,dashed]
                  table[row sep=crcr]{%
                1	0.4106\\
                2	0.2022\\
                3	0.259550556202123\\
                4	0.002683222239123\\
                5	0.002670342928827\\
                6	0.002697790603487\\
                7	0.002601180484814\\
                };
                \addlegendentry{$\mathrm{P}_b^{(2)}(t)$, DE}
                
                \draw[dashed, color=black] 
                (axis cs:3,0.005) ellipse[ x radius = 25, y radius = 7];    
                \node[] at (axis cs: 3,0.000001) {\tiny DS is turned on};

                \end{axis}
                \end{tikzpicture}%
                \vspace{-0.11em}
                 \caption{BER at different iterations when SNR = 8dB}
                 \vspace{-0.11em}
                \label{Fig::BER::iter}
             \end{subfigure}
             \vspace{-0.11em}
             \caption{The BER performance of the two-user system with $h^{(1)} = 1.225$, $h^{(2)} = 0.707$.}
             \vspace{-0.11em}
             \label{Fig::BER}
        \end{figure}
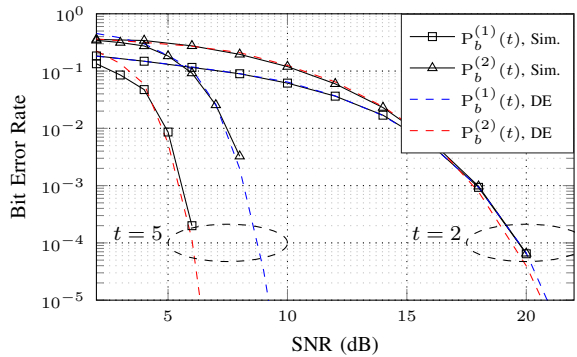
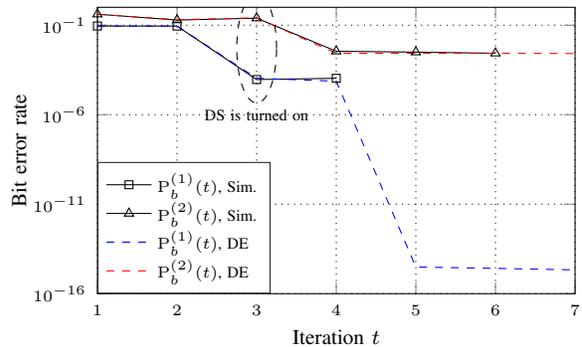

    \vspace{-0.5em}
    \subsection{Conditions of Convergence}
    \vspace{-0.5em}
    In regard to the condition of convergence, we have the following proposition.
    \begin{proposition} \label{Pro::Convergence}
    The decoding results of the JD receiver converge at iteration $t^{*}$ in the DS-on phase, when the following condition is satisfied.
    \begin{equation}\small
        \bm{\mathcal{L}}_{t^{*}}(\ell;\mathbf{x}) = \left(\bm{\Gamma} \circ \Theta \circ \Delta\right)\left(\bm{\mathcal{L}}_{t^{*}}(\ell;\mathbf{x})\right),
    \end{equation}
    or equivalently,
    \begin{equation}\small\label{equ::Convergence}
        \dot{\bm{\mathcal{L}}}_{t^{*}}(\ell) = \left(\bm{\Gamma} \circ \Theta \circ \Delta\right)\left(\dot{\bm{\mathcal{L}}}_{t^{*}}(\ell)\right).
    \end{equation}
    \end{proposition}
    \begin{IEEEproof}
        Proposition \ref{Pro::Convergence} is directly obtained from (\ref{equ::DE::densityRecur}).
    \end{IEEEproof}
     Eq. (\ref{equ::Convergence}) contains $n_u$ specific conditions for $n_u$ users, i.e., $\dot{\mathcal{L}}_{t^{*}}^{(u)}(\ell) = \left(\Gamma_u \circ \Theta \circ \Delta\right)\left(\dot{\bm{\mathcal{L}}}_{t^{*}}(\ell)\right)$ for $1\leq u \leq n_u$. According to (\ref{equ::DE::PIC::Densityl::deparamu}), $\Gamma_u$ involves multiple-integral operations, which complicates further analyses. However, we can simplify (\ref{equ::Convergence}) by taking the following assumption.
     \begin{assumption} \label{Assum::LLRGaussian}
         $\dot{\mathcal{L}}_{t^{*}}^{(u)}$ is a Gaussian density, when the JD is converged for user $u$. Also, its mean and variance, respectively, denoted by $\mathbb{E}[\dot{\mathcal{L}}_{t^{*}}^{(u)}]$ and $\mathbb{V}[\dot{\mathcal{L}}_{t^{*}}^{(u)}]$, satisfy $2\mathbb{E}[\dot{\mathcal{L}}_{t^{*}}^{(u)}] = \mathbb{V}[\dot{\mathcal{L}}_{t^{*}}^{(u)}]$.
     \end{assumption}
     
     Assumption \ref{Assum::LLRGaussian} is obtained from \cite[Eq. (8)]{ten2001convergence}, which is a common assumption for the EXIT-chart analysis. It has been argued in \cite{ten2001convergence} that if an $L$-value (i.e., LLR) follows a Gaussian distribution, then its mean and variance, denoted by $\mu$ and $\sigma^2$, satisfies $2\mu = \sigma^2$. This is valid for a symmetric Gaussian distribution. Also, the priori information input to the decoder tends to be Gaussian as the number of iterations increases in the iterative decoding of concatenated codes \cite{ten2001convergence}. In the proposed JD receiver, we take Assumption \ref{Assum::LLRGaussian} because $\dot{\mathcal{L}}_{t^{*}}^{(u)}$ is the priori LLR input to the SOSD decoder. To validate Assumption \ref{Assum::LLRGaussian}, we compare the density $\dot{\mathcal{L}}_{t^{*}}^{(u)}$ obtained from simulation and its Gaussian approximation $\mathcal{N}\left(\mathbb{E}[\dot{\mathcal{L}}_{t^{*}}^{(u)}], 2\mathbb{E}[\dot{\mathcal{L}}_{t^{*}}^{(u)}]\right)$, where $\mathbb{E}[\dot{\mathcal{L}}_{t^{*}}^{(u)}]$ is obtained from DE. The two-user and three-user NOMA systems described in Example \ref{exam::64DEDS-2} and Example \ref{exam::64DEDS-3} are considered. As shown in Fig. \ref{Fig::Gaussian-L}, the Gaussian distribution $\mathcal{N}\left(\mathbb{E}[\dot{\mathcal{L}}_{t^{*}}^{(u)}], 2\mathbb{E}[\dot{\mathcal{L}}_{t^{*}}^{(u)}]\right)$ well approximates the corresponding density $\dot{\mathcal{L}}_{t^{*}}^{(u)}$.

       \begin{figure} 
             \centering
             \hspace{-0.81em}
             \begin{subfigure}[b]{0.48\columnwidth}
                 \centering
                 
                \begin{tikzpicture}
                
                \begin{axis}[%
                width=2.5in,
                height=1.5in,
                at={(0.785in,0.587in)},
                scale only axis,
                xmin=-10,
                xmax=30,
                xlabel style={at={(0.5,1ex)},font=\color{white!15!black},font=\scriptsize},
                xlabel={Value of LLR},
                ymin= 0,
                ymax=0.2,
                yminorticks=true,
                ylabel style={at={(2ex,0.5)},font=\color{white!15!black},font=\scriptsize},
                ylabel={Densities of LLR},
                yticklabel style={
                        /pgf/number format/fixed,
                        /pgf/number format/precision=5
                },
                axis background/.style={fill=white},
                tick label style={font=\tiny},
                xmajorgrids,
                ymajorgrids,
                yminorgrids,
                minor grid style={dotted},
                major grid style={dotted,black},
                legend style={at={(1,1)}, anchor=north east, legend cell align=left, align=left, draw=white!15!black,font = \tiny,row sep=-1pt,legend columns=1}
                ]
                
                \addplot [color=black, mark=triangle, mark size = 1.3pt ,mark options={solid, black}]
                  table[row sep=crcr]{%
                -7	0.000129826773990133\\
                -6	0.000231833524982381\\
                -5	0.000454393708965466\\
                -4	0.000927334099929523\\
                -3	0.00171556808486962\\
                -2	0.00322712266775474\\
                -1	0.00593493823954894\\
                0	0.00985756148225083\\
                1	0.0157924997217998\\
                2	0.0236284728662042\\
                3	0.0329760005934938\\
                4	0.0450498905745762\\
                5	0.0571794206016544\\
                6	0.0683445231648058\\
                7	0.0785359249230313\\
                8	0.087605252420342\\
                9	0.0904428947661264\\
                10	0.0883656663822842\\
                11	0.083460068993657\\
                12	0.0752160688452836\\
                13	0.0621406580362773\\
                14	0.0511424756111132\\
                15	0.0396528061129864\\
                16	0.0278849363848807\\
                17	0.0187877888645721\\
                18	0.0128713973070218\\
                19	0.0077617864164101\\
                20	0.00494269075262436\\
                21	0.00275418227679068\\
                22	0.00137245446789569\\
                23	0.000686227233947847\\
                24	0.000306020252976742\\
                25	0.000166920137987314\\
                26	0.000120553432990838\\
                };
                \addlegendentry{$\dot{\mathcal{L}}_{5}^{(2)}$, Sim.}

                \addplot [color=black, mark=square, mark size = 1pt, mark options={solid, black}]
                  table[row sep=crcr]{%
                -13.5	9.25360427886662e-06\\
                -11.5	9.25360427886662e-06\\
                -10.5	1.85072085577332e-05\\
                -8.5	1.85072085577332e-05\\
                -8	9.25360427886662e-06\\
                -7.5	4.62680213943331e-05\\
                -7	9.25360427886662e-06\\
                -6.5	0.000129550459904133\\
                -6	0.000231340106971665\\
                -5.5	0.000453426609664464\\
                -5	0.000851331593655729\\
                -4.5	0.00153609831029186\\
                -4	0.00279458849221772\\
                -3.5	0.00495993189347251\\
                -3	0.00768974515573816\\
                -2.5	0.0130845964503174\\
                -2	0.0200340532637462\\
                -1.5	0.0287972165158329\\
                -1	0.0414468935650436\\
                -0.5	0.0549386486036311\\
                -0	0.07339958913997\\
                0.5	0.0918420224677512\\
                1	0.110025354875724\\
                1.5	0.12695945070605\\
                2	0.143680713637962\\
                2.5	0.153878185553273\\
                3	0.158930653489534\\
                3.5	0.157320526345011\\
                4	0.150195251050284\\
                4.5	0.136259323006311\\
                5	0.120685507004978\\
                5.5	0.102002479965947\\
                6	0.0838006403494161\\
                6.5	0.0658764088612515\\
                7	0.048220531897174\\
                7.5	0.0340902781633446\\
                8	0.0239390742694279\\
                8.5	0.0162955971350841\\
                9	0.0103732903966095\\
                9.5	0.00625543649251383\\
                10	0.00408083948698018\\
                10.5	0.00223937223548572\\
                11	0.00136027982899339\\
                11.5	0.00057372346528973\\
                12	0.000323876149760332\\
                12.5	0.000148057668461866\\
                13	9.25360427886662e-05\\
                13.5	3.70144171154665e-05\\
                14	1.85072085577332e-05\\
                14.5	9.25360427886662e-06\\
                18.5	1.85072085577332e-05\\
                };
                \addlegendentry{$\dot{\mathcal{L}}_{5}^{(1)}$, Sim.}

                \addplot [color=blue, dashed]
                  table[row sep=crcr]{%
                -7	7.19845923142243e-05\\
                -6	0.000167703372900479\\
                -5	0.000370461665638507\\
                -4	0.0007759683047448\\
                -3	0.00154114648270061\\
                -2	0.0029023052416623\\
                -1	0.0051825259169744\\
                0	0.00877483802822392\\
                1	0.0140875660256296\\
                2	0.0214452962468632\\
                3	0.030954754582324\\
                4	0.0423664339234212\\
                5	0.0549813861238133\\
                6	0.0676563693939023\\
                7	0.0789406908277327\\
                8	0.087335816796083\\
                9	0.0916184794185648\\
                10	0.0911324408828828\\
                11	0.0859532208344139\\
                12	0.0768688772031566\\
                12.9	0.0664261257055916\\
                14	0.0524113071530747\\
                15	0.0399586735053687\\
                16	0.028886595747618\\
                17	0.0198007159887278\\
                18	0.0128695883294634\\
                19	0.00793135967877642\\
                20	0.0046347879049412\\
                21	0.00256809635516165\\
                22	0.00134924865454339\\
                23	0.00067215882120068\\
                24	0.000317505345265029\\
                25	0.000142209761930012\\
                25.5	9.32944491197508e-05\\
                };
                \addlegendentry{$\mathcal{N}(\mathbb{E}[\mathcal{L}_{5}^{(2)}],2\mathbb{E}[\mathcal{L}_{5}^{(2)}])$}
                
                \addplot [color=red, dashed]
                  table[row sep=crcr]{%
                -7	4.84370715681166e-05\\
                -6.5	0.000105848782329167\\
                -6	0.000222310484997524\\
                -5.5	0.000448745511812118\\
                -5	0.000870575381346997\\
                -4.5	0.00162322505497079\\
                -4	0.00290882226195783\\
                -3.5	0.00500981541916414\\
                -3	0.00829263130502395\\
                -2.5	0.0131925600773589\\
                -2	0.0201712074693592\\
                -1.5	0.0296415427809798\\
                -1	0.0418635259298012\\
                -0.5	0.0568246664819093\\
                0	0.0741317204076301\\
                0.5	0.0929474233133047\\
                1	0.11200482257559\\
                1.5	0.129718578855785\\
                2	0.144388875359889\\
                2.5	0.15446545830602\\
                3	0.158816314278123\\
                3.5	0.156936853334479\\
                4	0.149046183574264\\
                4.5	0.136045089811996\\
                5	0.119346853949918\\
                5.5	0.100624832579976\\
                6	0.0815390145320488\\
                6.5	0.0635026483684303\\
                7	0.0475318052859451\\
                7.5	0.034193446146824\\
                8	0.0236410917437004\\
                8.5	0.0157093521726327\\
                9	0.010032637378184\\
                9.5	0.00615797677688803\\
                10	0.00363267939197295\\
                10.5	0.00205959669199079\\
                11	0.00112228560644222\\
                11.5	0.000587747367262455\\
                12	0.000295831274006641\\
                12.5	0.000143107886286586\\
                13	6.65348439419229e-05\\
                };
                \addlegendentry{$\mathcal{N}(\mathbb{E}[\mathcal{L}_{5}^{(1)}],2\mathbb{E}[\mathcal{L}_{5}^{(1)}])$}

                \end{axis}
                \end{tikzpicture}%

                \vspace{-0.11em}
                \caption{Two-user system}     
                \vspace{-0.11em}
                \label{Fig::Gaussian-L::2}

             \end{subfigure}
             \hspace{-0.3em}
             \begin{subfigure}[b]{0.48\columnwidth}
                \centering
                \begin{tikzpicture}
                
                \begin{axis}[%
                width=2.5in,
                height=1.5in,
                at={(0.642in,0.505in)},
                scale only axis,
                xmin=-10,
                xmax=60,
                xlabel style={at={(0.5,1ex)},font=\color{white!15!black},font=\scriptsize},
                xlabel={Value of LLR},
                ymin=0,
                ymax=0.2,
                ylabel style={at={(2ex,0.5)},font=\color{white!15!black},font=\scriptsize},
                ylabel={Densities of LLR },
                yticklabel style={
                        /pgf/number format/fixed,
                        /pgf/number format/precision=5
                },
                axis background/.style={fill=white},
                tick label style={font=\tiny},
                xmajorgrids,
                ymajorgrids,
                minor grid style={dotted},
                major grid style={dotted,black},
                legend style={at={(1,1)}, anchor=north east, legend cell align=left, align=left, draw=white!15!black,font = \tiny,row sep=-2pt,legend columns=2}
                ]
                
                \addplot [color=black, mark size=1.3pt, mark=diamond, mark options={solid, black}]
                  table[row sep=crcr]{%
                -5.5	0.000249896917521522\\
                -5	0.000374845376282284\\
                -4.5	0.000974597978333937\\
                -4	0.00159934027213774\\
                -3.5	0.00352354653705347\\
                -3.25	0.00439818574837879\\
                -3	0.00777179413491935\\
                -2.75	0.011170392213212\\
                -2.5	0.0149938150512913\\
                -2.25	0.0184174028213362\\
                -2	0.0239901040820661\\
                -1.75	0.0302625167118564\\
                -1.5	0.0373096097859633\\
                -1.25	0.0443567028600702\\
                -1	0.0553251158896955\\
                -0.75	0.0631489510576887\\
                -0.5	0.0738445391276099\\
                -0.25	0.0867392200717204\\
                -0	0.0996838803993353\\
                0.25	0.110679344770282\\
                0.5	0.119000912123749\\
                0.75	0.136393737583247\\
                1	0.14908850099334\\
                1.25	0.165181862481726\\
                1.5	0.168255594567241\\
                2	0.175977409318656\\
                2.5	0.181275223970112\\
                3	0.173878275211475\\
                3.5	0.160883635500356\\
                3.75	0.147913985480989\\
                4	0.142691139904789\\
                4.25	0.129946397111192\\
                4.5	0.116901778016568\\
                4.75	0.103957117688953\\
                5	0.090612622293304\\
                5.25	0.0812914672697512\\
                5.5	0.0703709719740607\\
                5.75	0.0576012394887109\\
                6	0.0509789711743906\\
                6.25	0.040208414029213\\
                6.5	0.0326865168118151\\
                6.75	0.0254894855871953\\
                7	0.0226156710356978\\
                7.25	0.0170429697749678\\
                7.5	0.0130945984781278\\
                7.75	0.00977096947509152\\
                8	0.00604750540402084\\
                8.5	0.00309872177726688\\
                9	0.00112453612884685\\
                9.5	0.000649731985555958\\
                10	0.000199917534017218\\
                };
                \addlegendentry{$\dot{\mathcal{L}}_{8}^{(3)}$, Sim.}
                
                \addplot [color=purple, dashed]
                  table[row sep=crcr]{%
                -7	1.39794568193492e-05\\
                -6	9.7712874515632e-05\\
                -5.6	0.000199854356131897\\
                -5.3	0.000333926936235393\\
                -5	0.000546894300369681\\
                -4.7	0.000877949330582414\\
                -4.4	0.00138149595344997\\
                -4.1	0.00213080590989033\\
                -3.8	0.0032214565530446\\
                -3.5	0.00477391601825686\\
                -3.2	0.00693443978451194\\
                -2.9	0.00987329465732453\\
                -2.6	0.0137792927002982\\
                -2.3	0.0188497616567823\\
                -2	0.0252754508604404\\
                -1.5	0.0394201870745485\\
                -1.1	0.0540458955372157\\
                -1	0.0581582030613198\\
                -0.5	0.0811663108130578\\
                0	0.107155178199201\\
                0.5	0.133820623101746\\
                0.600000000000001	0.138972774548576\\
                0.800000000000001	0.148883919160399\\
                1	0.158090386557417\\
                1.2	0.166380618690768\\
                1.4	0.173555992317015\\
                1.6	0.179438692377562\\
                1.8	0.183879021108471\\
                2	0.18676172433356\\
                2.2	0.188010961994517\\
                2.4	0.18759362373767\\
                2.6	0.185520787224426\\
                2.8	0.181847229508207\\
                3	0.176669021562686\\
                3.2	0.170119353413576\\
                3.4	0.162362843106662\\
                3.6	0.15358866883537\\
                3.8	0.1440029237148\\
                4	0.133820623101746\\
                4.5	0.107155178199201\\
                5	0.0811663108130578\\
                5.5	0.0581582030613198\\
                6	0.0394201870745485\\
                6.5	0.0252754508604404\\
                7	0.0153303358849172\\
                7.4	0.0098732946573245\\
                7.6	0.00781855902450397\\
                7.8	0.00613664421149348\\
                8	0.00477391601825686\\
                8.2	0.00368093551637\\
                8.4	0.00281307459579001\\
                8.6	0.00213080590989032\\
                8.8	0.00159972798104448\\
                9	0.00119038650366966\\
                9.2	0.000877949330582409\\
                9.4	0.000641786406767278\\
                9.6	0.000464998103168477\\
                9.8	0.000333926936235392\\
                10	0.000237679283183395\\
                11	3.80001034437649e-05\\
                };
                \addlegendentry{$\mathcal{N}(\mathbb{E}[\mathcal{L}_{8}^{(3)}],2\mathbb{E}[\mathcal{L}_{8}^{(3)}])$}
                
                \addplot [color=black, mark=square,mark size = 1, mark options={solid, black}]
                  table[row sep=crcr]{%
                -7	0.000175056893490384\\
                -6	0.000225073148773351\\
                -5	0.000687723510140796\\
                -4	0.00127541450971566\\
                -3	0.00260084527471428\\
                -2	0.00656463350588941\\
                -1	0.0139045189686648\\
                -0	0.0233825993447871\\
                1	0.0344111836346813\\
                2	0.0489033936029209\\
                3	0.0691349688648811\\
                4	0.084164203366094\\
                5	0.0974431940380624\\
                6	0.104071323180034\\
                7	0.101808087628479\\
                8	0.0959061695050892\\
                9	0.0826643659189236\\
                10	0.0706729687148323\\
                11	0.0543676694925851\\
                12	0.0396753945032135\\
                13	0.0278465501287919\\
                14	0.0183184534973867\\
                15	0.0108910395878661\\
                16	0.00636456848475755\\
                17	0.00201315427513942\\
                18	0.000900292595093405\\
                19	0.000462650361367444\\
                20	0.000150048765848901\\
                };
                \addlegendentry{$\dot{\mathcal{L}}_{8}^{(2)}$, Sim.}

                \addplot [color=red, dashed]
                  table[row sep=crcr]{%
                -7	9.72265050459144e-05\\
                -6	0.000255016710462716\\
                -5	0.000622775572478731\\
                -4	0.00141603329659742\\
                -3	0.00299774251242083\\
                -2	0.00590873011701515\\
                -1	0.0108435864569714\\
                0	0.0185280922618741\\
                1	0.02947592402131\\
                2	0.043659938308066\\
                3	0.060211267919438\\
                4	0.0773127983795536\\
                5	0.0924280901234842\\
                5.5	0.0983891897848094\\
                6	0.10288108382246\\
                6.5	0.105674066163494\\
                7	0.106621809311462\\
                7.5	0.105674066163494\\
                8	0.10288108382246\\
                8.5	0.0983891897848094\\
                9	0.0924280901234842\\
                9.5	0.0852914141987776\\
                10	0.0773127983795536\\
                10.5	0.0688402140670209\\
                11	0.060211267919438\\
                12	0.043659938308066\\
                13	0.02947592402131\\
                14	0.0185280922618741\\
                15	0.0108435864569714\\
                16	0.00590873011701515\\
                17	0.00299774251242083\\
                18	0.00141603329659742\\
                19	0.000622775572478731\\
                20	0.000255016710462716\\
                21	9.72265050459144e-05\\
                };
                \addlegendentry{$\mathcal{N}(\mathbb{E}[\mathcal{L}_{8}^{(2)}],2\mathbb{E}[\mathcal{L}_{8}^{(2)}])$}
                
                \addplot [color=black, mark=triangle, mark size = 1.3pt, mark options={solid, black}]
                  table[row sep=crcr]{%
                -9	1.87575030012005e-05\\
                -8	6.25250100040016e-06\\
                -7	1.25050020008003e-05\\
                -6	6.25250100040016e-05\\
                -5	6.25250100040016e-05\\
                -4	0.000137555022008804\\
                -3	0.000181322529011605\\
                -2	0.000418917567026811\\
                -1	0.000568977591036415\\
                0	0.000900360144057623\\
                1	0.00150685274109644\\
                2	0.00221963785514206\\
                3	0.00332633053221289\\
                4	0.00458308323329332\\
                5	0.00666516606642657\\
                6	0.00870348139255702\\
                7	0.0107292917166867\\
                8	0.0135491696678671\\
                9	0.0171506102440976\\
                10	0.0203706482593037\\
                11	0.0243472388955582\\
                12	0.0287677571028411\\
                13	0.0323942076830732\\
                14	0.0372899159663866\\
                15	0.0419917967186875\\
                16	0.0455369647859144\\
                17	0.04703131252501\\
                18	0.0509766406562625\\
                19	0.0534026110444178\\
                20	0.0543342336934774\\
                21	0.0548656962785114\\
                22	0.0540591236494598\\
                23	0.0517269407763105\\
                24	0.0497011304521809\\
                25	0.0461184473789516\\
                26	0.0408725990396158\\
                27	0.0369585334133653\\
                28	0.0317001800720288\\
                29	0.0273984593837535\\
                30	0.0226965786314526\\
                31	0.0190263605442177\\
                32	0.0148809523809524\\
                33	0.0119297719087635\\
                34	0.00884728891556623\\
                35	0.00626500600240096\\
                36	0.00485819327731092\\
                37	0.00343887555022009\\
                38	0.00241971788715486\\
                39	0.00166316526610644\\
                40	0.00125050020008003\\
                41	0.000812825130052021\\
                42	0.000550220088035214\\
                43	0.000287615046018407\\
                44	0.000206332533013205\\
                45	8.75350140056022e-05\\
                46	6.25250100040016e-05\\
                47	5.00200080032013e-05\\
                48	1.25050020008003e-05\\
                50	6.25250100040016e-06\\
                };
                \addlegendentry{$\dot{\mathcal{L}}_{8}^{(1)}$, Sim.}
                
                \addplot [color=blue, dashed]
                  table[row sep=crcr]{%
                -10	6.6183501586185e-07\\
                -9	1.36813910822865e-06\\
                -8	2.76166139603842e-06\\
                -7	5.44340009062114e-06\\
                -6	1.04768261741369e-05\\
                -5	1.969014325212e-05\\
                -4	3.61349674932137e-05\\
                -3	6.47539277178504e-05\\
                -2	0.000113308951049586\\
                -1	0.000193607410931204\\
                0	0.000323027421407587\\
                1	0.000526279506989294\\
                2	0.000837246195816377\\
                3	0.00130061740609831\\
                4	0.00197290237663727\\
                5	0.00292227636322941\\
                6	0.00422665334306516\\
                7	0.00596941303396761\\
                8	0.00823239659602556\\
                9	0.0110861460185749\\
                10	0.0145778863386864\\
                11	0.0187183765804583\\
                12	0.0234693713054937\\
                13	0.0287338874819706\\
                14	0.0343516002194535\\
                15	0.0401013703908639\\
                16	0.0457120934138618\\
                17	0.0508818229169133\\
                18	0.0553036593754151\\
                19	0.0586954897934029\\
                20	0.0608296402766786\\
                21	0.0615581303074573\\
                22	0.0608296402766786\\
                23	0.0586954897934029\\
                24	0.0553036593754151\\
                25	0.0508818229169133\\
                26	0.0457120934138618\\
                27	0.0401013703908639\\
                28	0.0343516002194535\\
                29	0.0287338874819706\\
                30	0.0234693713054937\\
                31	0.0187183765804583\\
                32	0.0145778863386864\\
                33	0.0110861460185749\\
                34	0.00823239659602556\\
                35	0.00596941303396761\\
                36	0.00422665334306516\\
                37	0.00292227636322941\\
                38	0.00197290237663727\\
                39	0.00130061740609831\\
                40	0.000837246195816377\\
                41	0.000526279506989294\\
                42	0.000323027421407587\\
                43	0.000193607410931204\\
                44	0.000113308951049586\\
                45	6.47539277178504e-05\\
                46	3.61349674932137e-05\\
                47	1.969014325212e-05\\
                48	1.04768261741369e-05\\
                49	5.44340009062114e-06\\
                50	2.76166139603842e-06\\
                };
                \addlegendentry{$\mathcal{N}(\mathbb{E}[\mathcal{L}_{8}^{(1)}],2\mathbb{E}[\mathcal{L}_{8}^{(1)}])$}

                \end{axis}
                \end{tikzpicture}%
                \vspace{-0.11em}
                 \caption{Three-user system}
                 \vspace{-0.11em}
                \label{Fig::Gaussian-L::3}
             \end{subfigure}
             \vspace{-0.11em}
             \caption{Gaussian approximation of the converged density $\dot{\mathcal{L}}_{t^{*}}^{(u)}$.}
             \vspace{-0.31em}
             \label{Fig::Gaussian-L}
        \end{figure}
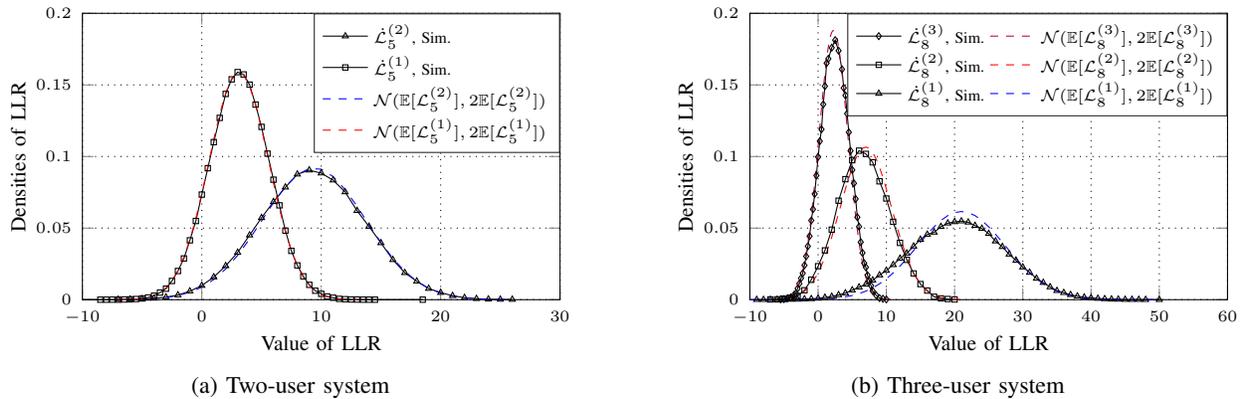
        
        Under Assumption \ref{Assum::LLRGaussian}, we can conclude that $\dot{\mathcal{L}}_{t^{*}}^{(u)}$ is the density of the LLR of an observation $\bar{x}^{(u)}(t^{*}) = x^{(u)} + w^{(u)}(t^{*})$ of symbol $x^{(u)}$, where $w^{(u)}(t^{*})$ is the interference-noise variable following a Gaussian distribution $\mathcal{N}(0,\sigma_u^2(t^{*}))$, and $\sigma_u^2(t^{*})$ is the power sum of the noise and interference. Therefore, $\dot{\mathcal{L}}_{t^{*}}^{(u)}$ follows $\mathcal{N}(2/\sigma_u^2(t^{*}),4/\sigma_u^2(t^{*}))$. With $\bar{x}^{(u)}(t^{*})$, it is assumed that the interference over the $u$-th user is Gaussian. This assumption was widely used in the analysis of NOMA and MIMO systems, when the number of users is large and deep interleavers are applied \cite{liu2019capacity,yuan2008evolution,wang2019near,yuan2014energy}; however, in this paper, we only consider Gaussian interference under Assumption \ref{Assum::LLRGaussian}, i.e., upon convergence, which holds even for small numbers of users.
        
        According to (\ref{equ::Receiver::PICLLR}), $\sigma_u^2(t^{*})$ is derived as
        \begin{equation}\small
            \sigma_u^2(t^{*}) = \frac{1}{\left(h^{(u)}\right)^2}\left(\sum_{j\neq u}\left(h^{(j)}\right)^2  \upsilon^{(j)}(t^{*}) + \sigma^2 \right),
        \end{equation}
        where 
        \begin{equation}\small
            \upsilon^{(j)}(t^{*}) = \int_{-\infty}^{+\infty} \frac{4\exp(\ell)}{(1+\exp(\ell))^2}\dot{\mathcal{D}}_{t^{*}}^{(j)}(\ell) d\ell .
        \end{equation}
        Therefore, we can rewrite condition (\ref{equ::Convergence}) as
        \begin{equation}\small\label{equ::Convergence::simple}
            \sigma_u^2(t^{*}) = \frac{1}{\left(h^{(u)}\right)^2}\left(\sum_{j\neq u}\left(h^{(j)}\right)^2 \int_{-\infty}^{+\infty} \frac{4\exp(\ell)}{(1+\exp(\ell))^2} \Delta\left(\mathcal{N}\left(\frac{2}{\sigma_j^2(t^{*})},\frac{4}{\sigma_j^2(t^{*})} \right)(\ell)\right) d\ell + \sigma^2 \right),
        \end{equation}
        for $1 \leq u \leq n_u$. Eq. (\ref{equ::Convergence::simple}) is referred to as the condition of convergence for the OSD-based JD. As shown, the converged interference-noise power $\sigma_u^2(t^{*})$ is determined by the noise power $ \sigma^2$, the receiving power $\mathbf{h}$, as well as the density-transform feature $\Delta(\cdot)$ of the decoder. Moreover, $\Delta(\cdot)$ is determined by the decoding order of SOSD and the structure of $\mathcal{C}(n,k)$, as as shown by Theorem \ref{the::Extrinsicdensity} and Remark \ref{rem::GeneralSOSD}. 
        
        Next, we provide some preliminary discussions on the convergence behavior for the two-user system and the equal-power system.
        
        \subsubsection{Two-User System}
        We define a function $g_d(\xi)$ as 
        \begin{equation}\small
            g_d(\xi) \triangleq \int_{-\infty}^{+\infty} \frac{4\exp(\ell)}{(1+\exp(\ell))^2} \Delta\left(\mathcal{N}\left(\frac{2}{\xi}, \frac{4}{\xi}\right)(\ell)\right) d\ell.
        \end{equation}
        In a two-user system, the condition (\ref{equ::Convergence::simple}) can be rewritten as
        \begin{equation}\small\label{equ::Convergence::2user}
            \begin{cases}
                   g_{(-)}\left(\xi_1, h^{(1)} ,h^{(2)}\right) = g_{(+)}\left(\xi_1, h^{(1)} ,h^{(2)}\right) \ \ \text{for user 1}, \\
                   g_{(-)}\left(\xi_2, h^{(2)} ,h^{(1)}\right) = g_{(+)}\left(\xi_2, h^{(2)} ,h^{(1)}\right)  \ \ \text{for user 2},
            \end{cases}
        \end{equation}
        where
        \begin{equation}\small
            g_{(-)}\left(x, y ,z\right) = g_d^{-1}\left(\left(\frac{y}{z}\right)^2 x-\left(\frac{1}{z}\right)^2\sigma^2\right),
        \end{equation}
        and
        \begin{equation}\small
            g_{(+)}\left(x, y ,z\right) = \left(\frac{y}{z}\right)^2 g_d(x) + \left(\frac{1}{z}\right)^2\sigma^2.
        \end{equation}
        
        Therefore, we have $\sigma_1^2(t^{*}) = \xi_1^{*}$ and $\sigma_2^2(t^{*}) = \xi_2^{*}$, where $\xi_1^{*}$ and $\xi_2^{*}$ are the solutions of (\ref{equ::Convergence::2user}). Accordingly, the converged density $\dot{\mathcal{L}}_{t^{*}}^{(1)}$ and $\dot{\mathcal{L}}_{t^{*}}^{(2)}$ are respectively given by $\mathcal{N}(2/\xi_1^{*}, 4/\xi_1^{*})$ and $\mathcal{N}(2/\xi_2^{*}, 4/\xi_2^{*})$. We validate the condition (\ref{equ::Convergence::2user}) in following numerical results.
        
        In Fig. \ref{Fig::Converge::2user::special}, we depict the convergence point of the two-user system described in Example \ref{exam::64DEDS-2}. It is found that $\sigma_1^2(t^{*}) = 0.213$ and $\sigma_1^2(t^{*}) = 0.637$. Thus, when converged, we have $\dot{\mathcal{L}}_{t^{*}}^{(1)} = \mathcal{N}(2/0.213,4/0.213)$ and $\dot{\mathcal{L}}_{t^{*}}^{(2)} = \mathcal{N}(2/0.637,4/0.637)$, which is consistent with $\dot{\mathcal{L}}_{5}^{(1)}$ and $\dot{\mathcal{L}}_{4}^{(2)}$ depicted in Fig. \ref{Fig::DE-on-2user}.
            
        Furthermore, for given $y$, $z$, $\xi_1$ and $\xi_2$, we observe $\xi_2 = g_{(-)}\left(\xi_1, y ,z\right)$ and $\xi_1 = g_{(+)}\left(\xi_2, z ,y\right)$. This indicates that the coordinate $\left(\xi_1^{*}, \xi_2^{*}\right)$ is exactly the intersection point of $g_{(-)}\left(\xi_1, h^{(1)} ,h^{(2)}\right)$ and  $g_{(+)}\left(\xi_1, h^{(1)} ,h^{(2)}\right)$. In Fig \ref{Fig::Converge::2user::trend}. we illustrate the intersection point $\left(\xi_1^{*}, \xi_2^{*}\right)$ with different values of $h^{(1)}$ and $h^{(2)}$, satisfying $(h^{(1)})^2 + (h^{(1)})^2 = 2$ at SNR = 5dB. For comparison, we also include results of the uncoded system, where $g_d(\xi) = \int_{-\infty}^{+\infty} \frac{4\exp(\ell)}{(1+\exp(\ell))^2}\mathcal{N}\left(2/\xi,4/\xi\right)(\ell) d\ell$, and the decoding-error-free system, i.e., $g_d(\xi) = 0$. The error-free decoding can be achieved by asymptotically increasing the code length to infinity ($n\rightarrow \infty$), as considered in the Shannon Capacity Theorem \cite{Shannon}. We can observe that finite block-length codes (e.g., length-32 and length-64 codes demonstrated) will have the JD performance between that of the error-free and uncoded system.

           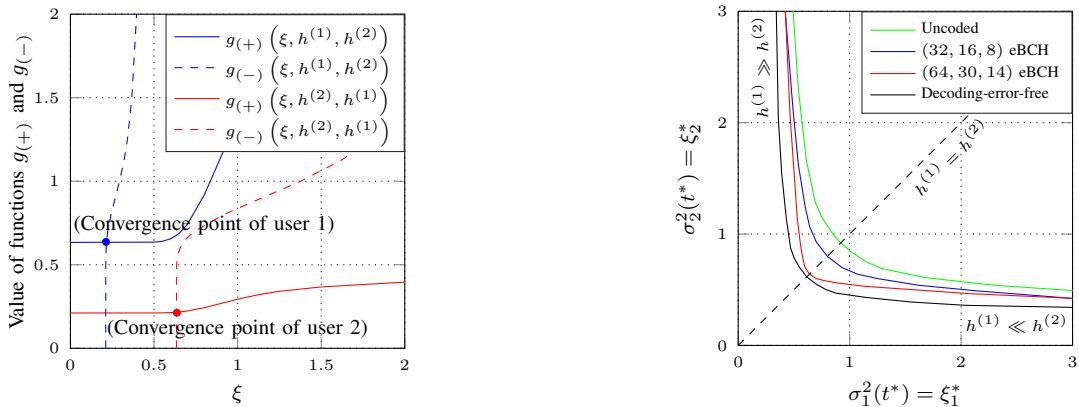
\begin{figure} 
             \centering
             \hspace{-0.81em}
             \begin{subfigure}[b]{0.48\columnwidth}
                 \centering
                 
                \begin{tikzpicture}
                
                \begin{axis}[%
                width=1.75in,
                height=1.75in,
                at={(0.785in,0.587in)},
                scale only axis,
                xmin=0,
                xmax=2,
                xlabel style={at={(0.5,1ex)},font=\color{white!15!black},font=\scriptsize},
                xlabel={$\xi$},
                ymin= 0,
                ymax=2,
                yminorticks=true,
                ylabel style={at={(3ex,0.5)},font=\color{white!15!black},font=\scriptsize},
                ylabel={Value of functions $g_{(+)}$ and $g_{(-)}$},
                axis background/.style={fill=white},
                tick label style={font=\tiny},
                xmajorgrids,
                ymajorgrids,
                yminorgrids,
                minor grid style={dotted},
                major grid style={dotted,black},
                legend style={at={(1,1)}, anchor=north east, legend cell align=left, align=left, draw=white!15!black,font = \tiny,row sep=-1pt,legend columns=1}
                ]
                
                \addplot [color=blue]
                  table[row sep=crcr]{%
                0	0.634\\
                0.5	0.6354196\\
                0.55	0.640555\\
                0.6	0.6556\\
                0.65	0.6868\\
                0.7	0.742\\
                0.8	0.9172\\
                0.875	1.0876\\
                1	1.3768\\
                1.125	1.6192\\
                1.25	1.8016\\
                1.5	2.038\\
                2	2.296\\
                3	2.5864\\
                4	2.77\\
                };
                \addlegendentry{$g_{(+)}\left(\xi, h^{(1)} ,h^{(2)}\right)$}
                
                \addplot [color=blue, dashed]
                  table[row sep=crcr]{%
                0.211333333333333	0\\
                0.211491066666667	0.5\\
                0.212061666666667	0.55\\
                0.213733333333333	0.6\\
                0.2172	0.65\\
                0.223333333333333	0.7\\
                0.2428	0.8\\
                0.261733333333333	0.875\\
                0.293866666666667	1\\
                0.3208	1.125\\
                0.341066666666667	1.25\\
                0.367333333333333	1.5\\
                0.396	2\\
                0.428266666666667	3\\
                0.448666666666667	4\\
                };
                \addlegendentry{$g_{(-)}\left(\xi, h^{(1)} ,h^{(2)}\right)$}
                
                \addplot [color=red]
                  table[row sep=crcr]{%
                0	0.211333333333333\\
                0.5	0.211491066666667\\
                0.55	0.212061666666667\\
                0.6	0.213733333333333\\
                0.65	0.2172\\
                0.7	0.223333333333333\\
                0.8	0.2428\\
                0.875	0.261733333333333\\
                1	0.293866666666667\\
                1.125	0.3208\\
                1.25	0.341066666666667\\
                1.5	0.367333333333333\\
                2	0.396\\
                3	0.428266666666667\\
                4	0.448666666666667\\
                };
                \addlegendentry{$g_{(+)}\left(\xi, h^{(2)} ,h^{(1)}\right)$}
                
                \addplot [color=red, dashed]
                  table[row sep=crcr]{%
                0.634	0\\
                0.6354196	0.5\\
                0.640555	0.55\\
                0.6556	0.6\\
                0.6868	0.65\\
                0.742	0.7\\
                0.9172	0.8\\
                1.0876	0.875\\
                1.3768	1\\
                1.6192	1.125\\
                1.8016	1.25\\
                2.038	1.5\\
                2.296	2\\
                2.5864	3\\
                2.77	4\\
                };
                \addlegendentry{$g_{(-)}\left(\xi, h^{(2)} ,h^{(1)}\right)$}
          
                \addplot [mark = *, mark size = 1.3pt, mark options={solid, red}]
                  table[row sep=crcr]{%
            	0.637   0.213 \\
                };
                
                \addplot [mark = *, mark size = 1.3pt, mark options={solid, blue}]
                  table[row sep=crcr]{%
            	0.213   0.637 \\
                };
                 \node[] at (axis cs: 0.8,0.737) {\scriptsize (Convergence point of user 1)};
                 \node[] at (axis cs: 1,0.113) {\scriptsize (Convergence point of user 2)};
          
                \end{axis}
                \end{tikzpicture}%

                \vspace{-0.11em}
                \caption{Convergence points with $h^{(1)}=1.225$ and $h^{(2)}=0.707$.}     
                \vspace{-0.11em}
                \label{Fig::Converge::2user::special}

             \end{subfigure}
             \hspace{-0.3em}
             \begin{subfigure}[b]{0.48\columnwidth}
                \centering
                \begin{tikzpicture}
                
                \begin{axis}[%
                width=1.75in,
                height=1.75in,
                at={(0.642in,0.505in)},
                scale only axis,
                xmin=0,
                xmax=3,
                xlabel style={at={(0.5,1ex)},font=\color{white!15!black},font=\scriptsize},
                xlabel={$\sigma_1^2(t^{*}) = \xi_1^{*}$},
                ymin=0,
                ymax=3,
                ylabel style={at={(3ex,0.5)},font=\color{white!15!black},font=\scriptsize},
                ylabel={$\sigma_2^2(t^{*}) = \xi_2^{*}$},
                axis background/.style={fill=white},
                tick label style={font=\tiny},
                xmajorgrids,
                ymajorgrids,
                minor grid style={dotted},
                major grid style={dotted,black},
                legend style={at={(1,1)}, anchor=north east, legend cell align=left, align=left, draw=white!15!black,font = \tiny,row sep=-2pt,legend columns=1}
                ]
                
                \addplot [color= green ]
                  table[row sep=crcr]{%
                0.37	8.67\\
                0.43	4.49\\
                0.49	3.04\\
                0.54	2.3\\
                0.61	1.67\\
                0.69	1.29\\
                0.75	1.15\\
                0.85	1\\
                0.92	0.92\\
                1	0.85\\
                1.15	0.75\\
                1.29	0.69\\
                1.67	0.61\\
                2.3	0.54\\
                3.04	0.49\\
                4.49	0.43\\
                8.67	0.37\\
                };
                \addlegendentry{Uncoded}
                
                \addplot [color=blue]
                  table[row sep=crcr]{%
                0.35	6.24\\
                0.41	3.15\\
                0.49	2.14\\
                0.54	1.62\\
                0.6	1.25\\
                0.64	1.07\\
                0.7	0.93\\
                0.8	0.8\\
                0.93	0.7\\
                1.07	0.64\\
                1.25	0.6\\
                1.62	0.54\\
                2.14	0.49\\
                3.15	0.41\\
                6.24	0.35\\
                };
                \addlegendentry{$(32,16,8)$ eBCH}

                \addplot [color=red]
                  table[row sep=crcr]{%
                0.369	6.235\\
                0.418	3.11\\
                0.465	2.07\\
                0.501	1.558\\
                0.53	1.1394\\
                0.56	0.895\\
                0.58	0.78\\
                0.6	0.7\\
                0.64	0.64\\
                0.7	0.6\\
                0.78	0.58\\
                0.895	0.56\\
                1.1394	0.53\\
                1.558	0.501\\
                2.07	0.465\\
                3.11	0.418\\
                6.235	0.369\\
                };
                \addlegendentry{$(64,30,14)$ eBCH}
                
                \addplot [color=black]
                  table[row sep=crcr]{%
                0.34	3.09\\
                0.36	2.04\\
                0.39	1.545\\
                0.432	1.137\\
                0.47	0.875\\
                0.501	0.787\\
                0.550	0.7\\
                0.615	0.615\\
                0.7	0.550\\
                0.787	0.501\\
                0.875	0.47\\
                1.137	0.432\\
                1.545	0.39\\
                2.04	0.36\\
                3.09	0.34\\
                };
                \addlegendentry{Decoding-error-free}
                
                \addplot [color=black, dashed]
                  table[row sep=crcr]{%
                0	0\\
                3   3\\
                };
                
                 \node[rotate=45] at (axis cs:1.9,1.7) {\tiny $h^{(1)} = h^{(2)}$};
                 
                 \node[] at (axis cs:2.5, 0.2) {\tiny $h^{(1)}\ll h^{(2)}$};
                 \node[rotate=90] at (axis cs:0.2 ,2.5) {\tiny $h^{(1)}\gg h^{(2)}$};

                \end{axis}
                \end{tikzpicture}%
                \vspace{-0.11em}
                 \caption{The converged coordinate $\left(\xi_1^{*}, \xi_2^{*}\right)$ for different codes}
                 \vspace{-0.11em}
                \label{Fig::Converge::2user::trend}
             \end{subfigure}
             \vspace{-0.11em}
             \caption{The Convergence point of the two-user system.}
             \vspace{-0.31em}
             \label{Fig::Converge::2user}
        \end{figure}

        \subsubsection{Equal-Power Case}
        If all users have the equal receiving power (e.g., the equal-power interleave-division multiple-access (IDMA) system \cite{ping2006interleave}), we further simplify (\ref{equ::Convergence::simple}) to 
        \begin{equation}\small \label{equ::Convergence::equal}
            \sigma_u^2(t^{*}) = (n_u-1) g_d(\sigma_u^2(t^{*})) + \sigma^2,
        \end{equation}
        by assuming $h^{(u)} = 1$, for $1\leq u \leq n_u$. Taking $\sigma^2 = n_u\cdot\mathrm{SNR}^{-1}$, we conclude that the converged interference-noise power, $\sigma_u^2(t^{*})$, is given by the intersection point of functions $g_d(\xi)$ and
        \begin{equation}\small
            g_e(\xi) = \frac{1}{n_u-1}\left(\xi  - n_u \cdot \mathrm{SNR}^{-1}\right).
        \end{equation}
        Therefore, for $\xi^{*}$ that makes $g_e(\xi^{*}) = g_d(\xi^{*})$, the converged density $\dot{\mathcal{L}}_{t^{*}}^{(u)}$ is given by $\mathcal{N}\left(2/\xi^{*},4/\xi^{*} \right)$. This will be elaborated in following numerical results.
        
        Consider an equal-power system with $h^{(u)}=1$ for $1\leq u \leq n_u$, which is decoded by the iterative JD with the order-3 SOSD. In Fig. \ref{Fig::Convergence::2user::Numerical}, we demonstrate convergence points obtained by $g_e(\xi^{*}) = g_d(\xi^{*})$ when $n_u = 2$ at different SNRs. As shown, for a specific SNR, the $(30,16,8)$ eBCH code will give a larger $\xi^{*}$ than that of the $(64,30,14)$ eBCH code; accordingly, the converged density $\dot{\mathcal{L}}_{t^{*}}^{(u)}$ will have a smaller mean value (i.e., $\frac{2}{\xi^{*}}$). This is because the $(64,30,14)$ code has a higher minimum Hamming distance and thus provides a larger coding gain.
        
        In Fig. \ref{Fig::Convergence::2user::Numerical}, we find the values of $\xi^{*}$ with applying the $(64,30,14)$ eBCH code as $\xi^{*} = 0.159$, $\xi^{*} = 0.318$, $\xi^{*} = 0.656$, and $\xi^{*} = 1.452$, for $\mathrm{SNR} = 11$ dB, $8$ dB, $5$ dB, and $3$ dB, respectively. In Fig. \ref{Fig::Convergence::2user::Simulation}, we compare the density $\dot{\mathcal{L}}_{t^{*}}^{(u)}$ obtained from simulations and that of $\mathcal{N}(2/\xi^{*},4/\xi^{*})$. As shown, the Gaussian density $\mathcal{N}(2/\xi^{*},4/\xi^{*})$ accurately describes the converged density $\dot{\mathcal{L}}_{t^{*}}^{(u)}$, especially at moderate-to-high SNRs.
        
        Furthermore, in Fig. \ref{Fig::Convergence::muser::Numerical}, we illiterate values of $\xi^{*}$ at SNR = 8 dB for different numbers of users. The corresponding densities $\mathcal{N}(2/\xi^{*},4/\xi^{*})$ with the $(64,30,14)$ eBCH code is depicted in Fig. \ref{Fig::Convergence::muser::Simulation}. As shown, as the number of user increases, the value of $\xi^{*}$ increases, indicating that the JD performance is degraded. Moreover, we notice that $g_e(\xi)$ for $n_u=3$ has multiple intersection points (multiple possible values of $\xi^{*}$) with $g_d(\xi)$ when applying the $(64,30,14)$ eBCH code. Despite this, JD will converge at the largest $\xi^{*}$, because at the beginning of the JD iterations without any interference cancellation, $\sigma_u^2(t)$ is initialized to be $\sigma_u^2(0) = (n_u-1)+\sigma^2$, which is larger than all possible values of $\xi^{*}$.
    
       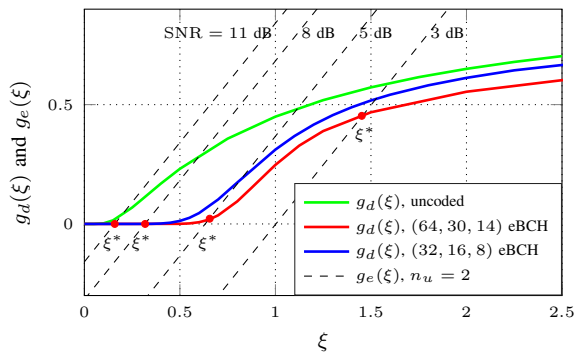
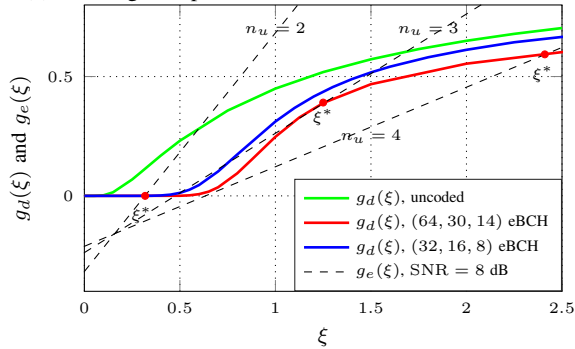
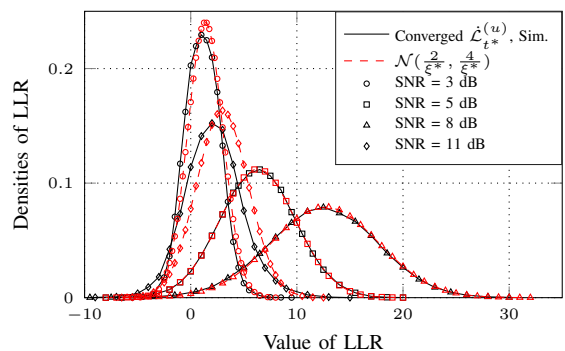
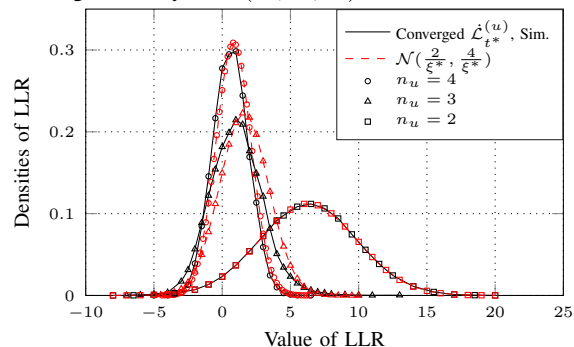
\begin{figure} 
             \centering
             
             \begin{subfigure}[b]{0.48\columnwidth}
                 \centering
                 
                \begin{tikzpicture}
                
                \begin{axis}[%
                width=2.5in,
                height=1.5in,
                at={(0.785in,0.587in)},
                scale only axis,
                xmin=0,
                xmax=2.5,
                xlabel style={at={(0.5,1ex)},font=\color{white!15!black},font=\scriptsize},
                xlabel={$\xi$},
                ymin= -0.3,
                ymax=0.9,
                yminorticks=true,
                ylabel style={at={(2ex,0.5)},font=\color{white!15!black},font=\scriptsize},
                ylabel={$g_d(\xi)$ and $g_e(\xi)$},
                axis background/.style={fill=white},
                tick label style={font=\tiny},
                xmajorgrids,
                ymajorgrids,
                yminorgrids,
                minor grid style={dotted},
                major grid style={dotted,black},
                legend style={at={(1,0)}, anchor=south east, legend cell align=left, align=left, draw=white!15!black,font = \tiny,row sep=-1pt,legend columns=1}
                ]
                
                \addplot [color=green, line width=1.0pt]
                  table[row sep=crcr]{%
                   0          0 \\
                    0.0500    1.2037e-05\\
                    0.1          0.0024\\
                    0.15            0.0150
                    0.2          0.0385\\
                    0.25          0.0686\\
                    0.3        0.1017 \\
                    0.4        0.1688\\
                    0.5         0.2310\\
                    0.75         0.3574\\
                    1             0.4496\\
                    1.25          0.5188\\
                    1.5         0.5723\\
                    1.75     0.6151  \\
                    2         0.6499\\
                    2.25     0.6789\\
                    2.5      0.7033\\
                };
                \addlegendentry{$g_d(\xi)$, uncoded}
                
                \addplot [color=red, line width=1.0pt]
                  table[row sep=crcr]{%
                    4	0.712\\
                    3	0.6508\\
                    2	0.554\\
                    1.5	0.468\\
                    1.25	0.3892\\
                    1.125	0.3284\\
                    1	0.2476\\
                    0.875	0.1512\\
                    0.8	0.0944\\
                    0.7	0.036\\
                    0.65	0.0176\\
                    0.6	0.0072\\
                    0.55	0.00218552\\
                    0.5	0.0004732\\
                    0	0\\
                };
                \addlegendentry{$g_d(\xi)$, $(64,30,14)$ eBCH}
                
                \addplot [color=blue, line width=1.0pt]
                  table[row sep=crcr]{%
                    0       0\\
                    0.25	8.90692465565192e-06\\
                    0.3	7.47036609804761e-05\\
                    0.35	0.000419807064338085\\
                    0.4	0.00168945080462062\\
                    0.45	0.00516479472899244\\
                    0.5	0.0126145455978419\\
                    0.55	0.0257061297476133\\
                    0.6	0.0453246065802833\\
                    0.7	0.102346339674643\\
                    0.8	0.173278146866535\\
                    1	0.311596994044812\\
                    1.1	0.368426654006934\\
                    1.2	0.416121073278132\\
                    1.3	0.455882823776456\\
                    1.4	0.489174700798229\\
                    1.5	0.517335442170611\\
                    1.6	0.541466361762104\\
                    1.8	0.580879982416608\\
                    2	0.612121466593376\\
                    2.25	0.643703627119796\\
                    3	0.710652628304421\\
                };
                \addlegendentry{$g_d(\xi)$, $(32,16,8)$ eBCH}

                \addplot [color=black, dashed]
                  table[row sep=crcr]{%
                0	-0.317\\
                1	0.683\\
                2	1.683\\
                3	2.683\\
                4	3.683\\
                5	4.683\\
                };
                \addlegendentry{$g_e(\xi)$, $n_u=2$}
                
                \addplot [color=black, dashed, forget plot]
                  table[row sep=crcr]{%
                0	-0.6324\\
                1	0.3676\\
                2	1.3676\\
                3	2.3676\\
                4	3.3676\\
                5	4.3676\\
                };
                
                \addplot [color=black, dashed, forget plot]
                  table[row sep=crcr]{%
                0	-0.1588\\
                1	0.8412\\
                2	1.8412\\
                3	2.8412\\
                4	3.8412\\
                5	4.8412\\
                };
                
                \addplot [color=black, dashed, forget plot]
                  table[row sep=crcr]{%
                0	-1.0024\\
                1	-0.00239999999999996\\
                2	0.9976\\
                3	1.9976\\
                4	2.9976\\
                5	3.9976\\
                };
                \node[] at (axis cs: 0.70,0.8) {\tiny $\mathrm{SNR} = 11$ dB};
                \node[] at (axis cs: 1.225,0.8) {\tiny $8$ dB};
                \node[] at (axis cs: 1.525,0.8) {\tiny $5$ dB};
                \node[] at (axis cs: 1.9,0.8) {\tiny $3$ dB};
                \node[] at (axis cs: 0.15,-0.08) {\tiny $\xi^{*}$};
                \node[] at (axis cs: 0.3,-0.08) {\tiny $\xi^{*}$};
                \node[] at (axis cs: 0.65,-0.08) {\tiny $\xi^{*}$};
                \node[] at (axis cs: 1.47,0.375) {\tiny $\xi^{*}$};

                \addplot [mark = *, mark size = 1.3pt, mark options={solid, red}]
                  table[row sep=crcr]{%
            	0.159   0\\
                };
                \addplot [mark = *, mark size = 1.3pt, mark options={solid, red}]
                  table[row sep=crcr]{%
            	0.318   0\\
                };
                \addplot [mark = *, mark size = 1.3pt, mark options={solid, red}]
                  table[row sep=crcr]{%
            	0.656   0.022 \\
                };
                \addplot [mark = *, mark size = 1.3pt, mark options={solid, red}]
                  table[row sep=crcr]{%
            	1.452   0.453 \\
                };

                \end{axis}
                \end{tikzpicture}%

                \vspace{-0.11em}
                \caption{Convergence points when $n_u=2$ at different SNR}     
                \label{Fig::Convergence::2user::Numerical}

             \end{subfigure}
             \hspace{-0.3em}
             \begin{subfigure}[b]{0.48\columnwidth}
                \centering
                \begin{tikzpicture}
                
                \begin{axis}[%
                width=2.5in,
                height=1.5in,
                at={(0.642in,0.505in)},
                scale only axis,
                xmin=-10,
                xmax=35,
                xlabel style={at={(0.5,1ex)},font=\color{white!15!black},font=\scriptsize},
                xlabel={Value of LLR},
                ymin=0,
                ymax=0.25,
                ylabel style={at={(2ex,0.5)},font=\color{white!15!black},font=\scriptsize},
                ylabel={Densities of LLR },
                axis background/.style={fill=white},
                tick label style={font=\tiny},
                xmajorgrids,
                ymajorgrids,
                minor grid style={dotted},
                major grid style={dotted,black},
                legend style={at={(1,1)}, anchor=north east, legend cell align=left, align=left, draw=white!15!black,font = \tiny,row sep=-2pt,legend columns=1}
                ]
                
                \addplot [color=black]
                  table[row sep=crcr]{%
                -10	10\\
                };
                \addlegendentry{Converged $\dot{\mathcal{L}}_{t^{*}}^{(u)}$, Sim.}
                
               \addplot [color=red, dashed]
                  table[row sep=crcr]{%
                -10	10\\
                };
                \addlegendentry{$\mathcal{N}(\frac{2}{\xi^{*}},\frac{4}{\xi^{*}})$}
                
               \addplot [only marks, mark=o, mark size = 1pt, mark options={solid, black}]
                  table[row sep=crcr]{%
                -10	10\\
                };
                \addlegendentry{SNR = 3 dB}
                
               \addplot [only marks, mark=square, mark size = 1pt, mark options={solid, black}]
                  table[row sep=crcr]{%
                -10	10\\
                };
                \addlegendentry{SNR = 5 dB}
                
               \addplot [only marks, mark=triangle, mark size = 1.3pt, mark options={solid, black}]
                  table[row sep=crcr]{%
                -10	10\\
                };
                \addlegendentry{SNR = 8 dB}
                
               \addplot [only marks, mark=diamond, mark size = 1.3pt, mark options={solid, black}]
                  table[row sep=crcr]{%
                -10	10\\
                };
                \addlegendentry{SNR = 11 dB}
                
                \addplot [color=black, mark= o, mark size = 1pt, mark options={solid, black}]
                  table[row sep=crcr]{%
                -6.5	2.50475904218014e-05\\
                -5.5	0.000112714156898106\\
                -4.5	0.000726380122232241\\
                -3.5	0.00328123434525599\\
                -2.5	0.01765855124737\\
                -2	0.0361061015930267\\
                -1.5	0.0706592525799018\\
                -1	0.116947199679391\\
                -0.5	0.162208195571586\\
                -0	0.20272267307885\\
                0.5	0.219604749023144\\
                1	0.229410880673279\\
                1.5	0.224739505059613\\
                2	0.209786093577798\\
                2.5	0.179465985372207\\
                3	0.132714657849915\\
                3.5	0.0838217613465585\\
                4	0.0487551347560365\\
                4.5	0.0262749223524697\\
                5	0.0126365093677988\\
                5.5	0.00652489730487927\\
                6.5	0.0018034265103697\\
                7.5	0.000375713856327021\\
                8.5	0.000162809337741709\\
                9.5	2.50475904218014e-05\\
                };
                
                \addplot [color=black, mark= diamond, mark size = 1.3pt, mark options={solid, black}]
                  table[row sep=crcr]{%
                -9	0.000221660681828257\\
                -7	0.000731480250033249\\
                -5	0.00374606552289755\\
                -4	0.00919891829587268\\
                -3	0.0181318437735514\\
                -2	0.042226359888283\\
                -1	0.0733222502992419\\
                -0	0.113996985414727\\
                1	0.141087023983686\\
                2	0.152369552688744\\
                3	0.146539876756661\\
                4	0.116020304118455\\
                5	0.0759389989803609\\
                6	0.0444872988429312\\
                7	0.0260229640466374\\
                8	0.0148512656824932\\
                9	0.00797978454581726\\
                11	0.002416101431928\\
                13	0.000531985636387818\\
                15	6.64982045484772e-05\\
                };
                
                \addplot [color=black, mark= square, mark size = 1pt, mark options={solid, black}]
                  table[row sep=crcr]{%
                -8	3.12836299021448e-05\\
                -6.5	0.000131391245589008\\
                -5	0.000744550391671046\\
                -3.5	0.00270290562354531\\
                -2	0.00695747929023701\\
                -1	0.013389393598118\\
                -0	0.0229371574442526\\
                1	0.0369522236404134\\
                2	0.0541519633606127\\
                3	0.0726656155367019\\
                4	0.091642265435343\\
                4.5	0.0977801136221438\\
                5.5	0.10852291213054\\
                6.5	0.11187026053007\\
                7.5	0.106383111845234\\
                8.5	0.0940949020196711\\
                9.5	0.0746677678504392\\
                10.5	0.0557912255674851\\
                11.5	0.0385101484095403\\
                12.5	0.0244137447756338\\
                14	0.011187026053007\\
                15.5	0.00427334384463298\\
                17	0.00116375103235979\\
                18.5	0.000337863202943164\\
                20	8.75941637260055e-05\\
                };
                
                \addplot [color=black, mark= triangle, mark size = 1.3pt, mark options={solid, black}]
                  table[row sep=crcr]{%
                -9.5	9.99518981490158e-05\\
                -7.5	0.000187409809029405\\
                -5.5	0.000243632751738226\\
                -3.5	0.000649687337968603\\
                -1.5	0.00216145979747247\\
                0.5	0.00530369759553215\\
                2	0.00940797241327611\\
                3.5	0.0171417505325562\\
                5	0.0263435721559\\
                6.5	0.0384377518319309\\
                8	0.0527621082354117\\
                9.5	0.0654060233512622\\
                11	0.0753137552552834\\
                12.5	0.0788620476395735\\
                14	0.0738269707703168\\
                15.5	0.0665054942309014\\
                17	0.0540489889240803\\
                18.5	0.038818818443624\\
                20	0.0263873011113402\\
                21.5	0.0159985506974768\\
                23	0.00896443586523985\\
                24.5	0.00475396215571256\\
                26	0.00217395378474109\\
                28	0.000718404267946051\\
                30	0.000237385758103912\\
                };

                \addplot [color=red, dashed , mark= triangle, mark size = 1.3pt, mark options={solid, red}]
                  table[row sep=crcr]{%
                -6	8.18827019538721e-05\\
                -5	0.000167760565788849\\
                -4	0.000330394620322295\\
                -3	0.000625491511578361\\
                -2	0.00113829608296343\\
                -1	0.00199128962223193\\
                0	0.00334856695676193\\
                1	0.00541288639531214\\
                2	0.00841093361420983\\
                3	0.0125633328509476\\
                4	0.0180389350505004\\
                5	0.0248978755415586\\
                6	0.0330338396571758\\
                7	0.0421309439005178\\
                8	0.0516521840482948\\
                9	0.0608725588891182\\
                10	0.0689604036833106\\
                11	0.07509714035493\\
                12	0.0786126383852428\\
                13	0.0791055059820048\\
                14	0.0765184952981636\\
                15	0.0711494418787325\\
                16	0.0635948492362225\\
                17	0.0546408887242093\\
                18	0.0451293350764341\\
                19	0.0358298906791733\\
                20	0.0273449667325275\\
                21	0.0200610972654861\\
                22	0.0141474228988827\\
                23	0.00959059067894197\\
                24	0.00624969386041238\\
                25	0.00391487139738406\\
                26	0.00235733673625816\\
                27	0.00136449238481274\\
                28	0.000759217091227645\\
                29	0.000406074954353638\\
                30	0.000208781397414439\\
                31	0.000103186472001345\\
                32	4.90229109477914e-05\\
                };

                \addplot [color=red, dashed , mark= square, mark size = 1pt, mark options={solid, red}]
                  table[row sep=crcr]{%
                -8	3.36398382019608e-05\\
                -7	0.000100492794197494\\
                -6	0.000277330802192349\\
                -5	0.000707039140253133\\
                -4	0.00166521750147777\\
                -3	0.00362310265736995\\
                -2	0.0072823655412947\\
                -1	0.0135221725114033\\
                0	0.0231954408997449\\
                1	0.0367570758190361\\
                2	0.053809807517546\\
                3	0.0727719622011034\\
                4	0.0909177949835009\\
                5	0.104933912414137\\
                6	0.111883230926832\\
                7	0.110203730300099\\
                8	0.100278944370926\\
                9	0.0842956908131069\\
                10	0.0654610786159555\\
                11	0.0469616234774885\\
                12	0.0311232692377362\\
                13	0.0190550230964501\\
                14	0.0107774473927786\\
                15	0.00563124608090209\\
                16	0.00271816153741573\\
                17	0.001212071338931\\
                18	0.00049930195121559\\
                19	0.000190011771480367\\
                20	6.68005307186608e-05\\
                };

                \addplot [color=red, dashed , mark= diamond, mark size = 1.3pt, mark options={solid, red}]
                  table[row sep=crcr]{%
                -6.5	8.69259442226473e-05\\
                -5.5	0.000392115539636646\\
                -4.5	0.00149563559156078\\
                -3.5	0.00482374820306277\\
                -2.5	0.0131549914079281\\
                -2	0.020399696256476\\
                -1.5	0.0303349726971966\\
                -1	0.0432563918696095\\
                -0.5	0.0591485018042762\\
                0	0.0775575321557951\\
                0.5	0.0975193930547551\\
                1	0.117583063983863\\
                1.5	0.13595191558603\\
                2	0.150734500040247\\
                2.5	0.160260603380191\\
                3	0.163390821904117\\
                3.5	0.159740601879354\\
                4	0.149757903168251\\
                4.5	0.13463282542326\\
                5	0.116064377005547\\
                5.5	0.0959475090645184\\
                6	0.0760598105234261\\
                6.5	0.0578180644941747\\
                7	0.0421462204484271\\
                7.5	0.0294605258750148\\
                8	0.0197473639460084\\
                8.5	0.0126930076277953\\
                9.5	0.00462419018894249\\
                10.5	0.00142447206069377\\
                11.5	0.000371038763261108\\
                12.5	8.17206349982153e-05\\
                };
                
                \addplot [color=red, dashed , mark= o, mark size = 1pt, mark options={solid, red}]
                  table[row sep=crcr]{%
                -5.5	4.436117835912e-05\\
                -4.75	0.00026151449853076\\
                -4	0.00125622280734717\\
                -3.5	0.00319187796294999\\
                -3	0.00740465937026985\\
                -2.5	0.015683503801662\\
                -2.25	0.0220593036401654\\
                -2	0.0303291545629631\\
                -1.75	0.0407613548174029\\
                -1.5	0.0535496599481061\\
                -1.25	0.0687677198561104\\
                -1	0.0863241543145567\\
                -0.75	0.10592533368393\\
                -0.5	0.127053651079671\\
                -0.25	0.148968436315802\\
                0	0.170734453620995\\
                0.25	0.191279258513718\\
                0.5	0.209476057055165\\
                0.75	0.224243933168583\\
                1	0.234653380030354\\
                1.25	0.240022926700484\\
                1.5	0.239992925709747\\
                1.75	0.234565401509846\\
                2	0.224103824498873\\
                2.25	0.209292845671912\\
                2.5	0.191064190346418\\
                2.75	0.170499855070731\\
                3	0.148726559185183\\
                3.25	0.126815648680873\\
                3.5	0.105700481340079\\
                3.75	0.0861193777165254\\
                4	0.0685874413105984\\
                4.25	0.0533959257745223\\
                4.5	0.0406341744064903\\
                4.75	0.0302269662061642\\
                5	0.0219794834260152\\
                5.5	0.0108548224502108\\
                6	0.00489448688754843\\
                6.5	0.0020149803564643\\
                7.25	0.000448760674039618\\
                8	8.14398511993284e-05\\
                };

                \end{axis}
                \end{tikzpicture}%
                \vspace{-0.11em}
                 \caption{Converged density with $(64,30,14)$ eBCH code when $n_u=2$}
                \label{Fig::Convergence::2user::Simulation}
             \end{subfigure}

             \begin{subfigure}[b]{0.48\columnwidth}
                 \centering
                 
                \begin{tikzpicture}
                
                \begin{axis}[%
                width=2.5in,
                height=1.5in,
                at={(0.785in,0.587in)},
                scale only axis,
                xmin=0,
                xmax=2.5,
                xlabel style={at={(0.5,1ex)},font=\color{white!15!black},font=\scriptsize},
                xlabel={$\xi$},
                ymin= -0.4,
                ymax=0.8,
                yminorticks=true,
                ylabel style={at={(2ex,0.5)},font=\color{white!15!black},font=\scriptsize},
                ylabel={$g_d(\xi)$ and $g_e(\xi)$},
                axis background/.style={fill=white},
                tick label style={font=\tiny},
                xmajorgrids,
                ymajorgrids,
                yminorgrids,
                minor grid style={dotted},
                major grid style={dotted,black},
                legend style={at={(1,0)}, anchor=south east, legend cell align=left, align=left, draw=white!15!black,font = \tiny,row sep=-1pt,legend columns=1}
                ]
                
                \addplot [color=green, line width=1.0pt]
                  table[row sep=crcr]{%
                   0          0 \\
                    0.0500    1.2037e-05\\
                    0.1          0.0024\\
                    0.15            0.0150
                    0.2          0.0385\\
                    0.25          0.0686\\
                    0.3        0.1017 \\
                    0.4        0.1688\\
                    0.5         0.2310\\
                    0.75         0.3574\\
                    1             0.4496\\
                    1.25          0.5188\\
                    1.5         0.5723\\
                    1.75     0.6151  \\
                    2         0.6499\\
                    2.25     0.6789\\
                    2.5      0.7033\\
                };
                \addlegendentry{$g_d(\xi)$, uncoded}
                
                \addplot [color=red, line width=1.0pt]
                  table[row sep=crcr]{%
                    4	0.712\\
                    3	0.6508\\
                    2	0.554\\
                    1.5	0.468\\
                    1.25	0.3892\\
                    1.125	0.3284\\
                    1	0.2476\\
                    0.875	0.1512\\
                    0.8	0.0944\\
                    0.7	0.036\\
                    0.65	0.0176\\
                    0.6	0.0072\\
                    0.55	0.00218552\\
                    0.5	0.0004732\\
                    0	0\\
                };
                \addlegendentry{$g_d(\xi)$, $(64,30,14)$ eBCH}
                
                \addplot [color=blue, line width=1.0pt]
                  table[row sep=crcr]{%
                    0       0\\
                    0.25	8.90692465565192e-06\\
                    0.3	7.47036609804761e-05\\
                    0.35	0.000419807064338085\\
                    0.4	0.00168945080462062\\
                    0.45	0.00516479472899244\\
                    0.5	0.0126145455978419\\
                    0.55	0.0257061297476133\\
                    0.6	0.0453246065802833\\
                    0.7	0.102346339674643\\
                    0.8	0.173278146866535\\
                    1	0.311596994044812\\
                    1.1	0.368426654006934\\
                    1.2	0.416121073278132\\
                    1.3	0.455882823776456\\
                    1.4	0.489174700798229\\
                    1.5	0.517335442170611\\
                    1.6	0.541466361762104\\
                    1.8	0.580879982416608\\
                    2	0.612121466593376\\
                    2.25	0.643703627119796\\
                    3	0.710652628304421\\
                };
                \addlegendentry{$g_d(\xi)$, $(32,16,8)$ eBCH}

                \addplot [color=black, dashed]
                  table[row sep=crcr]{%
                0	-0.316978638492223\\
                1	0.683021361507777\\
                2	1.68302136150778\\
                3	2.68302136150778\\
                };
                \addlegendentry{$g_e(\xi)$, $\mathrm{SNR} = 8$ dB}
                
                \addplot [color=black, dashed]
                  table[row sep=crcr]{%
                0	-0.237733978869167\\
                1	0.262266021130833\\
                2	0.762266021130833\\
                3	1.26226602113083\\
                };
                
                \addplot [color=black, dashed]
                  table[row sep=crcr]{%
                0	-0.211319092328148\\
                1	0.122014241005185\\
                2	0.455347574338518\\
                3	0.788680907671852\\
                4	1.12201424100518\\
                5	1.45534757433852\\
                };

                \node[] at (axis cs: 1,0.7) {\tiny $n_u = 2$};
                \node[] at (axis cs: 1.8,0.7) {\tiny $n_u = 3$};
                \node[] at (axis cs: 1.5,0.25) {\tiny $n_u = 4$};
                \node[] at (axis cs: 0.3,-0.08) {\tiny $\xi^{*}$};
                \node[] at (axis cs: 1.25,0.33) {\tiny $\xi^{*}$};
                \node[] at (axis cs: 2.4,0.51) {\tiny $\xi^{*}$};

                \addplot [mark = *, mark size = 1.3pt, mark options={solid, red}]
                  table[row sep=crcr]{%
            	0.318   0\\
                };
                \addplot [mark = *, mark size = 1.3pt, mark options={solid, red}]
                  table[row sep=crcr]{%
            	1.250   0.391 \\
                };
                \addplot [mark = *, mark size = 1.3pt, mark options={solid, red}]
                  table[row sep=crcr]{%
            	2.41   0.593 \\
                };

                \end{axis}
                \end{tikzpicture}%

                \vspace{-0.11em}
                \caption{Convergence points at SNR = 8 dB for various user numbers}     
                \label{Fig::Convergence::muser::Numerical}

             \end{subfigure}             
             \hspace{-0.3em}
             \begin{subfigure}[b]{0.48\columnwidth}
                \centering
                \begin{tikzpicture}
                
                \begin{axis}[%
                width=2.5in,
                height=1.5in,
                at={(0.642in,0.505in)},
                scale only axis,
                xmin=-10,
                xmax=25,
                xlabel style={at={(0.5,1ex)},font=\color{white!15!black},font=\scriptsize},
                xlabel={Value of LLR},
                ymin=0,
                ymax=0.35,
                ylabel style={at={(2ex,0.5)},font=\color{white!15!black},font=\scriptsize},
                ylabel={Densities of LLR },
                axis background/.style={fill=white},
                tick label style={font=\tiny},
                xmajorgrids,
                ymajorgrids,
                minor grid style={dotted},
                major grid style={dotted,black},
                legend style={at={(1,1)}, anchor=north east, legend cell align=left, align=left, draw=white!15!black,font = \tiny,row sep=-2pt,legend columns=1}
                ]
                
                \addplot [color=black]
                  table[row sep=crcr]{%
                -10	10\\
                };
                \addlegendentry{Converged $\dot{\mathcal{L}}_{t^{*}}^{(u)}$, Sim.}
                
               \addplot [color=red, dashed]
                  table[row sep=crcr]{%
                -10	10\\
                };
                \addlegendentry{$\mathcal{N}(\frac{2}{\xi^{*}},\frac{4}{\xi^{*}})$}
                
               \addplot [only marks, mark=o, mark size = 1pt, mark options={solid, black}]
                  table[row sep=crcr]{%
                -10	10\\
                };
                \addlegendentry{$n_u=4$}
                
               \addplot [only marks, mark=triangle, mark size = 1.3pt, mark options={solid, black}]
                  table[row sep=crcr]{%
                -10	10\\
                };
                \addlegendentry{$n_u=3$}
                
               \addplot [only marks, mark=square, mark size = 1pt, mark options={solid, black}]
                  table[row sep=crcr]{%
                -10	10\\
                };
                \addlegendentry{$n_u=2$}

                \addplot [color=black, mark= o, mark size = 1pt, mark options={solid, black}]
                  table[row sep=crcr]{%
                -5	4.70639913402256e-05\\
                -4	0.000219631959587719\\
                -3.5	0.00122366377484587\\
                -3	0.00365530332742419\\
                -2.5	0.0129112216243352\\
                -2	0.0391415527979543\\
                -1.5	0.0846994964152927\\
                -1	0.145553237218204\\
                -0.5	0.216588488147718\\
                -0	0.277520668936197\\
                0.5	0.294761777763833\\
                1	0.298322953108577\\
                1.5	0.244309179047111\\
                2	0.169320552845018\\
                2.5	0.113502659115511\\
                3	0.0591594371146636\\
                3.5	0.0245830914767112\\
                4	0.00999325416124123\\
                4.5	0.00288659146886717\\
                5	0.00108247180082519\\
                5.5	0.000251007953814536\\
                6.5	7.84399855670427e-05\\
                };
                
                \addplot [color=red, dashed ,mark= o, mark size = 1pt, mark options={solid, red}]
                  table[row sep=crcr]{%
                -4.4	8.34982671997996e-05\\
                -4.2	0.000154597641346079\\
                -4	0.000279450696900982\\
                -3.8	0.000493156160538933\\
                -3.6	0.000849651172146646\\
                -3.4	0.00142913676414973\\
                -3.2	0.00234684201394329\\
                -3	0.00376245097677348\\
                -2.8	0.00588890811228664\\
                -2.6	0.00899861431570859\\
                -2.4	0.0134243550807224\\
                -2.2	0.0195518634800434\\
                -2	0.0278009613368557\\
                -1.8	0.0385929879468674\\
                -1.6	0.052303882346431\\
                -1.4	0.0692048260335527\\
                -1.2	0.0893955228154343\\
                -1	0.112738457491467\\
                -0.8	0.13880507357812\\
                -0.6	0.166845894766946\\
                -0.399999999999999	0.195795471289367\\
                -0.199999999999999	0.2243193263041\\
                0	0.250904051362369\\
                0.2	0.273984243856459\\
                0.4	0.29209251975579\\
                0.600000000000001	0.304013041622829\\
                0.800000000000001	0.308916372330514\\
                1	0.306454900367558\\
                1.2	0.296803588116405\\
                1.4	0.280639408236754\\
                1.6	0.25906282518362\\
                1.8	0.233473971676886\\
                2	0.205422862922229\\
                2.2	0.176455831846485\\
                2.4	0.147979033851165\\
                2.6	0.121154985824716\\
                2.8	0.0968410137740017\\
                3	0.0755708480156673\\
                3.2	0.0575739725749087\\
                3.4	0.0428227969823957\\
                3.6	0.0310957327935688\\
                3.8	0.0220446645162587\\
                4	0.0152574910762666\\
                4.2	0.0103095494889509\\
                4.4	0.00680100639544329\\
                4.6	0.00438009581868486\\
                4.8	0.00275404461586276\\
                5	0.00169057819993546\\
                5.2	0.00101315639424481\\
                5.4	0.000592781508225217\\
                5.6	0.000338602170380215\\
                5.8	0.000188825985545157\\
                6	0.000102804166210903\\
                6.2	5.4643260153556e-05\\
                6.4	2.83556379401574e-05\\
                };
                
                \addplot [color=black, mark= triangle, mark size = 1.3pt, mark options={solid, black}]
                  table[row sep=crcr]{%
                -7	8.5404389785635e-05\\
                -6	0.00068323511828508\\
                -5	0.00153727901614143\\
                -4	0.00358698437099667\\
                -3	0.0171662823469126\\
                -2.5	0.0312580066615424\\
                -2	0.056623110427876\\
                -1.5	0.0888205653770604\\
                -1	0.122128277393458\\
                -0.5	0.154154923563071\\
                0	0.18156973268426\\
                0.5	0.198976855410368\\
                1	0.214442736356649\\
                1.5	0.208981381843027\\
                2	0.176530873686908\\
                2.5	0.149116064565719\\
                3	0.121018020326245\\
                3.5	0.0817320010248527\\
                4	0.0508156119224528\\
                4.5	0.0342471603040396\\
                5	0.0231445896319071\\
                5.5	0.0148603638227005\\
                6	0.00956529165599112\\
                6.5	0.00649073362370826\\
                7	0.00427021948928175\\
                7.5	0.00315996242206849\\
                8	0.0017080877957127\\
                9	0.000939448287641985\\
                10	0.000427021948928175\\
                11	0.00017080877957127\\
                13	0.00034161755914254\\
                };
                
                \addplot [color=red, dashed, mark= triangle, mark size = 1.3pt, mark options={solid, red}]
                  table[row sep=crcr]{%
                -5	0.000246841546803984\\
                -4.5	0.000665763245336339\\
                -4	0.00166070355497138\\
                -3.5	0.00383120295319876\\
                -3	0.00817426938459663\\
                -2.5	0.0161299669131699\\
                -2	0.0294366748518869\\
                -1.5	0.0496837958584363\\
                -1	0.0775553085762858\\
                -0.5	0.111964162327786\\
                0	0.149491769959102\\
                0.5	0.184597695985942\\
                1	0.210817186003452\\
                1.5	0.222667324871662\\
                2	0.217509241846573\\
                2.5	0.196503224506039\\
                3	0.164184589544398\\
                3.5	0.126872017340654\\
                4	0.0906713418549037\\
                4.5	0.0599300974705028\\
                5	0.0366345337589458\\
                5.5	0.0207112874868425\\
                6	0.0108291474944576\\
                6.5	0.00523663355676915\\
                7	0.00234196738951433\\
                7.5	0.000968679806132724\\
                8	0.000370553050179831\\
                8.5	0.000131096555928559\\
                9	4.28946289955358e-05\\
                9.5	1.29803168272456e-05\\
                10	3.63277427643464e-06\\
                };
                
                \addplot [color=black, mark= square, mark size = 1pt, mark options={solid, black}]
                  table[row sep=crcr]{%
                -8	3.12836299021448e-05\\
                -6.5	0.000131391245589008\\
                -5	0.000744550391671046\\
                -3.5	0.00270290562354531\\
                -2	0.00695747929023701\\
                -1	0.013389393598118\\
                -0	0.0229371574442526\\
                1	0.0369522236404134\\
                2	0.0541519633606127\\
                3	0.0726656155367019\\
                4	0.091642265435343\\
                4.5	0.0977801136221438\\
                5.5	0.10852291213054\\
                6.5	0.11187026053007\\
                7.5	0.106383111845234\\
                8.5	0.0940949020196711\\
                9.5	0.0746677678504392\\
                10.5	0.0557912255674851\\
                11.5	0.0385101484095403\\
                12.5	0.0244137447756338\\
                14	0.011187026053007\\
                15.5	0.00427334384463298\\
                17	0.00116375103235979\\
                18.5	0.000337863202943164\\
                20	8.75941637260055e-05\\
                };

                \addplot [color=red, dashed , mark= square, mark size = 1pt, mark options={solid, red}]
                  table[row sep=crcr]{%
                -8	3.36398382019608e-05\\
                -7	0.000100492794197494\\
                -6	0.000277330802192349\\
                -5	0.000707039140253133\\
                -4	0.00166521750147777\\
                -3	0.00362310265736995\\
                -2	0.0072823655412947\\
                -1	0.0135221725114033\\
                0	0.0231954408997449\\
                1	0.0367570758190361\\
                2	0.053809807517546\\
                3	0.0727719622011034\\
                4	0.0909177949835009\\
                5	0.104933912414137\\
                6	0.111883230926832\\
                7	0.110203730300099\\
                8	0.100278944370926\\
                9	0.0842956908131069\\
                10	0.0654610786159555\\
                11	0.0469616234774885\\
                12	0.0311232692377362\\
                13	0.0190550230964501\\
                14	0.0107774473927786\\
                15	0.00563124608090209\\
                16	0.00271816153741573\\
                17	0.001212071338931\\
                18	0.00049930195121559\\
                19	0.000190011771480367\\
                20	6.68005307186608e-05\\
                };

                \end{axis}
                \end{tikzpicture}%
                \vspace{-0.11em}
                 \caption{Converged density with $(64,30,14)$ eBCH code at SNR = 8 dB}
                \label{Fig::Convergence::muser::Simulation}
             \end{subfigure}

             \vspace{-0.11em}
             \caption{The convergence point of the multi-user transmission with equal receiving power.}
             \vspace{-0.31em}
             \label{Fig::Convergence}
        \end{figure}
        
        \vspace{-0.5em}
       \subsection{Future Works based on the Analytical Framework}
       \vspace{-0.5em}
        Based on the analytical framework introduced in this paper, one can further analyze, improve, and optimize the OSD-based JD. We outline a few potential directions as follows.
        
        \subsubsection{Improving the BER Performance} As shown by Fig. \ref{Fig::Convergence::muser::Numerical}, there may exist multiple convergence points in certain circumstances when using finite block-length codes, where JD may converge to a local suboptimal convergence point. Therefore, one can investigate the convergence behavior of JD utilizing the proposed DE framework, and devise the JD strategies that avoid suboptimal convergence points.      
        
        \subsubsection{Speeding up the Convergence} One can devise techniques such as adaptive DS, to reduce the number of iterations to achieve convergence. This can be done by investigating the density-transform feature of SOSD; that is, DS is turned on only when the extrinsic density $\mathcal{D}(\ell)$ exhibits a higher quality than the priori density $\mathcal{L}(\ell)$. This ensures that the decoding results is always improved towards the convergence point. Moreover, combining techniques, e.g., DSC and DC, can be thoroughly studied by examining their impact on the LLR density during the JD process.
        
        \subsubsection{Reducing the Complexity per Iteration}
        To reduce the complexity per iteration is to reduce the complexity of single SOSD decoding. A potential approach is to select relatively low decoding orders at some early JD iterations, when the low-order decoding can provide sufficient improvement over the priori LLR. This can be done by examining the density transform feature of SOSD in conjunction with the JD convergence behavior. However, low-order decoding will inevitably increase the total number of JD iterations. Therefore, one can optimize the overall complexity to answer the following question: Is it better to reduce complexity per iteration or to reduce the number of iterations?

\section{Conclusion} \label{sec::Conclusion}

In this paper, we introduced an analytical framework for the joint decoding (JD) of NOMA systems for short-packet communications, based on ordered-statistics decoding (OSD).  We first introduced  a  variant OSD algorithm, namely Dual-OSD. It was shown that by carefully selecting the parameters in Dual-OSD, it can approach the density-transform feature of SOSD or its variants. The density-transform feature of SOSD, i.e., computing the density of extrinsic LLR based on the priori LLR, was derived by analyzing Dual-OSD. Next, we developed the density evolution (DE) framework for the OSD-based JD by representing it as bipartite graphs (BGs). It was shown that the  proposed  DE  can  accurately describe the evolution of the priori and extrinsic LLRs during the JD process. Finally, we discussed the BER performance and the convergence point of the OSD-based JD, and analyzed the converged densities of LLR in the two-user system and equal-power system. The proposed DE framework can be used to further optimize the OSD-based NOMA JD in future works.







%

\begin{spacing}{1.3}
\bibliographystyle{IEEEtran}
\bibliography{BibAbrv/IEEEabrv, BibAbrv/OSDAbrv, BibAbrv/LPAbrv, BibAbrv/SurveyAbrv, BibAbrv/ClassicAbrv, BibAbrv/MLAbrv, BibAbrv/MathAbrv,BibAbrv/NOMAAbrv}
\end{spacing}

\end{document}